\pdfoutput=1
\documentclass[hyperfootnotes=false,letterpaper]{article}
\usepackage[english]{babel}
\usepackage{cite}
\usepackage[utf8]{inputenc}
\usepackage{fullpage}
\usepackage{amsmath, amsfonts, amssymb, amsthm}
\usepackage{lmodern}

\usepackage{xcolor}
\definecolor{darkgreen}{rgb}{0,0.5,0}
\definecolor{darkblue}{rgb}{0,0,0.6}

\usepackage[linktocpage=true,backref=page,
pagebackref=true,colorlinks, linkcolor=darkblue,citecolor=darkgreen,urlcolor=blue,filecolor=magenta,bookmarksnumbered] {hyperref}

\usepackage[capitalize, nameinlink]{cleveref}

\usepackage{thm-restate}
\usepackage{tikz}
\usetikzlibrary{arrows.meta, bending, calc, backgrounds}

\usepackage[inline]{enumitem}

\usepackage{caption}
\usepackage{subcaption}
\usepackage{xspace}

\usepackage{mathtools} 

\input{defs}

\title{Range Queries for Permutation Pattern Matching:\\Separations, Equivalences, and Applications}
\author{Omri Ben-Eliezer\thanks{Faculty of Computer Science, Technion--Israel Institute of Technology, Haifa, Israel. \texttt{omribene@cs.technion.ac.il}}
    \and Tomer Even\thanks{Faculty of Computer Science, Technion--Israel Institute of Technology, Haifa, Israel. \texttt{tomer.even@campus.technion.ac.il}}}
\date{}

\begin{document}
\maketitle

\begin{abstract}
    The problem of Permutation Pattern Matching (PPM) asks, given a permutation $\pi:[n]\to[n]$, viewed as the point set $P=\{(i,\pi(i)):i\in[n]\}$, and a pattern $\sigma$ of small length $k$, whether $P$ contains an order-isomorphic copy of $\sigma$. The PPM problem and its (exact and approximate) counting variants have been studied extensively in the literature.
    We consider a \emph{preprocessing-query} version: after preprocessing $P$,
    one must determine whether a query $P'\subseteq P$ contains $\sigma$,
    where $P'$ is obtained by intersecting $P$ with a geometric shape
    (e.g., an axis-aligned rectangle or $k$ such rectangles).
    This naturally generalizes the classical range query problem in
    computational geometry, which corresponds to the case $|\sigma|=1$.
    Length-$2$ patterns have also been considered: here it is known that detection and approximate counting can be done efficiently [Ben-Eliezer, Mitrovi\'c, and Srivastava, SODA 2026], with $\tilde{O}(n)$ preprocessing and $\tilde{O}(1)$ query time.
    In contrast, exact counting of length-$2$ patterns is harder [Duraj, Kleiner, Polak, Vassilevska Williams, SODA 2020]: even in an easier offline batch setting, it requires super-linear time, assuming 3SUM requires $n^{3/2+\Omega(1)}$ expected time.

    In this work we extend the investigation to patterns of length $3$ (and in some cases beyond), establishing an even more nuanced fine-grained picture. We prove several new fine-grained equivalences, separations, and applications. A few of the main results are highlighted below.

    \begin{enumerate}
        \item \textbf{Near-optimal detection for single rectangle, via rectangle-to-slab reductions:} We first consider the case where the queried structure is a single axis-parallel rectangle.
        Here, for all patterns of length $3$, detection can be done with $\tilde{O}(n)$ preprocessing and $\tilde{O}(1)$ query time. While this result is similar to the case $k=2$, the proof differs considerably, and goes through a new and somewhat surprising tool we develop, and which may be of independent interest: reductions from \emph{rectangle queries} to a simpler and less expressive family of \emph{slab queries}.
        \item \textbf{Hardness of multi-rectangle queries via equivalence to online set disjointness:}
              We next establish hardness of queries consisting of multiple rectangles: for $\sigma=132$, we show that this problem is equivalent to the \emph{online set disjointness} problem. We further show connections to online $k$-set disjointness.
              To our knowledge, this is the first equivalence between a PPM-type problem and online set disjointness -- a problem playing a fundamental role in fine-grained complexity.

        \item \textbf{Faster algorithms for exact and approximate counting:} As an application of our new data structures and techniques, we obtain improved exact and approximate algorithms for pattern counting. These include
              an $O(n^{1.725})$ time algorithm for exact counting of $5$-patterns based on the rectangle-to-slab reduction, improving upon the best known algorithm [Beniamini and Lavee, ICALP 2025]; and an $\tilde{O}(n^{7/4})$ approximate counting algorithm for $6$-patterns, which is the first to break the $O(n^2)$ barrier for $6$-patterns, and which makes use of techniques from binary CSPs, approximate counting to colorful detection oracles, and additional tools.
    \end{enumerate}
\end{abstract}

\pagebreak
\tableofcontents

\clearpage
\section{Introduction}

In the Permutation Pattern Matching (PPM) problem, we are given a large permutation $\pi:[n]\to[n]$ and a small permutation $\sigma:[k]\to[k]$, and ask whether there exist indices $1\leq t_1<\cdots<t_k\leq n$ such that $(\pi(t_1),\ldots,\pi(t_k))$ has the same relative order as $\sigma$, i.e., for all $i,j \in [k]$, $\pi(t_i) < \pi(t_j)$ if and only if $\sigma(i) < \sigma(j)$.
We refer to $(t_1,\ldots,t_k)$ as a \emph{copy} or an \emph{embedding} of $\sigma$ in $\pi$.
Many variants of PPM including detection, exact and approximate counting, over various models have been extensively studied in recent years  \cite{BBL98,AAAH01,AR08,BL12,GM14, BenEliezerC18,NewmanRRS19,BenEliezerCLW19,BKM21,JOP21,EL21,gawrychowski2022faster,BenEliezerLW22,NewmanV24,Zhang24,BL2024counting,ben2024approximate,Opler26}, both due to the combinatorial elegance of this problem, but also because of many fundamental applications and beautiful connections to a range of other (seemingly unrelated) areas. Let us mention a few representative examples. In nonparametric statistics, efficient counting versions of PPM have led to the development of fast independence tests \cite{EL21,evenzohar2020independence}. In parameterized complexity, the popular emerging notion of twin-width has been invented as  a breakthrough paper showing the fixed parameter tractability of the detection version of PPM \cite{GM14}. In graph algorithms, the problem of counting $4$-cycles in sparse graphs is known to be equivalent to counting permutation patterns of length $4$ \cite{dudek2020counting}. In property testing, the first natural property shown to have an adaptivity hierarchy was the property of $\sigma$-freeness where $\sigma=132$ \cite{BenEliezerC18}. These and other examples demonstrate that PPM and its many variants underlie a rich and highly interconnected theory, waiting to be explored further.

In this paper we study a preprocessing-query version of PPM.
Such preprocess-once, answer-many-queries models are the backbone of \emph{fine-grained complexity} and have continuously gained popularity in recent years~\cite{KasliwalPS25,GolovnevGHPV20,AESZ23,GajulapalliGKS26,ChungL23}.
One of the most prominent problems in this area is \emph{Online Set Disjointness} (\prba). Here, we are provided with a family of sets $S_1,\ldots,S_m$ that we are allowed to preprocess. Then given a query pair $(i,j)$ the task is to decide whether $S_i\cap S_j=\emptyset$.
Starting with P\u{a}tra\c{s}cu's multiphase framework \cite{patrascu2010towards}, set disjointness has become a central hardness source for dynamic lower bounds, preprocessing-query tradeoffs, and static indexing problems \cite{kopelowitz2016higher,GoldsteinKLP17,Kopelowitz2020TowardsOS,bringmann2025tight}.

In our setting, the preprocessing-query data structure preprocesses the set of points $P$, and given a subset $P'\subseteq P$, it answers the detection or emptiness query $\ind_\sigma(P')$, of whether $P'$ contains an embedding of $\sigma$, or how many such embeddings it contains, denoted by $\#_\sigma(P')$. The focus is on \emph{geometric} preprocessing-query, where the subset of points $P'$ is obtained by intersecting $P$ with a geometric query range, such as an axis-aligned rectangle (see \Cref{fig:single-rect-query}), a slab (\Cref{fig:slab-query}), or multiple axis-aligned rectangles (\Cref{fig:multi-rect-query}). This is a natural generalization of orthogonal range searching in computational geometry, which corresponds to the case of $|\sigma|=1$.

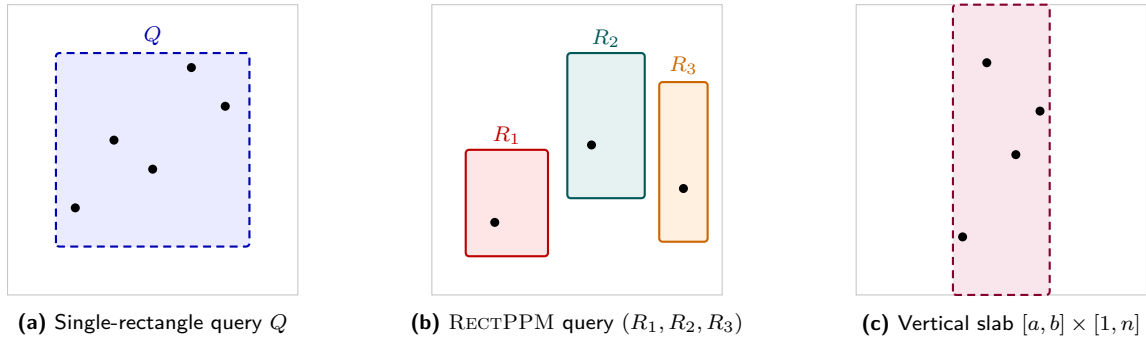
\begin{figure}[ht]
    \centering
    \begin{subfigure}[]{0.32\textwidth}
        \centering
        \begin{tikzpicture}[scale=0.64,
                rect/.style={rounded corners=1.3pt, line width=0.8pt},
                point/.style={circle, fill=black, inner sep=1.2pt}]
            \draw[gray!45] (0,0) rectangle (6,6);
            \fill[blue!8] (1.0,1.0) rectangle (5.0,5.0);
            \draw[rect, blue!70!black, densely dashed] (1.0,1.0) rectangle (5.0,5.0);
            \node[blue!70!black, font=\footnotesize] at (3.0,5.35) {$Q$};

            \node[point] at (1.4,1.8) {};
            \node[point] at (2.2,3.2) {};
            \node[point] at (3.0,2.6) {};
            \node[point] at (3.8,4.7) {};
            \node[point] at (4.5,3.9) {};
        \end{tikzpicture}
        \caption{Single-rectangle query $Q$}
        \label{fig:single-rect-query}
        \label{sig:single-rect-query}
    \end{subfigure}
    \hfill
    \begin{subfigure}[]{0.32\textwidth}
        \centering
        \begin{tikzpicture}[scale=0.64,
                rect/.style={rounded corners=1.3pt, line width=0.8pt},
                point/.style={circle, fill=black, inner sep=1.2pt}]
            \draw[gray!45] (0,0) rectangle (6,6);

            \fill[red!10] (0.7,0.8) rectangle (2.4,3.0);
            \draw[rect, red!75!black] (0.7,0.8) rectangle (2.4,3.0);
            \node[red!75!black, font=\footnotesize] at (1.55,3.3) {$R_1$};

            \fill[teal!10] (2.8,2.0) rectangle (4.4,5.0);
            \draw[rect, teal!70!black] (2.8,2.0) rectangle (4.4,5.0);
            \node[teal!70!black, font=\footnotesize] at (3.6,5.3) {$R_2$};

            \fill[orange!12] (4.7,1.1) rectangle (5.7,4.4);
            \draw[rect, orange!80!black] (4.7,1.1) rectangle (5.7,4.4);
            \node[orange!80!black, font=\footnotesize] at (5.2,4.75) {$R_3$};

            \node[point] at (1.3,1.5) {};
            \node[point] at (3.3,3.1) {};
            \node[point] at (5.2,2.2) {};
        \end{tikzpicture}
        \caption{\rppm query $(R_1,R_2,R_3)$}
        \label{fig:multi-rect-query}
        \label{sig:multi-rect-query}
    \end{subfigure}
    \hfill
    \begin{subfigure}[]{0.32\textwidth}
        \centering
        \begin{tikzpicture}[scale=0.64,
                rect/.style={rounded corners=1.3pt, line width=0.8pt},
                point/.style={circle, fill=black, inner sep=1.2pt}]
            \draw[gray!45] (0,0) rectangle (6,6);
            \fill[purple!10] (2.0,0) rectangle (4.0,6);
            \draw[rect, purple!70!black, densely dashed] (2.0,0) rectangle (4.0,6);
            \node[purple!70!black, font=\footnotesize] at (5.3,5.5) {};

            \node[point] at (2.2,1.2) {};
            \node[point] at (2.7,4.8) {};
            \node[point] at (3.3,2.9) {};
            \node[point] at (3.8,3.8) {};
        \end{tikzpicture}
        \caption{Vertical slab $[a,b]\times[1,n]$}
        \label{fig:slab-query}
        \label{sig:slab-query}
    \end{subfigure}
    \caption{Three query types over the permutation point set $P=\{(i,\pi(i)) : i\in[n]\}$. (a)~A single-rectangle query $Q$. (b)~A multi-rectangle (\rppm) query $(R_1,R_2,R_3)$: the $i$-th point must lie in~$R_i$. (c)~A vertical slab query $[a,b]\times[1,n]$.}
    \label{fig:summary-query-models}
\end{figure}

In a \emph{multi-rectangle} (or, respectively, \emph{multi-slab}) query for a pattern $\sigma$ of length $k$, the query is an ordered tuple $(R_1,\ldots,R_k)$ of $k$ axis-aligned rectangles (respectively, slabs), one per pattern index, and $\ind_\sigma(R_1,\ldots,R_k)=1$ if and only if there is an embedding $\phi$ of $\sigma$ such that the point realizing the \emph{$i$-th} element of the pattern lies in $R_i$, i.e., $(\phi(i),\pi(\phi(i)))\in R_i$ for every $i\in[k]$ (see \Cref{prob:rppm}).
    Repeated parameters are allowed; for example $\ind_{132}(I_1,I_1,I_2)$ means that the first two elements of the pattern must both come from $I_1$ while the third must come from $I_2$ (see, e.g., \Cref{ssec:tech-ov} and \Cref{sec:hardness}).

For standard point queries (that is, $|\sigma|=1$), both detection and counting type queries can be supported efficiently: with $\tilde{O}(n)$ preprocessing time and $\tilde{O}(1)$ per-query time.\footnote{Here and throughout the paper, $\tilde{O}$ hides polylogarithmic terms in $n$.}
In length-$2$ patterns~\cite{ben2024approximate,Duraj2019EquivalencesBT} the situation is already substantially different: there is a polynomial separation between detection (or approximate counting) and counting.
For detection and approximate counting, \cite{ben2024approximate} gives $\tf(n)$ preprocessing and $\tf(1)$ query time for both $\ind_{12}(Q)$ and $\ind_{12}(Q_1,Q_2)$, showing that single-rectangle and multi-rectangle queries both have near-linear preprocessing and near-constant query time.
For counting, ~\cite{Duraj2019EquivalencesBT} shows that single-slab queries $\#_{12}([a,b]\times[1,n])$ and two-slab queries $\#_{12}([a_1,b_1]\times[1,n],[a_2,b_2]\times[1,n])$ are equivalent, and belong to a larger complexity class, at least as hard as \prba.
In particular, there is a reduction from 3SUM to $\#_{12}$-slab queries~\cite{patrascu2010towards,kopelowitz2016higher,Duraj2019EquivalencesBT}. For example, the conjecture that 3SUM requires $n^{2-o(1)}$ time would imply that any such data structure with $\tf(n^p)$ preprocessing and $\tf(n^q)$ query time must satisfy $p+2q\geq 2$ for $q\in[0,1/2]$. Although this conjecture has recently been refuted~\cite{alman2026truly}, weaker conditional lower bounds on the preprocessing-query tradeoff still follow from the assumption that 3SUM requires $\Omega(n^{3/2+\eps})$ expected time for some constant $\eps>0$~\cite{kopelowitz2016higher,Duraj2019EquivalencesBT}.

\begin{figure}[t]
    \begin{tikzpicture}[
            x=10mm,
            y=12mm,
            >=Latex,
            src/.style={
                    draw,
                    rounded corners=2pt,
                    line width=0.8pt,
                    fill=gray!12,
                    inner xsep=6pt,
                    inner ysep=4pt,
                    align=center,
                    font=\footnotesize
                },
            det/.style={
                    draw,
                    rounded corners=2pt,
                    line width=0.8pt,
                    fill=blue!8,
                    inner xsep=6pt,
                    inner ysep=4pt,
                    align=center,
                    font=\footnotesize
                },
            cnt/.style={
                    draw,
                    rounded corners=2pt,
                    line width=0.8pt,
                    fill=orange!10,
                    inner xsep=6pt,
                    inner ysep=4pt,
                    align=center,
                    font=\footnotesize
                },
            newred/.style={
                    ->,
                    draw=red!70!black,
                    line width=0.95pt,
                    shorten <=2pt,
                    shorten >=2pt
                },
            known/.style={
                    ->,
                    draw=teal!65!black,
                    line width=0.95pt,
                    shorten <=2pt,
                    shorten >=2pt
                },
            knowneq/.style={
                    <->,
                    draw=teal!65!black,
                    line width=0.95pt,
                    shorten <=2pt,
                    shorten >=2pt
                },
            genstyle/.style={line width=0.95pt,
                    shorten <=2pt,
                    shorten >=2pt
                }
        ]

        \begin{scope}[xshift=-5mm]
            \node[src] (osd) at (0,1.3) {\prba};
            \node[src] (single3) at (0,3.3) {$\ind_{\sigma}(Q)$, $|\sigma|=3$};
            \node[src, text width=26mm] (sksd) at (0,-.75) {\textsc{Online~ $k$-SetDisj}};
            \node[src, text width=28mm] (aet) at (12.4,.5) {\alledges};

            \node[det, text width=30mm] (i132) at (4.6,1.3) {3-\sppm\\$\ind_{132}(S_1,S_2,S_3)$};
            \node[det, text width=30mm] (ir132) at (4.6,3.3) {3-\rppm\\$\ind_{132}(Q_1,Q_2,Q_3)$};
            \node[det, text width=36mm] (ik) at (4.6,-.75) {
            $(k+1)$-\sppm
            \\$\ind_{\sigma}(Q_1,Q_1,Q_2,\ldots,Q_k)$\\$\sigma=(1,k{+}1,k,k{-}1,\ldots,2)$};

            \node[cnt] (c12s) at (9.0,1.8) {$\#_{12}(S)$};
            \node[cnt] (ceqs) at (9.0,3.3) {$\#_{\textsc{EQ}}(S)$};
            \node[cnt] (ceqq) at (12.6,3.3) {$\#_{\textsc{EQ}}(Q)$};
            \node[cnt] (c123s) at (9.0,.7) {$\#_{123}(S)$};
            \node[cnt] (c132s) at (9.0,-.8) {$\#_{132}(S)$};

            \node[cnt] (c12q) at (12.6,1.8) {$\#_{12}(Q)$};

            \draw[newred,dashed] (single3.east) -- (ir132.west);
            \draw[newred,<->] (osd.east) -- (i132.west);
            \draw[newred,<->] (ir132.south) -- (i132.north);
            \draw[newred] (sksd.east) -- (ik.west);
            \draw[newred,dashed] (i132.east) -- (c132s.west);
            \draw[known,bend left=2, dashed] (osd.north east) to (ceqs.west);
            \draw[genstyle,draw=violet!70,<->] (c12s.south east) to (aet.north west);
            \node[font=\scriptsize] (offlabel) at (11.25,1) {Offline$^\dagger$};
            \draw[knowneq] (ceqs.south) -- (c12s.north);
            \draw[newred,dashed, bend left=30] (c12s.east) to (c12q.west);
            \draw[newred, bend left=30] (c12q.west) to (c12s.east);
            \draw[newred, dashed] (c12s.south) -- (c123s.north);

            \draw[newred,dashed, bend left=30] (ceqs.east) to (ceqq.west);
            \draw[newred, bend left=30] (ceqq.west) to (ceqs.east);
            \draw[newred, dashed] (c12s.south west) to[out=-145,in=130] (c132s.north west);

            \node[font=\scriptsize\bfseries] at (0,4.1) {Source Problems};
            \node[font=\scriptsize\bfseries] at (4.6,4.1) {Detection Queries};
            \node[font=\scriptsize\bfseries] at (9.0,4.1) {Counting (Slabs)};
            \node[font=\scriptsize\bfseries] at (12.6,4.1) {Counting (Rectangles)};

            \begin{scope}[on background layer]
                \fill[gray!4, rounded corners=3pt] (-1.8,-1.35) rectangle (1.8,4.35);
                \fill[blue!4, rounded corners=3pt] (2.7,-1.35) rectangle (6.5,4.35);
                \fill[orange!5, rounded corners=3pt] (7.6,-1.35) rectangle (10.4,4.35);
                \fill[orange!3, rounded corners=3pt] (11.2,1.2) rectangle (14.0,3.6);
            \end{scope}
        \end{scope}

        \begin{scope}[xshift=12mm]
            \draw[->, draw=black, line width=0.95pt,dashed] (-2.4,-1.85) -- (-1.5,-1.85);
            \node[font=\scriptsize, anchor=west] at (-1.5,-1.85) {dashed: trivial reductions};
            \draw[known] (2.8,-1.85) -- (3.8,-1.85);
            \node[font=\scriptsize, anchor=west] at (3.85,-1.85) {green: known reductions};
            \draw[newred] (7.6,-1.85) -- (8.6,-1.85);
            \node[font=\scriptsize, anchor=west] at (8.45,-1.85) {red: new in this paper};
        \end{scope}
    \end{tikzpicture}
    \caption{Reduction map among problems studied in this paper.
    An arrow $\textsc{A} \to \textsc{B}$ means that a data structure for \textsc{B} with preprocessing time $t_p$ and query time $t_q$ implies one for \textsc{A} with the same complexity (up to polylogarithmic factors); double arrows denote equivalences.
    We use $S$ for slab queries (vertical or horizontal slabs) and $Q$ for axis-aligned rectangle queries.
    $\#_{\sigma}(\cdot)$ denotes counting occurrences of pattern $\sigma$, and
    $\#_{\textsc{EQ}}(\cdot)$ counts pairs of points with equal $y$-coordinates within a query range.
    \alledges is the problem of computing, for each edge of a given graph, the number of triangles containing it~\cite{Duraj2019EquivalencesBT}.
    $^\dagger$The violet arrow labeled \emph{Offline} denotes an equivalence in the offline setting, where all queries are given in advance (rather than answered online via a data structure).
    }
    \label{fig:reductions-map}
\end{figure}
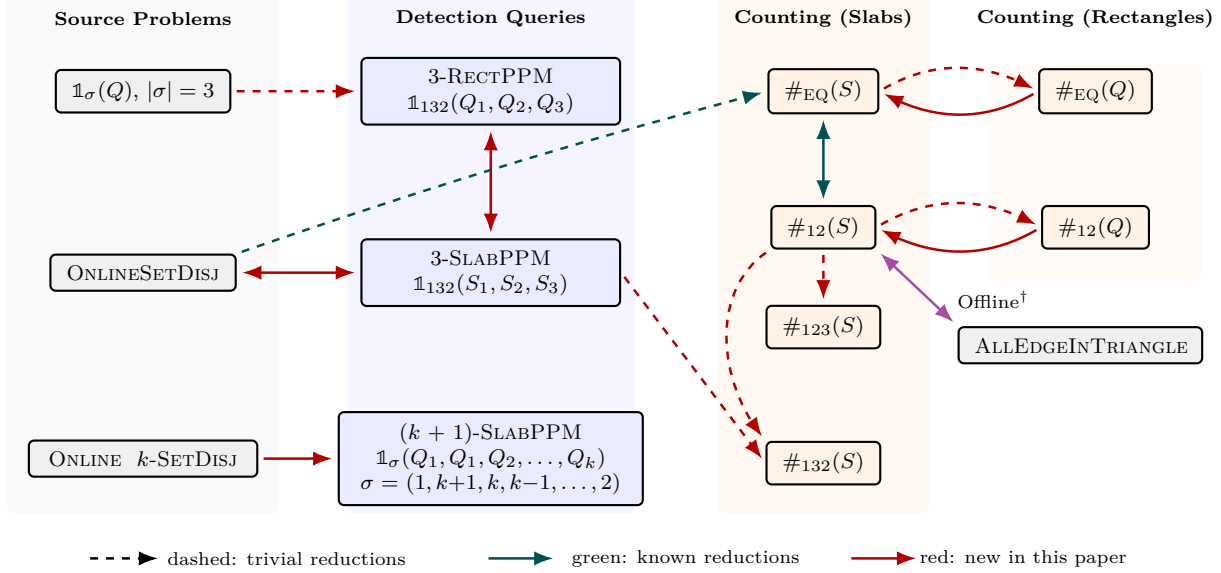
\subsection{Our Contributions}
In this paper, we extend the study of preprocessing-query PPM to patterns of length $3$ and beyond. As we shall see, the picture is nuanced, with a sharp dichotomy between single-rectangle and multi-rectangle queries. We also include algorithmic applications of our new data structures.

Our first main result is positive, showing that single-rectangle queries can be supported efficiently for all patterns of length $3$.
\begin{restatable}{theorem}{ThmOptEmpty}\label{thm:optimal emptiness}
    For all permutation patterns $\sigma$ of length $3$, there exists a deterministic data structure for computing $\ind_{\sigma}(Q)$ with $\tf(n)$ preprocessing time and $\tf(1)$ query time.
\end{restatable}

The following result is the main technical engine behind~\Cref{thm:optimal emptiness}, and we believe it may be of independent interest.
We prove an equivalence between rectangle and slab detection queries for patterns of length at most $3$, where two types of queries are \emph{equivalent} if a data structure for either one with $\tf(n^p)$ preprocessing and $\tf(n^q)$ query time can be transformed into a data structure for the other with the same bounds (up to polylogarithmic terms).

\begin{restatable}[Rectangle-to-Slab Reduction for Detection]{theorem}{ThmRSppm}\label{thm:re ppm}
    For any pattern $\sigma$ of length at most $3$,
    $\sigma$-multi-rectangle detection queries are equivalent to $\sigma$-multi-slab detection queries.
\end{restatable}

Both directions of the reduction preserve the number of ranges; in particular, a single-rectangle query is reduced to a single-slab query (and vice versa). This immediately yields a rectangle-to-slab equivalence for single-rectangle detection as well. We also prove the analogous equivalence between $\#_{12}$-rectangle queries and $\#_{12}$-slab queries.
This lets us plug the fastest slab data structure of Duraj, Kleiner, Polak, and Vassilevska Williams~\cite{Duraj2019EquivalencesBT} directly into the pattern-tree framework of Beniamini and Lavee~\cite{BL2024counting}, yielding the fastest known exact algorithm for $5$-pattern counting, improving upon the $\tilde{O}(n^{1.75})$-time algorithm of Beniamini-Lavee.
\begin{corollary}\label{cor:bl24}
    There is a deterministic algorithm for exact counting of any permutation pattern of length $5$, with running time $\tf(n^{2-2/(\omega+5)}) = O(n^{1.725})$.
\end{corollary}

So far, all results on preprocessing-query detection queries (in the literature and here) have exhibited an ideal (and near-optimal) efficiency tradeoff, with $\tilde{O}(n)$ preprocessing and $\tilde{O}(1)$ query time. This is true for standard range queries (patterns of length 1) and $12$-queries (patterns of length 2) for both \emph{single-rectangle} and \emph{multi-rectangle} queries. In addition,
 Theorem \ref{thm:optimal emptiness} shows that this ideal tradeoff is also achievable for single-rectangle queries for all patterns of length $3$. This raises a natural question:
 \begin{center}
    \emph{Is it always true that detection queries can be supported with $\tilde{O}(n)$ preprocessing and $\tilde{O}(1)$ query time?}
\end{center}

It turns out that the answer is (most likely) negative. But the story is even more interesting: we prove an \emph{equivalence} between preprocessing-query data structures with multi-rectangle queries for the pattern $\sigma=132$, and the \emph{Online Set Disjointness} problem. %

\begin{restatable}[\prba]{problem}{probOSDa}\label{prob:osd a}
    Given a family $\FF=(S_1,\ldots, S_m)$ of $m$ sets with total size $N=\sum_{S\in \FF} |S|$,
    preprocess $\FF$ so that given a query pair $(i,j)$, the task is to determine whether $S_i\cap S_j=\emptyset$.
    The generalization to $k$-wise queries $(i_1,\ldots,i_k)$, asking whether $\bigcap_{\ell=1}^k S_{i_\ell}=\emptyset$, is denoted \prbb.
\end{restatable}

\begin{restatable}[\rppm and \prba]{theorem}{ThmOSDEquiv}\label{thm:equiv osd and rppm}
    Multi-rectangle $132$ detection queries are equivalent to \prba.
\end{restatable}
We stress that \Cref{thm:equiv osd and rppm} is an \emph{equivalence}, not merely a one-way reduction. While reductions from conjectured-hard problems are common and yield lower bounds, an equivalence is strictly stronger: it means that the preprocessing/query complexity of multi-rectangle $132$-detection is \emph{precisely} that of \prba, so any improvement to either problem automatically transfers to the other.
Although \prba has served as a central hardness source for numerous problems~\cite{kopelowitz2016higher,GoldsteinKLP17,Kopelowitz2020TowardsOS,williams2020monochromatic,bringmann2025tight}, two-way equivalences with \prba are rare; to the best of our knowledge, ours is the first such equivalence arising from a geometric pattern-matching problem.

Both directions of the equivalence carry content.
On the algorithmic side, the reduction to \prba transfers both the preprocessing-query tradeoffs of~\cite{Kopelowitz2020TowardsOS} and the recent improvements for \prba in~\cite{alman2026truly} to multi-rectangle $132$-detection.
On the hardness side, the reverse reduction transfers the 3SUM-based conditional lower bounds~\cite{kopelowitz2016higher}; in particular, the weaker 3SUM assumption discussed above rules out $\tf(n)$ preprocessing with $\tf(1)$ query time.

The phenomenon continues beyond the $132$ pattern.
For longer patterns, we show that online $k$-set disjointness reduces to a natural $(k+1)$-slab query problem:
\begin{restatable}{theorem}{ThmHardLongInter}\label{thm:hardness interval k-1}
    $\prbb$ queries can be reduced to multi-slab queries with $k+1$ slabs, and a pattern $\sigma$ of length $k+1$ defined by
    $\sigma=(1,k+1,k,k-1,\ldots,2)$.
\end{restatable}
Hardness for \prbb is less understood.
The $k$-set disjointness conjecture~\cite{GoldsteinKLP17,bringmann2025tight} states that $N^{o(1)}$ query time requires $N^{k-o(1)}$ preprocessing time.
Under the $k$-set disjointness conjecture, our reduction implies that any data structure for the $(k+1)$-slab problem above with $n^{o(1)}$ query time must have $n^{k-o(1)}$ preprocessing time.
Bringmann et al.~\cite{bringmann2025tight} established this preprocessing lower bound under the conjecture that Zero-$(k+1)$-Clique requires $n^{k+1-o(1)}$ time, although the latter conjecture has recently been refuted~\cite{alman2026truly}.
The reductions remain valid and can still yield weaker preprocessing-query lower bounds from weaker time lower bounds for Zero-$(k+1)$-Clique, as in the 3SUM discussion above.

\paragraph{Algorithmic applications.}
Our final set of results are algorithms for exact and approximate counting which improve upon the state of the art in several regimes.

As mentioned, the rectangle-to-slab reduction above improves exact counting for $5$-patterns (see Corollary \ref{cor:bl24}). The next algorithmic application requires heavier machinery, and is thus stated as a theorem.
\begin{restatable}[$6$-Patterns]{theorem}{thmApproxCountingSix}\label{thm:approx counting6}
    For any pattern $\sigma$ of length $6$, and any $\eps>0$, there exists a randomized algorithm that outputs $\hat{t}$ such that \whp, $\hat{t}=(1\pm\eps)\#_{\sigma}(\pi)$.
    The running time of the algorithm is $\tf(n^{7/4}/\eps^{O(1)})$.
\end{restatable}
This is the first algorithm to beat the quadratic barrier for approximate counting of $6$-patterns. The proof makes use of several interesting ingredients, which we discuss in detail at the end of Section \ref{ssec:tech-ov}.

For completeness, we also record two black-box consequences of existing general frameworks. Combining \cite{GM14,DellLM24} gives the fastest known approximate-counting bound we are aware of, when $k$ is treated as a (small) parameter:
\begin{corollary}[\cite{GM14,DellLM24}]\label{cor:combine1}
    For any pattern $\sigma$ of length $k$, and any $\eps>0$, there exists a randomized algorithm that outputs $\hat{t}$ such that \whp, $\hat{t}=(1\pm\eps)\#_{\sigma}(\pi)$. The running time of the algorithm is
    \begin{align*}
        O_k(n^{\frac{k+2}{4}+\psi_k}\cdot \log^{6} n/\eps^2)\;,
    \end{align*}
    where $\psi_k=0$ if $k$ is even, and $\frac1{4k}$ if $k$ is odd.
\end{corollary}
In comparison, the best known algorithm is an exact counting algorithm by Berendsohn, Kozma and Marx \cite{BKM21}, with running time of the form $O(n^{k/4+o(k)})$, where the $o(k)$ term is growing with $k$ (and is not bounded by an absolute constant). In our upper bound, the leading term in the exponent is equal to theirs, $k/4$, but the other term is bounded by $1/2 + 1/4k < 1$.

Finally, combining \cite{GM14,meeks2019randomised} yields a general output-sensitive enumerating/listing algorithm:\footnote{By \emph{enumeration} or \emph{listing}, we mean outputting a list of $t$ copies of the pattern in the input (if such a number of copies exists). Both notions refer in the literature to this same task, and we use them interchangeably in the text.}
\begin{corollary}[\cite{GM14,meeks2019randomised}]\label{cor:combine2}
    For any pattern $\sigma$ of length $k$, there exists a randomized algorithm for enumerating (i.e., listing) all $N$ embeddings of $\sigma$ in $\pi$ in time $\tf(n\cdot (N+1))$ \whp.
\end{corollary}
This implies, in particular,
that exact counting can be done asymptotically faster than the general upper bound provided that $N \ll n^{k/4}$. Given that exact and approximate counting cannot be done in time $n^{o(k/\log k)}$ conditioned on ETH \cite{BKM21,Opler26}, we can also deduce that hard instances for the general case must have many copies of $\sigma$ in $\pi$.

\subsection{Recent Progress on 3SUM and Online Set Disjointness}
A very recent work~\cite{alman2026truly} provided a truly subquadratic algorithm for 3SUM and a truly subcubic algorithm for APSP, thereby refuting the corresponding hypotheses and, via known reductions, the Zero-Weight $k$-Clique hypotheses.
The conditional lower bounds discussed above are therefore no longer supported by these hypotheses, although weaker bounds still follow from weaker assumptions, as noted there.

None of our results relies on these hypotheses; in particular, the equivalence between \prba and multi-rectangle $132$-detection (\Cref{thm:equiv osd and rppm}) is unconditional.
Since \prba plays a central role in our results, we note that the same work also improved its preprocessing-query tradeoff when the universe is small and the preprocessing time is nearly quadratic.
We do not know how to leverage this improvement to speed up our approximate-counting algorithm for patterns of length~$6$ (\Cref{thm:approx counting6}), for two reasons.
First, that algorithm answers \rppm queries for patterns of length~$4$, which in general are at least as hard as \textsc{3-OnlineSetDisj} (\Cref{thm:hardness interval k-1}), whereas this improvement concerns \prba (i.e., \textsc{2-OnlineSetDisj}).
Second, the improvement requires nearly quadratic preprocessing time, while our algorithm can afford only $\tf(n^{7/4})$ preprocessing time.

\newcommand{\rangeT}{\wpa{RangeTree}}
\subsection{Technical Overview}
\label{ssec:tech-ov}
We start with the single-rectangle detection problem for patterns of length $3$, which proves~\Cref{thm:optimal emptiness}.
Given the permutation point set $P=\set{(i,\pi(i))\mid i\in[n]}$, we preprocess $P$ so that, for a query rectangle $Q$, we can decide whether $P\cap Q$ contains a prescribed pattern.
We focus on $132$ and write $\ind_{132}(Q)$ for this predicate; the other length-$3$ patterns are handled in the same way, up to the appropriate symmetries.

As a warm-up, we first overview \cite{ben2024approximate} algorithm for near-optimal detection for single- and multi-rectangle queries for patterns of length $2$, specifically $\ind_{12}(Q)$ queries, using the \emph{range-tree} data structures~\cite{BentleyFriedman79}.
We then explain why the na\"ive extension of this approach to length-$3$ patterns fails: it runs into the set-disjointness barrier captured by \Cref{thm:equiv osd and rppm}.
To overcome this obstacle, we have three main parts:
\begin{enumerate}
    \item Modifying the range-tree.
    \item Solving slab queries for $132$-patterns in $\tf(n)$ preprocessing and $\tf(1)$ query time.
    \item Showing that every rectangle query can be reduced to either a slab query or a ``geometrically simple'' query which can be solved in $\tf(1)$ time after $\tf(n)$ preprocessing.
\end{enumerate}
By combining these three parts, we get \Cref{thm:optimal emptiness}.

We conclude with an overview of the tools and ideas behind the proof of
\Cref{thm:approx counting6}, including the CSP formulation for PPM and the connection between approximate counting and decision reductions.

\paragraph{Warm-Up: $\ind_{12}(Q)$ Queries \cite{ben2024approximate}.}
We are given a permutation $\pi:[n]\to[n]$, viewed as a point set $P=\{(i,\pi(i))\}_{i\in[n]}$,
and wish to preprocess $P$ so that given an axis-parallel rectangle $Q=[a,b]\times[c,d]$,
we can quickly decide whether $Q\cap P$ contains two points that are increasing in both coordinates,
i.e.\ whether $\ind_{12}(Q)=1$.

Suppose we have a data structure \rangeT\footnote{See \Cref{ssec:range-trees} for a precise definition and overview.}~\cite{BentleyFriedman79} that preprocesses $P$ in $\tf(n)$ time, producing a family $\RR$ of $\tf(n)$ axis-parallel rectangles with the following properties:
$(i)$~each point of $P$ belongs to $\tf(1)$ rectangles from $\RR$; equivalently, $\sum_{R\in\RR}|P\cap R|=\tf(n)$;
$(ii)$~any axis-parallel rectangle $Q\subseteq [n]\times[n]$ has a partition $\CR{Q}\subseteq\RR$ containing $\tf(1)$ rectangles, such a partition can be computed in $\tf(1)$ time.
With such a structure, $\ind_{12}(Q)$ reduces to multi-rectangle queries over $\RR$:
\begin{align*}
    \ind_{12}(Q) = \bigvee_{R_1,R_2\in\CR{Q}}\ind_{12}(R_1,R_2)\;,
\end{align*}
where $\ind_{12}(R_1,R_2)$ indicates whether some two points of $P$ form a $12$-pattern with one point in $R_1$ and the other in $R_2$.
Indeed, any witness $(x_1,y_1),(x_2,y_2)\in Q\cap P$ with $x_1<x_2$ and $y_1<y_2$
must lie in some $R_1,R_2\in\CR{Q}$, giving $\ind_{12}(R_1,R_2)=1$.
Conversely, $\ind_{12}(R_1,R_2)=1$ for $R_1,R_2\in\CR{Q}$ implies $\ind_{12}(Q)=1$, by monotonicity with respect to set inclusion: if $R\subseteq R'$ then $\ind_{12}(R)\leq \ind_{12}(R')$.

It remains to show that each pair query $\ind_{12}(R_1,R_2)$ with $R_1,R_2\in\RR$ can be answered in $\tf(1)$ time after $\tf(n)$ preprocessing.
This task is split into two cases: $R_1=R_2$ and $R_1\neq R_2$.
The first case involves computing $\ind_{12}(R)$ for every $R\in\RR$ in $O(|P\cap R|)$ time. Computing $\ind_{12}$ of a point set in near-linear time is standard: sweep the points in increasing $x$-order while maintaining the minimum $y$-coordinate seen so far, and report that a copy was found once a point lies above the current minimum.
The total preprocessing time is $\sum_{R\in\RR} O(|P\cap R|)=\tf(n)$ by property~$(i)$.
We store these values in a table of size $|\RR|=\tf(n)$, so that $\ind_{12}(R)$ can be retrieved in $\tf(1)$ time for any $R\in\RR$.

For the second case, we use that any two distinct rectangles in $\CR{Q}$ are disjoint and hence are separated horizontally or vertically (illustrated in \Cref{fig:to2-12-pattern-pair}).
If $R_2$ is above $R_1$, then $\ind_{12}(R_1,R_2)=1$ if and only if the leftmost point in $P\cap R_1$ is left of the rightmost point in $P\cap R_2$; by storing these extremal points during preprocessing, we can answer such query in constant time. If $R_2$ is below $R_1$, then $\ind_{12}(R_1,R_2)=0$.
The case where $R_2$ is to the right of $R_1$ is symmetric. This completes the $\tf(1)$-time query algorithm for $\ind_{12}(Q)$.
\begin{figure}[h]
    \centering
    \begin{subfigure}[b]{0.48\textwidth}
        \centering
        \begin{tikzpicture}[scale=0.82,
                every node/.style={font=\small},
                point/.style={circle, inner sep=1.8pt, fill=#1, draw=#1!80!black, line width=0.4pt},
                point/.default={black},
                highlight/.style={circle, inner sep=2.6pt, fill=#1, draw=#1!50!black, line width=0.7pt},
                highlight/.default={black},
                rect/.style={rounded corners=1.5pt, line width=0.8pt},
            ]
            \fill[blue!4] (1,0.5) rectangle (6.25,5.75);
            \draw[rect, blue!40, densely dashed] (1,0.5) rectangle (6.25,5.75);
            \node[blue!60, font=\footnotesize\itshape] at (4.125, 6.05) {$Q$};

            \fill[red!6] (2.75,0.75) rectangle (5.75,2.5);
            \draw[rect, red!60!black] (2.75,0.75) rectangle (5.75,2.5);
            \node[red!70!black, font=\footnotesize] at (3.25, 0.95) {$R_1$};

            \fill[teal!6] (1.25,3.75) rectangle (4.5,5.5);
            \draw[rect, teal!70!black] (1.25,3.75) rectangle (4.5,5.5);
            \node[teal!70!black, font=\footnotesize] at (3.75, 5.3) {$R_2$};

            \node[point=red!70!black] at (4.8, 1.5) {};
            \node[point=red!70!black] at (5.2, 1.0) {};
            \node[point=red!70!black] at (4.2, 2.0) {};

            \node[point=teal!70!black] at (1.75, 5.15) {};
            \node[point=teal!70!black] at (2.3, 4.85) {};
            \node[point=teal!70!black] at (3.0, 4.4) {};

            \coordinate (leftmost) at (3.4, 2.25);
            \node[highlight=red!70!black] at (leftmost) {};
            \node[red!70!black, font=\scriptsize\itshape, anchor=south east] at ([xshift=-4pt, yshift=2pt]leftmost) {leftmost};

            \coordinate (rightmost) at (3.75, 4.0);
            \node[highlight=teal!70!black] at (rightmost) {};
            \node[teal!70!black, font=\scriptsize\itshape, anchor=north west] at ([xshift=4pt, yshift=-2pt]rightmost) {rightmost};

            \draw[{Stealth[length=4pt]}-{Stealth[length=4pt]}, thick, black!60]
            ([yshift=4pt]leftmost) -- ([yshift=-4pt]rightmost);
        \end{tikzpicture}
        \caption{Horizontal separation ($R_2$ above $R_1$).}
        \label{fig:to2-12-pattern-horizontal}
    \end{subfigure}
    \hfill
    \begin{subfigure}[b]{0.48\textwidth}
        \centering
        \begin{tikzpicture}[scale=0.82,
                every node/.style={font=\small},
                point/.style={circle, inner sep=1.8pt, fill=#1, draw=#1!80!black, line width=0.4pt},
                point/.default={black},
                highlight/.style={circle, inner sep=2.6pt, fill=#1, draw=#1!50!black, line width=0.7pt},
                highlight/.default={black},
                rect/.style={rounded corners=1.5pt, line width=0.8pt},
            ]
            \fill[blue!4] (1,0.5) rectangle (7.25,5.75);
            \draw[rect, blue!40, densely dashed] (1,0.5) rectangle (7.25,5.75);
            \node[blue!60, font=\footnotesize\itshape] at (4.125, 6.05) {$Q$};

            \fill[red!6] (1.25,2.5) rectangle (4.0,5.5);
            \draw[rect, red!60!black] (1.25,2.5) rectangle (4.0,5.5);
            \node[red!70!black, font=\footnotesize] at (3.5, 5.25) {$R_1$};

            \fill[teal!6] (4.75,0.75) rectangle (7.0,4.0);
            \draw[rect, teal!70!black] (4.75,0.75) rectangle (7.0,4.0);
            \node[teal!70!black, font=\footnotesize] at (6.5, 3.75) {$R_2$};

            \node[point=red!70!black] at (1.75, 5.15) {};
            \node[point=red!70!black] at (2.25, 4.85) {};
            \node[point=red!70!black] at (3.0, 4.4) {};

            \node[point=teal!70!black] at (6.5, 1.25) {};
            \node[point=teal!70!black] at (6.25, 1.75) {};
            \node[point=teal!70!black] at (5.75, 2.5) {};

            \coordinate (bottommost) at (3.75, 3.25);
            \node[highlight=red!70!black] at (bottommost) {};
            \node[red!70!black, font=\scriptsize\itshape, anchor=north] at ([yshift=-6pt]bottommost) {bottommost};

            \coordinate (topmost) at (5.0, 3.75);
            \node[highlight=teal!70!black] at (topmost) {};
            \node[teal!70!black, font=\scriptsize\itshape, anchor=south] at ([yshift=6pt]topmost) {topmost};

            \draw[{Stealth[length=4pt]}-{Stealth[length=4pt]}, thick, black!60]
            ([xshift=4pt]bottommost) -- ([xshift=-4pt]topmost);
        \end{tikzpicture}
        \caption{Vertical separation ($R_2$ right of $R_1$).}
        \label{fig:to2-12-pattern-vertical}
    \end{subfigure}
    \caption{Checking $\ind_{12}(R_1,\;R_2)$ for two rectangles in $\RR$ using one extreme-point comparison.}
    \label{fig:to2-12-pattern-pair}
\end{figure}
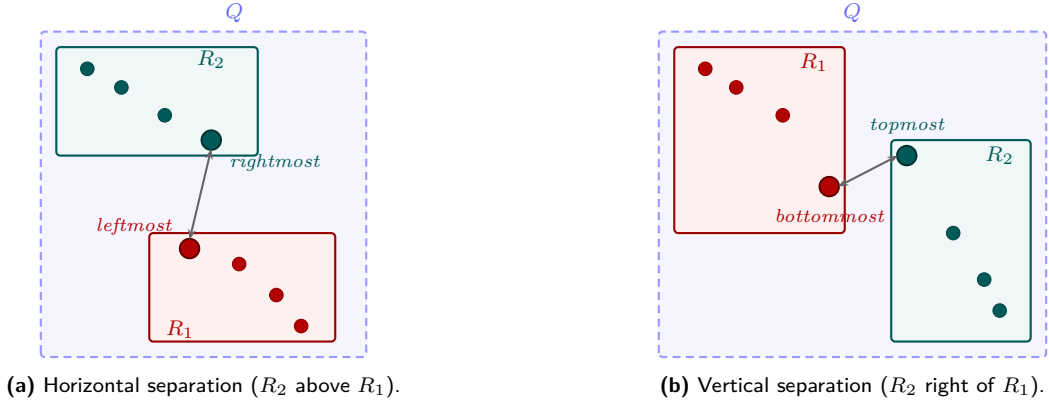

\paragraph{Why This Approach Breaks for (Non-Monotone) Length-$3$ Patterns.}
The above approach breaks down for non-monotone patterns such as $132$, as we show next.
Starting with the same data structure \rangeT, we can reduce $\indd(Q)$ to multi-rectangle queries $\indd(R_1,R_2,R_3)$ for $R_1,R_2,R_3\in\CR{Q}$, as
\begin{align*}
    \indd(Q) \;=\; \bigvee_{R_1,R_2,R_3\in\CR{Q}}\indd(R_1,R_2,R_3)\;,
\end{align*}
by the same witness and monotonicity argument we had for the length-$2$ pattern.
To get $\tf(1)$ query time for $\indd(Q)$, we need to answer each triple query $\indd(R_1,R_2,R_3)$ with $R_1,R_2,R_3\in\CR{Q}$ in $\tf(1)$ time, and since $|\CR{Q}|=\tf(1)$, this would give $\tf(1)$ query time for $\indd(Q)$.
This is where the approach breaks down: unlike the length-$2$ case, a triple query $\indd(R_1,R_2,R_3)$ over rectangles from $\RR$ cannot be answered in $\tf(1)$ time after $\tf(n)$ preprocessing, assuming that 3SUM requires $\Omega(n^{3/2+\eps})$ time for some constant $\eps>0$.
We explain this barrier, and then show how to bypass it.

\paragraph{The Online Set-Disjointness Barrier.}
\newcommand{\tds}{\wpa{SlabQuery}}
In \prba, the input is a universe $U$ and a family $\FF = (S_1,\ldots,S_m)$ of where $S_i\subseteq U$ with total size $N = \sum_{i=1}^m |S_i|$, and the goal is to preprocess $\FF$ to answer disjointness queries: given $S_i,S_j\in\FF$, is $S_i\cap S_j=\emptyset$?

We show a reduction from \prba to $132$-\sppm: a data structure that preprocesses a point set $P$ and answers queries $\indd(R_1,R_2,R_3)$, where each $R_i$ is a vertical slab, i.e., a rectangle of the form $[a,b]\times[1,n]$, and $\indd(R_1,R_2,R_3)=1$ iff some $132$-pattern in $P$ has its $i$th element in $R_i$.
We later give the reverse reduction, establishing equivalence between the two problems.
We note that encodings of set disjointness by $132$-copies were used in multiple occasions in the literature, including by Berendsohn~\cite{BerendsohnStreams26} to obtain space lower bounds for deterministic detection of non-monotone patterns in the streaming model, by reduction to the communication complexity of \textsc{Disjointness}; and by Jel\'inek, Opler, and Valtr~\cite{JelinekOV24} in the nondeterministic and randomized settings. Somewhat related constructions appear also in property testing lower bounds \cite{NewmanRRS19,BenEliezerC18}. The new connection we show here is between the \emph{online} variant of set disjointness, \prba, and the data structural, preprocessing-query variant of PPM.

The first step is to reduce the family $\FF$ to a point set $P$ in the plane, and The second step is to reduce a disjointness query $(S_i,S_j)$ to a slab query $\indd(R_1,R_2,R_3)$.
Given $\FF = (S_1,\ldots,S_m)$, we construct a new point set $P\subseteq[3N]\times[3N]$ of size $3N$, encoded as an array $A$ where every point $(x,y)\in P$ translates to the entry $A[x]=y$.
We assume that $U=[n]$, and that elements of $U$ are spaced at least $10$ apart (otherwise map $z$ to $10z$ for all $z\in U$).
We treat each set $S_i\in \FF$ as an ordered decreasing sequence of integers,
i.e., $S_i = (e_1,\ldots,e_{N_i})$ with $e_1 > e_2 > \cdots > e_{N_i}$,
and let $C = S_1 \circ \cdots \circ S_m = (e_1,\ldots,e_N)$ be the concatenation of all sets in order.
Define
\begin{align*}
    A = B_1 \circ B_2\;,&& B_1 = (e_1,e_1+2,e_2,e_2+2,\ldots,e_N,e_N+2)\;,&& B_2 = (e_1+1,e_2+1,\ldots,e_N+1)\;.
\end{align*}
In words, the array $A$ is the concatenation of $B_1$ and $B_2$. The array $B_1$ is obtained by duplicating each entry of $C$ and increasing the second copy by $2$, and $B_2$ is obtained by shifting $C$ up by $1$.
If values repeat, we replace $A$ by the stable-rank permutation obtained by ordering the pairs $(A[i],i)$ lexicographically; equivalently, equal values are broken from left to right.
This preserves the reduction, since intended witnesses use distinct values and tie-breaking equal values cannot create the required decreasing part of the pattern.
We keep the notation $A$ for this permutation.
This completes the construction of the point set $P$ and the permutation $A$.
Then, we initialize \tds on $P=\{(x,A[x])\mid x\in[3N]\}$, which after $\tf(N^p)$ preprocessing can answer slab queries $\indd(R_1,R_2,R_3)$ in $\tf(N^q)$ time.

We now explain how to use \tds to answer disjointness queries for \prba.
To answer a query $(S_i,S_j)$ for \prba, we turn it into a query $\indd(R_1,R_1,R_2)$ for \tds, where $R_1,R_2$ are vertical slabs. Specifically, we
set $R_1=[a,b]\times[1,3N]$ where $[a,b]$ is the interval of indices in $B_1$ corresponding to $S_i$, and $R_2=[c,d]\times[1,3N]$ where $[c,d]$ is the interval of indices in $B_2$ corresponding to $S_j$.
Then, $\indd(R_1,R_1,R_2)=1$ if and only if $S_i\cap S_j\neq\emptyset$.
See \Cref{fig:to2-hardness-two-rectangles} for an illustration of the reduction, and \Cref{ssec:hardness interval k-1} for the full details.

Two additional reductions we present $(i)$ a similar reduction from \prbb, a generalization of \prba in which queries ask whether $k$ sets have a common element, to \sppm with $k+1$ slabs, for the pattern $\sigma=(1,k+1,k,\ldots,2)$;
and $(ii)$ the reverse reduction from \sppm to \prba, which proceeds through a new intermediate problem we call \prbd, a generalization of \prba in which each set $S_i\in\FF$ is a collection of contiguous intervals $\set{[a_1,b_1],[a_2,b_2],\ldots,[a_k,b_k]}$ of $[U]$, the size of the input $N'$ is total number of intervals in the input, (and not the total size of the sets as in \prba),
and a query $(S_i,S_j)$ asks whether some interval in $S_j$ is contained in some interval in $S_i$ (instead of asking whether $S_i$ and $S_j$ have a common element as in \prba).
This reduction is slightly more involved and is given in \Cref{ssec:hardness prbd to prba}.

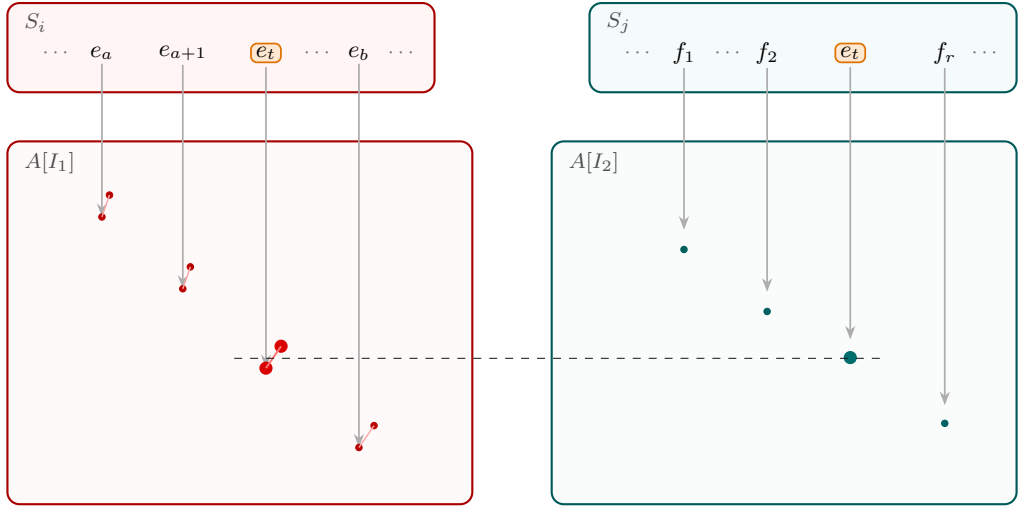
\begin{figure}[h]
    \centering
\begin{tikzpicture}[
    x=1cm,y=1cm,
    font=\small,
    >={Stealth[length=1.8mm]},
    every node/.style={inner sep=0pt},
    box/.style={rounded corners=1.6mm, line width=0.8pt},
    setL/.style={box, draw=red!65!black, fill=red!4},
    setR/.style={box, draw=teal!70!black, fill=teal!4},
    panelL/.style={box, draw=red!65!black, fill=red!2},
    panelR/.style={box, draw=teal!70!black, fill=teal!2},
    map/.style={-{Stealth[length=1.8mm]}, line width=0.65pt, draw=black!32},
    guide/.style={densely dashed, line width=0.65pt, draw=black!42},
    chain/.style={line width=0.65pt, draw=black!20},
    dotL/.style={circle, inner sep=.85pt, draw=red!72!black, fill=red!72!black},
    dotR/.style={circle, inner sep=.85pt, draw=teal!75!black, fill=teal!75!black},
    hitL/.style={circle, inner sep=1.5pt, draw=red!85!black, fill=red!85!black, line width=0.7pt},
    hitR/.style={circle, inner sep=1.5pt, draw=teal!85!black, fill=teal!85!black, line width=0.7pt},
    topname/.style={font=\footnotesize, text=black!70},
    emphasis/.style={
            rounded corners=0.8mm,
            draw=orange!88!black,
            fill=orange!18,
            line width=0.65pt,
            inner xsep=1.8pt,
            inner ysep=1.1pt
        },
    note/.style={font=\scriptsize\itshape, text=black!58}
    ]

    \draw[setL] (0.4,6.20) rectangle (6.05,7.38);
    \draw[setR] (8.10,6.20) rectangle (13.75,7.38);
    \node[topname, anchor=west] at (0.62,7.12) {$S_i$};
    \node[topname, anchor=west] at (8.32,7.12) {$S_j$};

    \node[font=\small] (sia)  at (1.65,6.72) {$e_a$};
    \node[font=\small] (siap) at (2.72,6.72) {$e_{a+1}$};
    \node[emphasis, font=\small] (sit) at (3.82,6.72) {$e_t$};
    \node[font=\small] (sib)  at (5.05,6.72) {$e_b$};
    \node[note] at (4.52,6.72) {$\cdots$};
    \node[note] at (1.06,6.72) {$\cdots$};
    \node[note] at (5.63,6.72) {$\cdots$};

    \node[font=\small] (sjl) at (9.35,6.72) {$f_1$};
    \node[font=\small] (sjm) at (10.45,6.72) {$f_2$};
    \node[emphasis, font=\small] (sjt) at (11.55,6.72) {$e_t$};
    \node[font=\small] (sjr) at (12.80,6.72) {$f_r$};
    \node[note] at (9.95,6.72) {$\cdots$};
    \node[note] at (8.76,6.72) {$\cdots$};
    \node[note] at (13.35,6.72) {$\cdots$};

    \draw[panelL] (0.4,0.75) rectangle (6.55,5.55);
    \draw[panelR] (7.60,0.75) rectangle (13.75,5.55);
    \node[topname, anchor=west] at (0.62,5.27) {$A[I_1]$};
    \node[topname, anchor=west] at (7.82,5.27) {$A[I_2]$};

    \coordinate (L1a) at (1.65,4.55);
    \coordinate (L1b) at (1.75,4.84);
    \coordinate (L2a) at (2.72,3.60);
    \coordinate (L2b) at (2.82,3.89);
    \coordinate (L3a) at (3.82,2.55);
    \coordinate (L3b) at (4.02,2.84);
    \coordinate (L4a) at (5.05,1.50);
    \coordinate (L4b) at (5.25,1.79);
    \coordinate (R1) at (9.35,4.12);
    \coordinate (R2) at (10.45,3.30);
    \coordinate (R3) at (11.55,2.69);
    \coordinate (R4) at (12.80,1.82);

    \draw[map] ($(sia.south)+(0,-0.05)$) -- (L1a);
    \draw[map] ($(siap.south)+(0,-0.05)$) -- (L2a);
    \draw[map] ($(sit.south)+(0,-0.05)$) -- (L3a);
    \draw[map] ($(sib.south)+(0,-0.05)$) -- (L4a);
    \draw[map] ($(sjl.south)+(0,-0.05)$) -- ($(R1)+(0.00,0.26)$);
    \draw[map] ($(sjm.south)+(0,-0.05)$) -- ($(R2)+(0.00,0.26)$);
    \draw[map] ($(sjt.south)+(0,-0.05)$) -- ($(R3)+(0.00,0.24)$);
    \draw[map] ($(sjr.south)+(0,-0.05)$) -- ($(R4)+(0.00,0.24)$);

    \node[dotL] at (L1a) {};
    \node[dotL] at (L1b) {};
    \node[dotL] at (L2a) {};
    \node[dotL] at (L2b) {};
    \node[hitL] (pLo) at (L3a) {};
    \node[hitL] (pHi) at (L3b) {};
    \node[dotL] at (L4a) {};
    \node[dotL] at (L4b) {};
    \node[dotR] at (R1) {};
    \node[dotR] at (R2) {};
    \node[hitR] (pMid) at (R3) {};
    \node[dotR] at (R4) {};

    \draw[red!35, line width=0.55pt] (L1a) -- (L1b);
    \draw[red!35, line width=0.55pt] (L2a) -- (L2b);
    \draw[red!55, line width=0.75pt] (L3a) -- (L3b);
    \draw[red!35, line width=0.55pt] (L4a) -- (L4b);

    \draw[dashed, gray!25!black] (3.4,2.68) -- (12,2.68);

\end{tikzpicture}
    \caption{Reduction from \prba to a two-rectangle $132$ query.
        Top: a set query $(S_i,S_j)$.
        Bottom: the induced query $\indd(I_1,I_1,I_2)$, where
        $A[I_1]=(e_1,e_1+2,e_2,e_2+2,\ldots)$ and
        $A[I_2]=(f_1+1,f_2+1,\ldots)$.}
    \label{fig:to2-hardness-two-rectangles}
\end{figure}

\paragraph{Bypassing The Online Set-Disjointness Barrier.}
\renewcommand{\CR}{\mathcal{C}}
Previously, we reduced the computation of $\indd(Q)$ to $\tf(1)$ multi-rectangle queries $\indd(R_1,R_2,R_3)$ with $R_1,R_2,R_3\in\CR(Q)$, where recall that $\CR(Q)$ is a partition of $Q$ into $\tf(1)$ rectangles from $\RR$, which is a family of $\tf(n)$ rectangles covering $[n]\times[n]$ we obtained from a range tree.
However, achieving our goal of $\tf(1)$-time queries after $\tf(n)$ preprocessing is impossible with this approach, assuming 3SUM requires $\Omega(n^{3/2+\eps})$ time.

To bypass this barrier, we begin with a simple observation.
Suppose that, given $Q$ and three rectangles $R_1,R_2,R_3\in\CR(Q)$, instead of answering $\indd(R_1,R_2,R_3)$ exactly, our data structure answers an easier query $\indd(R_1',R_2',R_3')$, where each $R_i'$ is some rectangle satisfying $R_i\subseteq R_i'\subseteq Q$. We refer to $\indd(R_1',R_2',R_3')$ as a \emph{relaxed} query for $\indd(R_1,R_2,R_3)$ with respect to $Q$.
By monotonicity,
\begin{align*}
    \indd(R_1,R_2,R_3) \leq \indd(R_1',R_2',R_3') \leq \indd(Q)\;.
\end{align*}
Thus, either answer is useful: if $\indd(R_1',R_2',R_3')=0$ then $\indd(R_1,R_2,R_3)=0$, and if $\indd(R_1',R_2',R_3')=1$ then $\indd(Q)=1$ and we are done
although we cannot conclude anything about $\indd(R_1,R_2,R_3)$.

It remains to specify the rectangles $R_1',R_2',R_3'$, and to show how to answer $\indd(R_1',R_2',R_3')$ within the desired time bounds.
For that, we must look into the details of a ``vanilla'' one-dimensional range tree.
Consider a family $\RR_x$ of $O(n\log n)$ rectangles covering $[n]\times[n]$, where each rectangle is a vertical slab of the form $[a,b]\times[1,n]$.
The family $\RR_x$ is constructed recursively: the first rectangle is $[1,n]\times[1,n]$, and each rectangle is split in half along its $x$-interval, e.g., the first split produces $[1,n/2]\times[1,n]$ and $[n/2+1,n]\times[1,n]$.
We identify $\RR_x$ with the nodes of a balanced binary tree $T_x$: the root corresponds to $[1,n]\times[1,n]$, and the $i$th level consists of $2^i$ vertical slabs of width $n/2^i$ that partition $[n]\times[n]$.
Each node $v$ of $T_x$ which associated with a rectangle (vertical slab) $R_v\in\RR_x$ also maintains the set of points $P_v=P\cap R_v$.

Using this simple family of rectangles $\RR_x$, we can solve a relaxed query for specific triples $R_1,R_2,R_3\in\CR(Q)$, as we explain next.
Consider a node $v$ of $T_x$, with rectangle $R_v=[a_v,b_v]\times[1,n]$, and three rectangles $R_1,R_2,R_3\in \CR(Q)$ with $R_i=[a_v,b_v]\times[c_i,d_i]$ for $i=1,2,3$.
That is, all three rectangles share the same $x$-interval $[a_v,b_v]$ as $R_v$. We then say that $R_1,R_2,R_3$ are \emph{aligned} with $R_v$.

We relax the query $\indd(R_1,R_2,R_3)$ to $\indd(Q',Q',Q')=\indd(Q')$, where $Q'\triangleq Q\cap R_v$, which in some sense ignores the original triplet.
To answer the relaxed query $\indd(Q')$, we use the node $v$ perspective, i.e., we restrict our attention to the rectangle $R_v$, and to the set of points $P_v$ contained in $R_v$.
From $v$'s perspective, the query $\indd(Q')$ is a ``horizontal slab'' query, since $Q'=[a_v,b_v]\times[c,d]$. That is the $x$-projection of $Q'$ and $R_v$ is $[a_v,b_v]$ and the $y$-projection of $Q'$ is contained in the $y$-projection of $R_v$.
Thus, if we had a data structure that preprocesses $P_v$ to answer such horizontal slab queries, we could answer $\indd(Q')$ efficiently, and thus answer the relaxed query for $\indd(R_1,R_2,R_3)$ efficiently, bypassing the hardness barrier.

To this end, we devise a new (to the best of our knowledge) data structure \tds that preprocesses a set of $n_v$ points $P_v$ in $\tf(n_v)$ time, and answers slab queries: given $a\leq b$ it computes $\indd([a,b]\times[1,n])$ in $\tf(1)$ time.
This data structure uses elementary dynamic programming techniques, as well as the $\ind_{12}(R)$ data structure from \cite{ben2024approximate}.
Initializing \tds on $P_v$ for every node $v$ of $T_x$ takes $\tf(n)$ total preprocessing time, since $\sum_{v\in T_x} |P_v| = O(n\log n)$, which follows as the tree has $\log n$ levels, and each level is a partition of $P$ so it has total size $n$.

~\\ To summarize, we showed how to answer a relaxed query $\indd(Q')$ for specific triples $R_1,R_2,R_3\in\CR(Q)$ that are aligned with some rectangle $R_v\in\RR_x$, by using a new data structure \tds that answers slab queries from the perspective of $v$.
Can this approach be extended to answer relaxed queries for any triple $R_1,R_2,R_3\in\CR(Q)$, even if they are not aligned with any rectangle in $\RR_x$?
We answer this question in the positive.

For that, we consider a simple variant of the range tree, which we call a \emph{grid range tree}.
Informally, a grid range tree allows us to partition a query rectangle $Q$ into $\tf(1)$ rectangles $\CR(Q)$, and further refine $\CR(Q)$ into a new partition $\Grid(Q)$ of $\tf(1)$ rectangles,
such that every triplet $(R_1,R_2,R_3)\in \Grid(Q)$ satisfies at least one of the following properties:
\begin{enumerate}
    \item
    The triplet $(R_1,R_2,R_3)$ is aligned with some vertex in $T_x$, or
    \item
    The triplet $(R_1,R_2,R_3)$ aligned with some vertex in $T_y$, where $T_y$ is a similar range tree built on the $y$-coordinates of the points: The family $\RR_y$ of rectangles corresponding to $T_y$ consists of horizontal slabs of the form $[1,n]\times[c,d]$, or
    \item
    Two rectangles $R_i,R_j$ out of $(R_1,R_2,R_3)$ have disjoint $x$-projection and disjoint $y$-projection.
\end{enumerate}

In the first and second case, we can reduce the query $\indd(R_1,R_2,R_3)$ to a relaxed query $\indd(Q')$ where $Q'$ is a horizontal or vertical slab which is aligned with some node of $T_x$ or $T_y$, and thus can be answered in $\tf(1)$ time using \tds.
In the last case, the additional structure makes it possible to answer $\indd(R_1,R_2,R_3)$ directly in $\tf(1)$ time, without considering a relaxed query.

\paragraph{Unaligned Triplets.}
Consider a triple query $\indd(R_1,R_2,R_3)$ with $R_1,R_2,R_3\in\Grid(Q)$ that are not all in the same row or column of the grid.
We can assume that each rectangle is non-empty, otherwise $\indd(R_1,R_2,R_3)=0$.
We do case analysis, based on whether the $x$-projection (and $y$-projection) of each pair of rectangles is disjoint (or equal), and so that using $\tf(1)$ emptiness axis parallel queries (i.e., whether an axis parallel rectangle $R'$ contains a point in $P$), we can determine whether $\indd(R_1,R_2,R_3)=1$ or $0$.

The general case, as well as a few additional cases (like three distinct rectangles that are in the same row or column) can be handled by similar case analysis, or more elegantly by presenting the problem as a \emph{binary constraint satisfaction problem} (CSP) following Berendsohn, Kozma, and Marx~\cite{BKM21} formulation.
That is, each query $\indd(R_1,R_2,R_3)$ naturally defines a binary CSP $M=(V,\DD,C)$ with three variables $V=\{v_1,v_2,v_3\}$, each with domain $\DD=\{D_1,D_2,D_3\}$, where $D_i=\set{t\in [n]\mid (t,\pi(t))\in R_i}$ is the set of indices of points in $R_i$, and binary set of constraints $C$ that encodes the constraints $v_1<v_2<v_3$ and $\pi(v_1)<\pi(v_3)<\pi(v_2)$.
The primal graph of $M$ is a multigraph on three vertices where each constraint between $v_i$ and $v_j$ corresponds to an edge.
A well-known result in the CSP literature is that if the primal graph of $M$ is acyclic, then $M$ can be solved in near-linear time via a sequence of domain reductions along the tree.
Here we show that if the primal graph is acyclic, then $M$ can be solved in $\tf(1)$ time (using axis parallel emptiness queries) after $\tf(n)$ preprocessing,
and that the primal graph is indeed acyclic when the triplet is unaligned, or even when it contains three distinct rectangles from $\Grid(Q)$.

\paragraph{Approximate Counting Patterns of Length $6$.}
Our starting point is the reduction of Dell, Lapinskas, and Meeks~\cite{DellLM22} from approximate counting to \emph{colorful} detection.
In our setting, colorful detection asks: given a point set $P$ and a coloring $\psi:[n]\to[k]$, does there exist an embedding $\varphi:[k]\to[n]$ of $\sigma$ in $P$ such that no color appears more than once among $\varphi(1),\ldots,\varphi(k)$?
Following~\cite{berendsohn2019complexity}, we denote this problem by \cppm.
By~\cite{DellLM22}, an algorithm for \cppm with running time $X$ implies an approximate counting algorithm with running time $X\cdot\poly(\log n,1/\eps)$, for any $\eps>0$.
Therefore, for patterns of length $6$, it suffices to solve \cppm efficiently.

We actually solve the stricter problem \pppm~\cite{GM14}.
In \pppm, the input is a pattern $\sigma$ of length $k$, a point set $P$, and a coloring $\psi:[n]\to[k]$.
Writing $A_i=\psi^{-1}(i)$ for each color class, the goal is to find an embedding $\varphi:[k]\to[n]$ such that $(\varphi(1),\ldots,\varphi(k))$ is order-isomorphic to $\sigma$ and $\varphi(i)\in A_i$ for every $i\in[k]$.

For constant $k$, \cppm reduces to \pppm by trying all $k!$ bijections between colors and positions.
Hence, solving \pppm for $k=6$ in time $\tf(n^{7/4})$ yields an $O(n^{7/4}\cdot\poly(\log n,1/\eps))$-time algorithm for approximating $\#_\sigma(P)$.

The algorithm reduces one \pppm instance with $k=6$ to $O(n)$ sub-instances with $k=4$.
Each sub-instance is obtained by fixing some embedding positions and shrinking the remaining color classes according to the induced order constraints.
As a result, an embedding of a sub-instance can be extended to an embedding of the original instance in constant time.
Each sub-instance is a colored multi-rectangle query for a length-$4$ pattern.

\textbf{Solving Each Sub-Instance.}
First, we reduce general $4$-rectangle queries to queries whose rectangles belong to a fixed family $\RR$ of $\tf(n)$ canonical rectangles.
While the fully general problem of $4$-rectangle queries is at least as hard as \textsc{3-OnlineSetDisj} (by \Cref{thm:hardness interval k-1}), the specific length-$4$ patterns arising from our reduction can be solved with $(p,q)=(7/4,3/4)$, i.e., $p+3q=4$. This matches the trade-off $p+(k-1)q=k$ that the standard heavy--light argument gives for \prbb with $k=4$ sets (precompute all $k$-tuples of heavy sets in $\tf(\tau(N/\tau)^k)$ time, and answer any query containing a light set in $O(\tau)$ time), which is the natural benchmark for a query consisting of four ranges.
The cases with this fast running time are characterized by the primal graph of the corresponding CSP: roughly speaking, the desired trade-off is achieved when this graph is a tree plus one extra edge.
Using a computer-assisted case analysis, we verify that all sub-instances produced by our reduction fall into this tractable family.

\textbf{Producing the Sub-Instances.}
To produce the sub-instances, we use several tools.
The simplest one is to fix one point in the embedding.
Pick a color class, say $A_1$, iterate over its points, and fix the first embedding position to each one.
This yields $|A_1|$ new instances of length $5$; for each fixed point, we shrink the remaining color classes according to the order constraints imposed by that point.
Each resulting instance preserves the extension property above.
By itself, this reduces length $6$ to length $5$, not to length $4$.
To reach length $4$, we combine several such reductions, choosing the fixed position according to the structure of $\sigma$.

For example, let $\sigma=124356$, and let $M=(V,\DD,C)$ be the associated CSP instance with $V=\{v_1,\ldots,v_6\}$, domains $D_i=A_i$, and binary order constraints $C$.
Fix $v_3=t$ for some $t\in A_3$, and let $M_t$ be the obtained CSP instance after removing $v_3$ and shrinking the remaining domains according to constraints incident to $v_3$.
The constraint graph of $M_t$ has only one edge between the two vertex sets $\set{v_1,v_2}$ and $\set{v_4,v_5,v_6}$.
We show that a CSP whose constraint graph has a cut of size at most one can be decomposed into two almost independent subproblems, one with $2$ variables and one with $3$ variables, where each oracle call is handled by our length-$4$ data structure.
This cut-based decomposition does not suffice for all patterns of length $6$: some yield constraint graphs with no small cut after a single variable is fixed.
To handle these cases, we additionally borrow ideas from \cite{ben2024approximate} that allow us to discard one $x$-constraint and one $y$-constraint (though not edges of our choosing), reducing the graph to a tractable structure; see \Cref{sec:approx counting6} for details.

\newcommand{\x}{x}
\crefname{enumi}{\unskip}{\unskip}

\subsection{Related Work}
Here we provide a more extensive background on PPM and its variants.

\paragraph{NP-completeness and early algorithms.}
The Permutation Pattern Matching problem was shown to be NP-complete by Bose, Buss, and Lubiw~\cite{BBL98}, who also gave polynomial-time algorithms when the pattern is separable (i.e., avoids both $2413$ and $3142$).
Albert, Aldred, Atkinson, and Holton~\cite{AAAH01} showed that detection of any fixed pattern of length $4$ can be solved in $O(n \log n)$ time.
Ahal and Rabinovich~\cite{AR08} gave an $O(n^{0.47k + o(k)})$ time algorithm for general $k$.

\paragraph{Detection: general $k$.}
A breakthrough was achieved by Guillemot and Marx~\cite{GM14}, who gave the first fixed-parameter tractable (FPT) algorithm for detection, running in $2^{O(k^2 \log k)} \cdot n$ time.
Their proof, which introduced a ``structure versus chaos'' dichotomy, laid the foundations for the subsequent development of the theory of twin-width~\cite{BGHO+22}.
The running time was improved to $2^{O(k^2)} \cdot n$ by Fox~\cite{fox2013stanley}, in his work showing that Stanley--Wilf limits are typically exponential, strengthening the Marcus--Tardos theorem~\cite{marcus2004excluded}.
In the exponential regime (where $k$ is large), Bruner and Lackner~\cite{BL12} gave an $O(1.79^n \cdot n \cdot k)$ time algorithm, which was improved by Berendsohn, Kozma, and Marx~\cite{BKM21} to $O(1.6181^n)$, and further to $O(1.415^n)$ by Gawrychowski and Rzepecki~\cite{gawrychowski2022faster}.

\paragraph{Exact counting:  general $k$.}
For exact counting, Berendsohn, Kozma, and Marx~\cite{BKM21} gave an $O(n^{k/4 + o(k)})$ time algorithm, significantly improving the previous $O(n^{0.47k + o(k)})$ time bound for detection that was given in~\cite{AR08} (which extends to counting).
Their algorithms are based on formulating the problem as a binary CSP and bounding the treewidth of the resulting constraint graph.
On the lower-bound side, the same work~\cite{BKM21} showed that no algorithm can solve exact counting in time $f(k) \cdot n^{o(k/\log k)}$ for any function $f$, unless the Exponential Time Hypothesis (ETH) fails.
Jel{\'\i}nek, Opler, and Pek{\'a}rek~\cite{JOP21} showed that a slightly weaker lower bound of $n^{\Omega(k/\log^2 k)}$ holds even when the pattern avoids a fixed large enough permutation.

\paragraph{Exact counting: small $k$.}
For small values of $k$, the picture is more nuanced.
When $k = 2$, counting matches is equivalent to counting inversions, which can be solved in $O(n \log n)$ time via merge sort, and in $O(n \sqrt{\log n})$ time in the Word RAM model \cite{ChanPatrascu2010}.
When $k = 3$, all patterns can be counted in near-linear time $\TO{n}$, as shown by Even-Zohar and Leng~\cite{EL21}, who introduced the framework of \emph{corner tree formulas}.
For $k = 4$, the same work~\cite{EL21} showed that $8$ of the $24$ patterns can be counted in $\TO{n}$ time, and gave an $\TO{n^{3/2}}$ time algorithm for the remaining $16$ patterns.
Dudek and Gawrychowski~\cite{dudek2020counting} proved a remarkable structural result: counting these ``hard'' $4$-patterns is \emph{equivalent} (via near-linear time reductions) to counting $4$-cycles in sparse graphs;
an $\tf(n^\gamma)$ time algorithm for counting $4$-patterns implies a $\tf(m^\gamma)$ time algorithm for counting $4$-cycles in a graph with $m$ edges, and vice versa.
For counting $4$-cycles the best known algorithm runs in $\tf(m^{\frac{4\omega-1}{2\omega+1}})$ time~\cite[Lemma 4.2]{williams2014finding}, where $\omega \le 9/4$ is the matrix multiplication exponent~\cite{OpenAI26omega}, giving an $O(n^{1.455})$ time algorithm for counting hard $4$-patterns.
On the lower-bound side, Yuster and Zwick~\cite{YusterZwick97} conjectured that detecting a $4$-cycle in an $n$-vertex graph requires $n^{2-o(1)}$ time, which implies that detecting a $4$-cycle in a graph with $m$ edges requires $m^{4/3-o(1)}$ time~\cite[Theorem~2]{DahlgaardKS17}.
Hence, under the same conjecture, counting hard $4$-patterns requires $n^{4/3-o(1)}$ time, by the reduction of Dudek and Gawrychowski~\cite{dudek2020counting} from counting $4$-cycles.
Closing the gap between the $\tf(n^{\frac{4\omega-1}{2\omega+1}})$ upper bound and the conditional $n^{4/3-o(1)}$ lower bound is a major open problem in the area; even if $\omega = 2$, the upper bound would only be $\tf(n^{7/5})$.

More recently, Beniamini and Lavee~\cite{BL2024counting} introduced \emph{pattern-trees}, a generalization of corner trees, and used them together with multidimensional segment trees (which they call \emph{rectangle-trees}) to obtain an $\tf(n^2)$ time algorithm for exactly counting all patterns of length $k \le 7$, and an improved $\TO{n^{7/4}}$ time algorithm for $k = 5$.

\paragraph{Approximate counting: small and general $k$.}
Approximate counting refers to the problem of retrieving a $(1+\eps)$-multiplicative approximation of the true count, where $\eps$ is some small constant.
For $k = 2$, Andersson and Petersson~\cite{AnderssonP95} gave an $O(n \log\log n)$ time $(1+\eps)$-approximation algorithm for counting inversions, and Chan and Pătrașcu~\cite{ChanPatrascu2010} improved this to linear time.
More recently, Ben-Eliezer, Mitrovi\'c and Srivastava \cite{ben2024approximate} proved that approximate counting can be done in $\tilde{O}(n)$ time for all patterns up to length $5$.
Together with the conditional $n^{4/3-o(1)}$ lower bound for exactly counting hard $4$-patterns discussed above, this gives a conditional separation between approximate and exact counting for $k = 4, 5$. However, it turns out that as $k$ grows, approximate counting is qualitatively much more similar to exact counting: Opler \cite{Opler26} proved very recently that no $n^{o(k / \log k)}$ approximate counting algorithm exists, assuming ETH.

\paragraph{\prba Overview.}
Online set disjointness is a well-studied problem~\cite{CohenP10,patrascu2010towards,kopelowitz2016higher,GoldsteinKLP17,williams2020monochromatic,bringmann2025tight}.
Kopelowitz, Pettie, and Porat~\cite[Corollary~1.8]{kopelowitz2016higher} showed that, under the 3SUM conjecture, any data structure for \prba with $\tf(N^p)$ preprocessing time and $\tf(N^q)$ query time, where $q\in[0,1/2]$, must satisfy $p+2q\ge 2$.
The 3SUM conjecture was very recently refuted~\cite{alman2026truly}; the same work gives an improved data structure for \prba over small universes~\cite[Theorem~3]{alman2026truly}.

An upper bound for \prba is given in \cite{Kopelowitz2020TowardsOS}, specified in the following theorem.
\begin{theorem}[{\cite[Theorem~1]{Kopelowitz2020TowardsOS}}]\label{thm:setdis upper bound}
    There exists a data structure for \prba with preprocessing $\tf(N^p)$ and query time $\tf(N^q)$ satisfying
    \begin{align*}
        p + \frac{2}{\omega-1}q & = 2,                      &  & 0 \le q \le \frac{\omega-1}{\omega+1}, \\
        \frac{2}{\omega-1}p + q & = 1 + \frac{2}{\omega-1}, &  & \frac{\omega-1}{\omega+1} \le q \le 1.
    \end{align*}
    If $\omega=2$, the upper bound simplifies to:
    \begin{align*}
        p+2q & = 2, &  & 0\le q\le 1/3, \\
        2p+q & = 3, &  & 1/3\le q\le 1.
    \end{align*}
\end{theorem}

\renewcommand{\CR}{\mathcal{C}}
\section{Preliminaries}
\label{sec:preliminaries}
\subsection{Permutation Pattern Matching Definitions}
Given a permutation $\pi:[n]\to [n]$, we identify $\pi$ with the point set $P=\set{(i,\pi(i))\mid i\in[n]}\subset [n]\times[n]$.
Two sequences of distinct integers $(a_1,\ldots,a_k)$ and $(b_1,\ldots,b_k)$ are \emph{order-isomorphic} if they have the same relative order, i.e., $a_i < a_j$ if and only if $b_i < b_j$ for all $i,j\in[k]$.

\begin{problem}[Permutation Pattern Matching (PPM)]
Given a text permutation $\pi:[n]\to[n]$ and a pattern $\sigma:[k]\to[k]$, the goal is to decide whether there exists an increasing mapping $\phi:[k]\to[n]$ such that the sequence $(\pi(\phi(1)),\ldots,\pi(\phi(k)))$ is order-isomorphic to $\sigma$. We call such a mapping $\phi$ an \emph{embedding} of $\sigma$ in $\pi$.
\end{problem}
We use $\inds(\pi)$ (or $\inds(P)$) to denote the indicator variable that is $1$ if there exists an embedding of $\sigma$ in $\pi$ (or $P$), and $0$ otherwise. We also use $\#_\sigma(\pi)$ to denote the number of embeddings of $\sigma$ in $\pi$.
The fastest known algorithm for PPM runs in time $O(2^{k^2}n)$~\cite{GM14,fox2013stanley}.
We consider several variations of PPM, each imposing additional constraints on the embedding $\phi$.
In all of the following problems, the input includes a partition of the point set $P$ into $k$ color classes $P_1,\ldots,P_k$.

\begin{problem}[\pppm {\cite{GM14}}]
\label{prob:pppm}
The embedding $\phi$ must respect the colors, i.e., $\phi(i)\in P_i$ for every $i\in[k]$.
\end{problem}
The \pppm problem was considered in~\cite{GM14}, where it was shown that under ETH, this problem cannot be solved in time $f(k)n^{o(k/\log k)}$.
A relaxation of \pppm, which we refer to as \cppm, is the following problem:
\begin{problem}[\cppm]
\label{prob:cppm}
The embedding $\phi$ must take exactly one point from each color class, i.e., $\phi(i)\in P_{c(i)}$ for some injective function $c:[k]\to[k]$.
\end{problem}
This problem was studied in~\cite{berendsohn2019complexity} under the name \textsc{Surjective Colored PPM}. \cite{berendsohn2019complexity} proves that under ETH it cannot be solved in time $f(k)n^{o(k/\log k)}$.
Clearly, an algorithm for \pppm with time $T(n,k)$ implies an algorithm for \cppm with time $O(k!\cdot T(n,k))$ by iterating over all injective functions $c:[k]\to[k]$.

The final problem we consider is a geometric preprocessing query of PPM, which we refer to as \rppm.

\begin{problem}[\rppm]
\label{prob:rppm}
Given a pattern $\sigma:[k]\to[k]$, a text permutation $\pi:[n]\to[n]$, and a partition of $[n]$ into $k$ color classes $P_1,\ldots,P_k$, preprocess $\sigma$ and $(P_1,\ldots,P_k)$ so that, given a query consisting of $k$ axis-parallel rectangles $R_1,\ldots,R_k\subseteq[n]\times[n]$,
the goal is to decide whether there exists an increasing embedding $\phi$ of $\sigma$ in $\pi$ such that
\begin{enumerate}
    \item $\phi(i)\in P_i$ for every $i\in[k]$ \emph{(color constraint)}, and
    \item $(\phi(i),\pi(\phi(i)))\in R_i$ for every $i\in[k]$ \emph{(rectangle constraint)}.
\end{enumerate}
\end{problem}
Let $R_i=[a_i,b_i]\times[c_i,d_i]$ be the $i$-th query rectangle. If $c_i=1$ and $d_i=n$ for every $i\in[k]$, then the rectangles are called vertical slabs, and the problem is denoted by \sppm.
\subsection{Range Trees}
\label{ssec:range-trees}
We build upon the well-known range tree data structure~\cite{bentley1978decomposable,deBerg08}.
A range tree for a set of $n$ points $P\subseteq [1,n]\times[1,n]$ has two levels of balanced binary search trees.
The \emph{outer tree} $T$ is a balanced BST on $x$-coordinates.
Each node $v\in T$ is associated with a vertical slab $R_v=[a_v,b_v]\times[1,n]$ and the subset $P_v=P\cap R_v$.
Nodes at the same depth have pairwise disjoint slabs whose union is $[1,n]\times[1,n]$;
for any node $u$ with children $v_1,v_2$, we have $R_u=R_{v_1}\cup R_{v_2}$ and $\lvert|P_{v_1}|-|P_{v_2}|\rvert\le 1$, so the tree has depth $O(\log n)$.
Each leaf contains exactly one point.

To obtain a two-dimensional range tree, each outer node $v$ additionally stores an \emph{inner tree} $T_v$ which is a balanced BST on the $y$-coordinates of $P_v$.
Each inner node $u\in T_v$ is associated with the rectangle $R_u=[a_v,b_v]\times[c_u,d_u]$ (inheriting the $x$-range from $v$) and the subset $P_u=P\cap R_u$; inner trees satisfy the same structural properties.
The set of all rectangles associated with nodes (outer or inner) is denoted $\RR$; standard bounds give $|\RR|=\tf(n)$.

\paragraph{Canonical Rectangles ($\CR(Q)$).}
Given a query rectangle $Q=[a,b]\times[c,d]$, we define its canonical nodes as follows.
Consider a node $v\in T$ with a parent $v'$ (if it exists) and associated slabs $R_v$ and $R_{v'}$, respectively. Then $v$ is an \emph{outer canonical node} for $Q$ if $R_v\subseteq [a,b]\times[1,n]$ and either $v$ is the root or $R_{v'}\not\subseteq [a,b]\times[1,n]$.
An inner node $u$ inside an inner tree $T_v$ is canonical for $Q$, if $P_u\subseteq Q$ and either $u$ is the root of $T_v$ or $P_{u'}\not\subseteq Q$, where $u'$ is the parent of $u$ in $T_v$.
The set $\CR(Q)$ is the set of all rectangles $R_u$ associated with canonical inner nodes $u$ in the inner trees $T_v$ of all outer canonical nodes $v$ for $Q$.

The canonical nodes and their associated rectangles can be found in
two stages: search for $a$ and $b$ in the outer tree $T$ to identify $O(\log n)$ outer canonical nodes whose slabs partition $[a,b]\times[1,n]$; then, for each such node $v$, search for $c$ and $d$ in $T_v$ to find $O(\log n)$ canonical inner nodes whose sets $P_u$ partition $P_v\cap Q$. This yields $|\CR(Q)|=O(\log^2 n)$ rectangles, computable in $O(\log^2 n)$ time.

\paragraph{Range Trees and \rppm.}
A range tree reduces a single-rectangle query $\inds(Q)$ for a length-$k$ pattern $\sigma$ to $|\CR(Q)|^k = O(\log^{2k} n)$ multi-rectangle queries over~$\RR$:
\begin{align}
    \inds(Q) = \bigvee_{(R_1,\ldots, R_k)\in {\CR(Q)}^k} \inds(R_1,\ldots,R_k)\;.
\end{align}
Correctness holds because every embedding $\phi$ of $\sigma$ in $Q$ maps each $i\in[k]$ to a point $(\phi(i),\pi(\phi(i)))$ lying in some $R_i\in\CR(Q)$. (Note that since the sets $P_u$ of the canonical inner nodes partition $P\cap Q$, each embedding is captured by exactly one tuple on the right-hand side; consequently the analogous decomposition holds for counting queries with a sum in place of the disjunction: $\#_\sigma(Q)=\sum_{(R_1,\ldots,R_k)\in\CR(Q)^k}\#_\sigma(R_1,\ldots,R_k)$.)
This reduction alone does not suffice: by \Cref{thm:equiv osd and rppm} and~\cite[Theorem~1.4]{kopelowitz2016higher}, $\indd(R_1,R_1,R_2)$ cannot be answered in $\tf(1)$ time after $\tf(n)$ preprocessing, assuming 3SUM requires $\Omega(n^{3/2+\eps})$ expected time for some constant $\eps>0$.

\paragraph{Grid Partition $\Grid(Q)$.}
\newcommand{\LW}{\ell}
We also partition $Q$ into $O(\log^2 n)$ rectangles forming a grid $\Grid(Q)$ (see \Cref{fig:grid-partition}), in which every pair of cells has $x$-projections that are either identical or disjoint, and likewise for $y$-projections.
This decomposition is standard in orthogonal range searching~\cite{bentley1978decomposable,Lueker78,Willard85,deBerg08}, and is more convenient than the canonical-rectangle decomposition for reducing rectangle queries to slab queries (\Cref{sec:rect-to-slab}).

Consider two one-dimensional range trees $T_x$ and $T_y$ built on the $x$- and $y$-coordinates of $P$, respectively.
Each node $v\in T_x$ corresponds to a vertical strip $R_v=[a_v,b_v]\times[1,n]$; each node $u\in T_y$ to a horizontal strip $R_u=[1,n]\times[c_u,d_u]$.
Given $Q=[a,b]\times[c,d]$, the \emph{canonical nodes} $\CR_x([a,b])$ are the $O(\log n)$ inclusion-maximal nodes in $T_x$ whose strips lie within $[a,b]\times[1,n]$; their strips partition $[a,b]\times[1,n]$ into $r_x=O(\log n)$ vertical slabs whose boundaries form a set $\LW_y$ of $O(r_x)$ vertical lines.
Similarly, $\CR_y([c,d])$ consists of $r_y=O(\log n)$ nodes partitioning $[1,n]\times[c,d]$ into horizontal slabs, inducing a set $\LW_x$ of $O(r_y)$ horizontal lines.
We define $\Grid(Q)$ as the faces of the arrangement of $\LW_x\cup\LW_y$ restricted to $Q$, giving a partition of $Q$ into $O(r_x\cdot r_y)=O(\log^2 n)$ axis-aligned cells.
As with canonical rectangles,
\begin{align*}
    \inds(Q)=\bigvee_{(R_1,\ldots,R_k)\in\Grid(Q)^k}\inds(R_1,\ldots,R_k)\;.
\end{align*}
While the grid construction above is standard, we highlight the following property which, to the best of our knowledge, is not usually emphasized in the literature.
Every row $L=[a,b]\times[y_1,y_2]$ of $\Grid(Q)$ has a corresponding node $u\in T_y$ with $R_u=[1,n]\times[y_1,y_2]$, so that $R_u$ and $L$ share the same $y$-projection (and symmetrically for columns and $T_x$).
This means that any data structure stored at $u$ for answering slab queries on $P_u$ can be used to answer queries restricted to the row $L$, which is the key enabler for reducing rectangle queries to slab queries in \Cref{sec:rect-to-slab}.
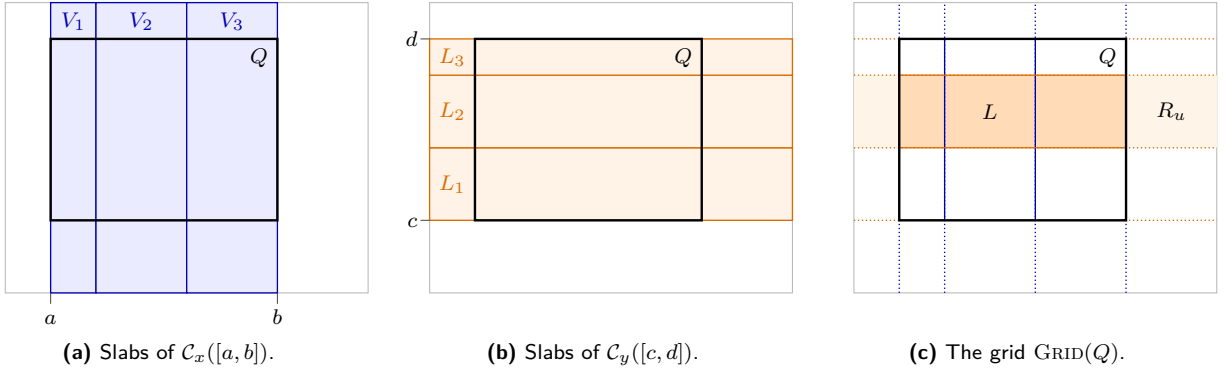
\begin{figure}[t]
    \centering
    \tikzset{
        gp/.style={x=0.30cm, y=0.24cm, every node/.style={font=\footnotesize}},
        gp frame/.style={draw=gray!55, line width=0.4pt},
        gp Q/.style={draw=black, line width=0.9pt},
        gp xslab/.style={draw=blue!70!black, fill=blue!8, line width=0.5pt},
        gp yslab/.style={draw=orange!85!black, fill=orange!10, line width=0.5pt},
        gp tick/.style={draw=gray!60!black, line width=0.4pt},
        gp ticklabel/.style={inner sep=1pt, text height=1.5ex, text depth=.25ex},
    }
    \begin{subfigure}[b]{0.32\textwidth}
        \centering
        \begin{tikzpicture}[gp]
            \useasboundingbox (-1.6,-1.8) rectangle (16.2,16.2);
            \draw[gp frame] (0,0) rectangle (16,16);
            \foreach \l/\r/\s in {2/4/1, 4/8/2, 8/12/3} {
                \draw[gp xslab] (\l,0) rectangle (\r,16);
                \node[blue!70!black] at ({(\l+\r)/2},15) {$V_\s$};
            }
            \draw[gp Q] (2,4) rectangle (12,14);
            \node[anchor=north east] at (12,14) {$Q$};
            \foreach \x/\lab in {2/a, 12/b}
                \draw[gp tick] (\x,0) -- (\x,-0.6) node[below, gp ticklabel] {$\lab$};
        \end{tikzpicture}
        \caption{Slabs of $\CR_x([a,b])$.}
    \end{subfigure}\hfill
    \begin{subfigure}[b]{0.32\textwidth}
        \centering
        \begin{tikzpicture}[gp]
            \useasboundingbox (-1.6,-1.8) rectangle (16.2,16.2);
            \draw[gp frame] (0,0) rectangle (16,16);
            \foreach \b/\t/\s in {4/8/1, 8/12/2, 12/14/3} {
                \draw[gp yslab] (0,\b) rectangle (16,\t);
                \node[orange!85!black] at (1,{(\b+\t)/2}) {$L_\s$};
            }
            \draw[gp Q] (2,4) rectangle (12,14);
            \node[anchor=north east] at (12,14) {$Q$};
            \foreach \y/\lab in {4/c, 14/d}
                \draw[gp tick] (0,\y) -- (-0.4,\y) node[left, gp ticklabel] {$\lab$};
        \end{tikzpicture}
        \caption{Slabs of $\CR_y([c,d])$.}
    \end{subfigure}\hfill
    \begin{subfigure}[b]{0.32\textwidth}
        \centering
        \begin{tikzpicture}[gp]
            \useasboundingbox (-1.6,-1.8) rectangle (16.2,16.2);
            \draw[gp frame] (0,0) rectangle (16,16);
            \fill[orange!9] (0,8) rectangle (16,12);
            \fill[orange!28] (2,8) rectangle (12,12);
            \foreach \x in {2,4,8,12} {
                \draw[blue!70!black, densely dotted, line width=0.5pt] (\x,0) -- (\x,16);
                \draw[blue!70!black, line width=0.5pt] (\x,4) -- (\x,14);
            }
            \foreach \y in {4,8,12,14} {
                \draw[orange!85!black, densely dotted, line width=0.5pt] (0,\y) -- (16,\y);
                \draw[orange!85!black, line width=0.5pt] (2,\y) -- (12,\y);
            }
            \draw[gp Q] (2,4) rectangle (12,14);
            \node[anchor=north east] at (12,14) {$Q$};
            \node at (6,10) {$L$};
            \node at (14,10) {$R_u$};
        \end{tikzpicture}
        \caption{The grid $\Grid(Q)$.}
    \end{subfigure}
    \caption{Construction of $\Grid(Q)$ for $Q=[a,b]\times[c,d]$.
    (a)~The canonical nodes $\CR_x([a,b])$ of $T_x$ give vertical slabs $V_1,V_2,V_3$ partitioning $[a,b]\times[1,n]$; their boundaries form $\LW_y$.
    (b)~Symmetrically, $\CR_y([c,d])$ gives horizontal slabs $L_1,L_2,L_3$ partitioning $[1,n]\times[c,d]$; their boundaries form $\LW_x$.
    (c)~The lines of $\LW_x\cup\LW_y$ cut $Q$ into $\Grid(Q)$.
    The shaded row $L$ has the same $y$-projection as the slab $R_u$ of a node $u\in T_y$ (here $R_u=L_2$); columns and $T_x$ behave symmetrically.}
    \label{fig:grid-partition}
\end{figure}

\paragraph{Well-Aligned Rectangles.}
A tuple $(R_1,R_2,\ldots, R_k)$ of axis-aligned rectangles is \emph{well-aligned} if for every $i,j\in[k]$ the $x$-projections of $R_i$ and $R_j$ are either identical or disjoint, and the same holds for the $y$-projections.
We reduce \rppm to \rppm for well-aligned rectangles (\Cref{fig2:arrangement-faces-disjoint}).
Given a tuple $(R_1,\ldots,R_k)$ with $R_i=[a_i,b_i]\times[c_i,d_i]$, extend the sides of each rectangle to full lines, namely $x=a_i$, $x=b_i$, $y=c_i$ and $y=d_i$.
These lines define a linear arrangement.
For each $i$, let $\FF_i$ be the faces of this arrangement contained in $R_i$.
We then have
\begin{align*}
\inds(R_1,\ldots,R_k)=\bigvee_{(F_1,\ldots,F_k)\in\FF_1\times\cdots\times\FF_k}\inds(F_1,\ldots,F_k)\;,
\end{align*}
proving that we can reduce \rppm into $(O(k^2))^k$ well-aligned \rppm instances, since each $\FF_i$ has $O(k^2)$ faces.

\begin{figure}[]
    \centering
    \begin{subfigure}[t]{0.46\textwidth}
        \centering
        \begin{tikzpicture}[scale=0.9,
                every node/.style={font=\small},
                rect/.style={rounded corners=1.5pt, line width=0.8pt},
            ]
            \useasboundingbox (-0.2,-0.7) rectangle (5.2,3.5);

            \fill[blue!6] (0.5,0.5) rectangle (2.0,1.5);
            \draw[rect, blue!70!black] (0.5,0.5) rectangle (2.0,1.5);
            \node[blue!70!black, font=\footnotesize] at (1.25,1.0) {$R_1$};

            \fill[red!6] (2.7,-0.2) rectangle (4.7,1.0);
            \draw[rect, red!60!black] (2.7,-0.2) rectangle (4.7,1.0);
            \node[red!70!black, font=\footnotesize] at (3.7,0.4) {$R_3$};

            \fill[teal!6] (1.0,2.0) rectangle (3.8,3.0);
            \draw[rect, teal!70!black] (1.0,2.0) rectangle (3.8,3.0);
            \node[teal!70!black, font=\footnotesize] at (2.4,2.5) {$R_2$};
        \end{tikzpicture}
        \caption{Three pairwise disjoint rectangles.}
        \label{fig2:arrangement-input}
    \end{subfigure}\hfill
    \begin{subfigure}[t]{0.46\textwidth}
        \centering
        \begin{tikzpicture}[scale=0.9,
                every node/.style={font=\small},
                rect/.style={rounded corners=1.5pt, line width=0.8pt},
                arrline/.style={gray!45, thin, densely dashed},
                facelbl/.style={font=\scriptsize},
            ]
            \useasboundingbox (-0.2,-0.7) rectangle (5.2,3.5);

            \fill[blue!10] (0.5,1.0) rectangle (1.0,1.5);
            \fill[blue!10] (0.5,0.5) rectangle (1.0,1.0);
            \fill[blue!10] (1.0,0.5) rectangle (2.0,1.0);
            \fill[blue!10] (1.0,1.0) rectangle (2.0,1.5);

            \node[facelbl, blue!70!black] at (0.75,1.25) {$F_1$};
            \node[facelbl, blue!70!black] at (0.75,0.75) {$F_2$};
            \node[facelbl, blue!70!black] at (1.5,0.75)  {$F_3$};
            \node[facelbl, blue!70!black] at (1.5,1.25)  {$F_4$};

            \fill[teal!10] (1.0,2.0) rectangle (2.0,3.0);
            \fill[teal!10] (2.0,2.0) rectangle (2.7,3.0);
            \fill[teal!10] (2.7,2.0) rectangle (3.8,3.0);

            \node[facelbl, teal!70!black] at (1.5,2.5) {$F_5$};
            \node[facelbl, teal!70!black] at (2.35,2.5) {$F_6$};
            \node[facelbl, teal!70!black] at (3.25,2.5) {$F_7$};

            \fill[red!8] (2.7,0.5) rectangle (3.8,1.0);
            \fill[red!8] (2.7,-0.2) rectangle (3.8,0.5);
            \fill[red!8] (3.8,0.5) rectangle (4.7,1.0);
            \fill[red!8] (3.8,-0.2) rectangle (4.7,0.5);

            \node[facelbl, red!70!black] at (3.25,0.75) {$F_8$};
            \node[facelbl, red!70!black] at (3.25,0.15) {$F_9$};
            \node[facelbl, red!70!black] at (4.25,0.75) {$F_{10}$};
            \node[facelbl, red!70!black] at (4.25,0.15) {$F_{11}$};

            \foreach \x in {0.5,1.0,2.0,2.7,3.8,4.7} {
                    \draw[arrline] (\x,-0.6) -- (\x,3.4);
                }
            \foreach \y in {-0.2,0.5,1.0,1.5,2.0,3.0} {
                    \draw[arrline] (-0.1,\y) -- (5.1,\y);
                }

            \draw[rect, red!60!black] (2.7,-0.2) rectangle (4.7,1.0);
            \draw[rect, teal!70!black] (1.0,2.0) rectangle (3.8,3.0);
            \draw[rect, blue!70!black] (0.5,0.5) rectangle (2.0,1.5);

        \end{tikzpicture}
        \caption{The arrangement and its faces inside $R_1\cup R_2\cup R_3$.}
        \label{fig2:arrangement-faces}
    \end{subfigure}
    \caption{Extending all sides of $R_1,R_2,R_3$ to lines yields an axis-parallel arrangement whose faces refine each rectangle into nice sub-rectangles.}
    \label{fig2:arrangement-faces-disjoint}
\end{figure}
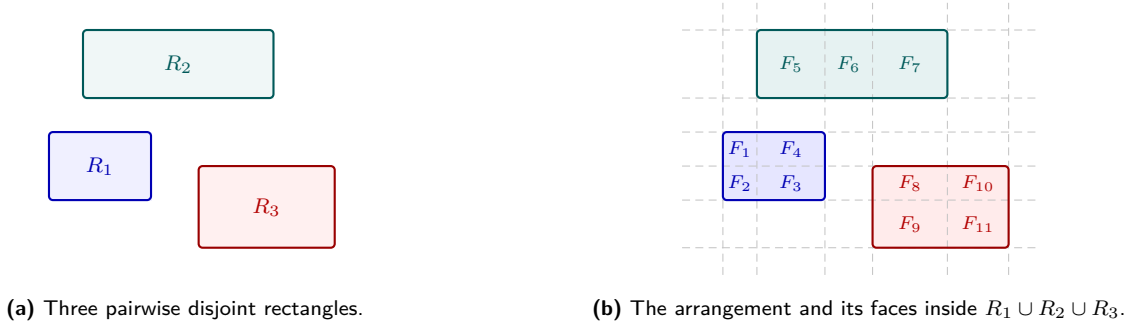

\newcommand{\Sb}{\wpa{SlabQuery}}
\section{Rectangle Queries to Slab Queries}
\label{sec:rect-to-slab}
This section is dedicated to proving \Cref{thm:re ppm}, which reduces multi-rectangle detection queries to multi-slab detection queries for patterns of length at most $3$:
\ThmRSppm* %

This theorem is used to prove both \Cref{thm:optimal emptiness} and \Cref{thm:equiv osd and rppm}, as well as \Cref{cor:bl24}, and therefore might be of independent interest.
As a warm-up, we explain how to reduce $\#_{12}$-rectangle queries to $\#_{12}$-slab queries, which is a simpler version of \Cref{thm:re ppm} that captures the main ideas of the reduction.

\begin{lemma}
    \label{lemma:reduction 12}
    Suppose that we have a data structure \Sb that takes a set of $n$ points $P\subseteq[n]\times [n]$ in the plane and answers $\#_{12}$-multi-slab queries in $\tf(n^q)$ time after $\tf(n^p)$ preprocessing time.
    Then, we can answer $\#_{12}$-multi-rectangle queries in $\tf(n^p)$ preprocessing time and $\tf(n^q)$ query time.
\end{lemma}
Recall that multi-slab query takes two intervals $[c_1,d_1]$ and $[c_2,d_2]$ and computes $\#_{12}([1,n]\times[c_1,d_1],[1,n]\times[c_2,d_2])$ for some fixed interval $[a,b]$.
Note that by using the same data structure on the point set $P'=\set{(y,x)\mid (x,y)\in P}$, we can also answer $\#_{12}$-multi-slab queries of the form $\#_{12}([a_1,b_1]\times[1,n],[a_2,b_2]\times[1,n])$. We note that \cite{Duraj2019EquivalencesBT} already proves that single-slab and multi-slab $\#_{12}$ queries are equivalent up to polylogarithmic factors (single- and multi-slab $\#_{12}$ queries are called range inversion queries or 2-range inversion queries in their language, respectively). Combined with \Cref{lemma:reduction 12}, this means that any data structure for \emph{single-slab} inversion counting with $\tf(n^p)$ preprocessing and $\tf(n^q)$ query time can be used to answer $\#_{12}$-\emph{multi-rectangle} queries within the same trade-off.

\begin{proof}[Proof of \Cref{lemma:reduction 12}]
    The input is a set of $n$ points in the plane $P\subseteq [n]\times [n]$.

    \textbf{Preprocessing.}
    We build two one-dimensional range trees $T_x$ and $T_y$ on the $x$-coordinates and $y$-coordinates of the points in $P$, respectively.
    That is $T_x$ is a binary tree where each node $v$ corresponds to a vertical strip $R_v=[a_v,b_v]\times[1,n]$, where the leaves of $T_x$ correspond to vertical slabs of width $1$ that contains a single point from $P$, and the two children of $v$ correspond to the two vertical strips obtained by splitting $R_v$ into two equal-width vertical strips. The tree $T_y$ is defined similarly.
    Each node $v$ in $T_x$ store the set of points $P_v=P\cap R_v$ that lie in the vertical strip $R_v=[a_v,b_v]\times[1,n]$, and additionally we equip $v$ with a data structure for answering $\#_{12}$-multi-slab queries on $P_v$, which we denote by \Sb.
    This data structure takes two interval $[c_1,d_1]$ and $[c_2,d_2]$ and computes $\#_{12}([a_v,b_v]\times[c_1,d_1],[a_v,b_v]\times[c_2,d_2])$ in $\tf(|P_v|^q)$ time after $\tf(|P_v|^p)$ preprocessing time, as this is a multi-slab query on the universe $[a_v,b_v]\times[c,d]$ with the point set $P_v$.
    An equivalent way to view this data structure is that $v$ also holds a set of points $P_v'=\set{(y,x)\mid (x,y)\in P_v}$. Then, to answer a multi-slab query of the form $\#_{12}([a_v,b_v]\times[c_1,d_1],[a_v,b_v]\times[c_2,d_2])$ on $P_v$, we can answer the multi-slab query $\#_{12}([c_1,d_1]\times[a_v,b_v],[c_2,d_2]\times[a_v,b_v])$ on $P_v'$.
    Strictly speaking, \Sb{} is stated for a point set forming a \emph{permutation} of $[n']\times[n']$ (one point per row and per column), while $P_v$ is a general subset of $[n]\times[n]$. Before instantiating \Sb{} on $P_v$ (and, symmetrically, on $P_u$ below), one should first apply the same order-preserving coordinate ``shrinking'', or compression, that is defined in \Cref{lemma:shrink}: map each coordinate of $P_v$ to its rank within $P_v$, and translate the query boundaries via the functions $\psi_{\textsc{left}},\psi_{\textsc{right}}$ defined there. The compression preserves all $\#_{12}$ counts, so the proof is otherwise unchanged; the same remark applies to the proof of \Cref{thm:re ppm} below, which reuses this preprocessing.

    We also equip every node $u$ in $T_y$ that is associated with a rectangle $R_u=[1,n]\times[c_u,d_u]$, with a data structure \Sb that answers $\#_{12}$-slab queries on $P_u=P\cap R_u$; given an interval $[a_1,b_1]$ and $[a_2,b_2]$, it computes $\#_{12}([a_1,b_1]\times[c_u,d_u],[a_2,b_2]\times[c_u,d_u])$.

    \textbf{Queries.}
    We are now ready to answer $\#_{12}$-rectangle queries.
    Given a query rectangle $Q=[a,b]\times[c,d]$, we compute $\CR_x([a,b])$ and $\CR_y([c,d])$, which are the set of canonical nodes in $T_x$ and $T_y$ that partition $[a,b]$ and $[c,d]$, respectively.
    Recall that a node $v$ is in $\CR_x([a,b])$ if the $x$-projection of $R_v$ is contained in $[a,b]$, and it is maximal with respect to inclusion, meaning that no ancestor $v'$ of $v$ in $T_x$ satisfies that the $x$-projection of $R_{v'}$
    is contained in $[a,b]$. The definition of $\CR_y([c,d])$ is similar.
    Using the vertical lines of rectangles in $\CR_x([a,b])$ and the horizontal lines of rectangles in $\CR_y([c,d])$, we can partition $Q$ into a grid of $O(\log^2 n)$ rectangles, which we denote by $\Grid(Q)$.
    A key property of this grid partition is that for every row $L$ in the grid, there is a vertex $v\in \CR_y([c,d])$ such that the $y$-projection of $R_v$ is the same as the $y$-projection of $L$, which is crucial for our reduction.
    To answer $\#_{12}(Q)$, we use
    \begin{align*}
        \#_{12}(Q) = \sum_{R_1,R_2\in \Grid(Q)} \#_{12}(R_1,R_2) \;.
    \end{align*}

    We consider two types of pairs $(R_1,R_2)$, based on whether $R_1$ and $R_2$ are from the same row in $\Grid(Q)$, or the same column.
    If $R_1$ and $R_2$ are not from the same row or the same column, then computing $\#_{12}(R_1,R_2)$ is easy; if the $x$-projections of $R_1$ is to the left of the $x$-projection of $R_2$, and the $y$-projection of $R_1$ is below the $y$-projection of $R_2$, then $\#_{12}(R_1,R_2)=|R_1|\cdot|R_2|$ as every pair of points from $R_1$ and $R_2$ form an $12$ pattern.
    Otherwise, $\#_{12}(R_1,R_2)=0$.
    Since checking projections of two rectangles, as well as counting the number of points inside them takes $\tf(1)$ time, we can compute $\#_{12}(R_1,R_2)$ in $\tf(1)$ time for every pair of rectangles $R_1,R_2$ that are not from the same row or the same column.

    ~\\We are left with pairs $(R_1,R_2)$ that are from the same row or the same column. Assume without loss of generality that $R_1$ and $R_2$ are from the same row, and let $R_i=[a_i,b_i]\times[c,d]$ for $i=1,2$.
    Let $v$ be the vertex in $\CR_y([c,d])$ such that the $y$-projection of $R_v$ is $[c,d]$, where such vertex exists by the property of the grid partition $\Grid(Q)$.
    Then, to compute $\#_{12}(R_1,R_2)$, we can use the data structure \Sb stored in $v$ to compute $\#_{12}([a_1,b_1]\times[c,d],[a_2,b_2]\times[c,d])$ as this is a multi-slab query on $P_v=P\cap R_v$.

    This shows that we reduced $\#_{12}$-rectangle queries to $\#_{12}$-multi-slab queries, making $\tf(1)$ pair queries in total, and therefore at most $\tf(1)$ multi-slab queries, as desired.
\end{proof}

\begin{remark}
    The exact same proof shows that $\#_{\textsc{EQ}}$-rectangle queries can be reduced to $\#_{\textsc{EQ}}$-slab queries, where $\#_{\textsc{EQ}}(Q)$ counts the number of pairs of points $p_1\in Q_1$ and $p_2\in Q_2$ such that $p_1$ and $p_2$ have the same $y$-coordinate.
\end{remark}

We are now ready to explain how to get faster $5$-pattern counting \Cref{cor:bl24} by combining \Cref{lemma:reduction 12} with the algorithm of~\cite{BL2024counting} for counting $5$-patterns, that uses rectangle $\#_{12}$-queries as a subroutine, which we replace with the data structure \Sb implemented in~\cite{Duraj2019EquivalencesBT} with better preprocessing query tradeoffs:
\begin{theorem}[{\cite{Duraj2019EquivalencesBT}}]
    \label{thm:12 slab}
    There is an algorithm for answering $z$ offline $\#_{12}$-slab queries in time:
    \begin{align*}
        T_z\triangleq\begin{cases}
                         \tf(n z^{\frac{\omega-1}{\omega+1}}),                       & z \le n\;, \\[2pt]
                         \tf(n^{\frac{2\omega-2}{\omega+1}} z^{\frac{2}{\omega+1}}), & z > n\;,
                     \end{cases}
    \end{align*}
    where $\omega$ is the exponent of matrix multiplication.
\end{theorem}
We only use the second case of the above theorem.
\begin{proof}[Proof of \Cref{cor:bl24}]
    We plug our rectangle-to-slab reduction into the \cite{BL2024counting} tradeoff.
    For every choice of parameter $m$, \cite{BL2024counting} give a running time of $\tf(\frac{n^2}{m})$ plus the cost of answering $nm^2$ offline $\#_{12}$-rectangle queries.
    Using their rectangle-query data structure, the total running time of the algorithm is
    \begin{align*}
        \tf(\frac{n^2}{m}+nm^2\cdot n^{1/4}+n^{7/4})\;,
    \end{align*}
    and setting $m=n^{1/4}$ gives $\tf(n^{7/4})$.

    We replace this rectangle-query step by slab queries using \Cref{thm:re ppm}.
    Since $z=nm^2>n$, the second case of \Cref{thm:12 slab} gives total time
    \[
        \tf(n^{\frac{2\omega-2}{\omega+1}}(nm^2)^{\frac{2}{\omega+1}})
        = \tf(n^{\frac{2\omega}{\omega+1}}m^{\frac{4}{\omega+1}}).
    \]
    Therefore the total running time becomes
    $\tf(\frac{n^2}{m}+n^{\frac{2\omega}{\omega+1}}m^{\frac{4}{\omega+1}})$.
    It remains to optimize $m$.
    Write $m=n^\alpha$.
    The two exponents are $2-\alpha$ and $\frac{2\omega}{\omega+1}+\frac{4\alpha}{\omega+1}$.
    Balancing them gives
    \[
        2-\alpha=\frac{2\omega}{\omega+1}+\frac{4\alpha}{\omega+1},
    \]
    hence $\alpha=\frac{2}{\omega+5}$. Thus we choose $m=n^{2/(\omega+5)}$,
    and obtain the improved running time $\tf(n^{\,2-\frac{2}{\omega+5}})$.
\end{proof}

We move on to prove \Cref{thm:re ppm}, which reduces multi-rectangle detection queries to multi-slab detection queries for patterns of length at most $3$.
The reduction is very similar to the one in the previous subsection.

We need one additional claim that we defer to the appendix.
\begin{restatable}{claim}{ClaimSev}\label{claim:17}
    Consider a multi-rectangle query $\indd(Q_1,Q_2,Q_3)$ where $Q_1,Q_2,Q_3$ are axis-parallel rectangles, and let $\Grid(Q_1,Q_2,Q_3)$ be a grid partition that is a refinement of $\Grid(Q_i)$ for every $i=1,2,3$.
    Let $R_1,R_2,R_3\in\Grid(Q_1,Q_2,Q_3)$. If $(R_1,R_2,R_3)$ are not all from the same row or the same column then $\indd(R_1,R_2,R_3)$ can be computed in $\tf(1)$ time after $\tf(n)$ preprocessing time.

    The same is true if $(R_1,R_2,R_3)$ are from the same row (or column) but consist of exactly three distinct rectangles.
\end{restatable}
The proof of \Cref{claim:17} is deferred to the appendix, and it is based on the fact that the primal graph of the boolean CSP defined by $\indd(R_1,R_2,R_3)$ is acyclic, and therefore can be solved in $\tf(1)$ time using range tree point queries.

\begin{proof}[Proof of \Cref{thm:re ppm}]
    The preprocessing of this step is identical to the preprocessing step in the proof of \Cref{lemma:reduction 12}; build two range trees $T_x$ and $T_y$ on the $x$-coordinates and $y$-coordinates of the points in $P$, respectively, and equip every node in these trees with a data structure for answering multi-slab queries on the points that lie in the slab associated with the node, as described in the proof of \Cref{lemma:reduction 12}.

    \textbf{Queries.}
    Given a query triple of rectangles $(Q_1,Q_2,Q_3)$ which we can assume that is well-aligned as explained in \Cref{ssec:range-trees}, meaning that for every pair of rectangles the $x$-projections are either identical or disjoint, and the same holds for their $y$-projections.
    Let $Q_i=[a_i,b_i]\times[c_i,d_i]$ for $i=1,2,3$.
    We compute a grid partition for all three rectangles together as follows.
    First compute $C_x^i\triangleq\CR_x([a_i,b_i])$ and $C_y^i\triangleq\CR_y([c_i,d_i])$ for every $i=1,2,3$, which are the set of canonical nodes in $T_x$ and $T_y$ that partition $[a_i,b_i]$ and $[c_i,d_i]$, respectively.
    Since the triplet is well-aligned, every pair of nodes in $C_x^1\cup C_x^2\cup C_x^3$ are associated with disjoint slabs, and the same holds for $C_y^1\cup C_y^2\cup C_y^3$.

    Recall that the partition $\Grid(Q_1)$ is obtained by taking the vertical lines of rectangles in $\CR_x([a_1,b_1])$ and the horizontal lines of rectangles in $\CR_y([c_1,d_1])$, and considering the linear arrangement defined by these lines. To get a grid partition $\Grid(Q_1,Q_2,Q_3)$ that is a refinement of $\Grid(Q_1),\Grid(Q_2),\Grid(Q_3)$, we take the union of the lines that define the linear arrangements of $\Grid(Q_1),\Grid(Q_2),\Grid(Q_3)$, and consider the linear arrangement defined by these lines.
    This gives us a grid partition $\Grid(Q_1,Q_2,Q_3)$.
    We use $\FF_i$ to denote the faces of the linear arrangement that are contained in $Q_i$.
    Crucially, the property that every row (or column) of the grid has a node $v\in T_y$ with associated rectangle $R_v$ that has the same $y$-projection as the row still holds for $\Grid(Q_1,Q_2,Q_3)$, as $\Grid(Q_1,Q_2,Q_3)$ is a refinement of $\Grid(Q_i)$ for every $i=1,2,3$.
    A row that is not associated with a node in $\CR_y^1\cup \CR_y^2\cup \CR_y^3$ implies that the slabs of two nodes are not disjoint, contradicting the well-alignment of the triplet of rectangles.
    The second preserved property, which is preserved for any partitions $\FF_i$ of $Q_i$ is that
    \begin{align*}
        \indd(Q_1,Q_2,Q_3) = \bigvee_{(F_1,F_2,F_3)\in \FF_1\times \FF_2\times \FF_3} \indd(F_1,F_2,F_3)\;.
    \end{align*}

    By the first property, we can reduce triplets of rectangles that are in the same row or the same column to multi-slab queries, and use the data structure stored in the node associated with the row or the column to answer these multi-slab queries.
    We are left with dealing with triplets of faces that are not from the same row or the same column. By \Cref{claim:17}, we can answer $\indd(F_1,F_2,F_3)$ in $\tf(1)$ time for every such triplet, which completes the proof.
\end{proof}
\renewcommand{\CR}[1]{\mathcal{C}(#1)}

\section{Single-Rectangle Queries for $132$ Patterns}
\label{sec:optimal emptiness}
This section is devoted to proving \Cref{thm:optimal emptiness}:
\ThmOptEmpty* %
The theorem can be extended to any pattern $\sigma$ of length $3$.
One key ingredient in the proof is a data structure for answering $132$-slab queries, which we prove in \Cref{ssec:132 single interval}:
\begin{restatable}[$132$ Slab Queries]{theorem}{ThmIntervalACB}
    \label{thm:132 interval queries}
    There is a data-structure \Sb for $132$-slab queries with $\tf(n)$ preprocessing time and $\tf(1)$ query time.
\end{restatable}

We explain how to use this data structure to solve $132$-rectangle queries.
Our starting point is as before to build two range trees $T_x$ and $T_y$ on the $x$-coordinates and $y$-coordinates of the points in $P$, respectively, and equip every node $v$ in $T_x$ (and symmetrically for $T_y$) with a data structure \Sb for answering \emph{single}-slab queries on the points $P_v$ in the slab $R_v=[a,b]\times[1,n]$.
That is $\Sb$ takes an interval $[c,d]$ and answers $\indd([a,b]\times[c,d])$.
Formally:
\begin{restatable}[\Sb]{theorem}{ThmSubQB}\label{thm3:subquery}
    There exists a data-structure \Sb that takes $\tf(n)$ preprocessing time.
    It supports the following query in $\tf(1)$ time:
    \begin{enumerate}[leftmargin=*]
        \item Given a vertex $v\in \Th$, with $R_v=[1,n]\times[c_v,d_v]$, and two values $a,b\in[n]$, it computes $\indd([a,b]\times [c_v,d_v])$.
        \item Given a vertex $v\in \Tv$, with $R_v=[a_v,b_v]\times[1,n]$, and two values $c,d\in[n]$, it computes $\indd([a_v,b_v]\times [c,d])$.
    \end{enumerate}
\end{restatable}

We are ready to complete the proof of \Cref{thm:optimal emptiness} using \Cref{thm3:subquery}.
\begin{proof}[Proof of \Cref{thm:optimal emptiness} using \Cref{thm3:subquery}]
    To answer queries of the form $\indd(Q)$ for axis parallel rectangle $Q=[a,b]\times[c,d]$, we partition $Q$ into a grid $\Grid(Q)$ using the trees $T_x$ and $T_y$. This is a special case of computing $\Grid(Q_1,Q_2,Q_3)$ for three rectangles $Q_1,Q_2,Q_3$ that are all equal to $Q$, which we explained in the previous section.
    We can compute $\indd(Q)$ by computing $\bigvee_{(R_1,R_2,R_3)\in \Grid(Q)^3}\indd(R_1,R_2,R_3)$, and therefore we reduce the single rectangle query $Q$ to $O(\log^6 n)$ triplet queries of the form $\indd(R_1,R_2,R_3)$ where $R_i\in \Grid(Q)$ for $i=1,2,3$.

    For each such triplet $(R_1,R_2,R_3)$, we consider two cases based on whether the (not necessarily distinct) rectangles $R_1,R_2,R_3$ are from the same row or the same column in $\Grid(Q)$, or not.
    If they are not from the same row or the same column, then we can answer $\indd(R_1,R_2,R_3)$ in $\tf(1)$ time using \Cref{claim:17}.

    We are left with triplets that fall in the same row (or column) of $\Grid(Q)$, which we call \emph{in-strip} queries.
    Suppose that $(R_1,R_2,R_3)$ is an in-strip query in the same row, meaning that $R_1,R_2,R_3$ have the same $y$-projection (the rectangles do not have to be distinct).
    Let $R_i=[a_i,b_i]\times[c,d]$ for $i=1,2,3$, where $a_1\leq a_2\leq a_3$, and $b_1\leq b_2\leq b_3$.
    Instead of computing $\indd(R_1,R_2,R_3)$, we compute $\indd([a_1,b_3]\times[c,d])$, we use $S=[a_1,b_3]\times[c,d]$ to denote the new query rectangle.
    We show that the new query can be answered by a node $v$ in the range tree $T_y$ that is associated with a rectangle $R_v$ that defines the row that contains $R_1,R_2,R_3$.
    Moreover, $R_v=[1,n]\times[c,d]$ meaning it has the same $y$-projection as $R_1,R_2,R_3$, which is crucial for our reduction.
    We claim that $v$ can answer $\indd(S)$ because this is an interval query on its point set $P_v=P\cap R_v$, and therefore we can use \Cref{thm3:subquery} to answer it in $\tf(1)$ time.
\end{proof}

\subsection{In-Strip Triplets}\label{ssec:in-strip}
\renewcommand{\phi}{\varphi}
\providecommand{\psia}{\psi_{\textsc{left}}}
\providecommand{\psib}{\psi_{\textsc{right}}}

In this subsection we prove \Cref{thm3:subquery}:
\ThmSubQB* %

We restrict our attention to node $v$ in the tree $\Th$ with rectangle $R_v=[1,n]\times[c_v,d_v]$ containing point set $P_v=P\cap R_v$, of size $m_v=|P_v|=d_v-c_v+1$.
Dealing with nodes in $\Tv$ is symmetric and follows the same approach.
Our approach consists of three key steps.
First, we apply coordinate compression: an order-preserving transformation $\rank_v^x,\rank_v^y:[n]\to[m_v]$, that maps each point to the ranks of its coordinates.
This reduction is crucial for applying \Cref{thm:132 interval queries}, which works in the permutation setting of $m_v$ points, one in each row and each column of $[m_v]\times[m_v]$.
We define a compressed set of points $P'_v$ by
\begin{align*}
    P'_v=\set{(\rank_v^x(p.x),\rank_v^y(p.y))\mid p\in P_v}\;.
\end{align*}
Second, we maintain two functions $\psia,\psib:[n]\to[m_v]$, which allows us to transform a query rectangle $Q=[a,b]\times[c_v,d_v]$ into a new rectangle $Q'=[\psia(a),\psib(b)]\times[1,m_v]$ such that
\begin{align*}
    \indd(Q;P_v)=\indd(Q';P'_v)\;,
\end{align*}
where we use $\indd(Q;P_v)$ to denote whether the set of points in $P_v$ that lie in the rectangle $Q$ contains a $132$-pattern, and $\indd(Q';P'_v)$ to denote whether the set of points in $P'_v$ that lie in the rectangle $Q'$ contains a $132$-pattern.
That is, we change the notation $\indd(Q)$ to either $\indd(Q;P)$ or $\indd(Q;P_v)$ to explicitly indicate the point set we are referring to.
The third and last step, is to explain how to answer queries of the form $\indd(Q';P'_v)$ in $\To(1)$ time for any $a',b'\in[m_v]$, after preprocessing $P'_v$ in $\tf(m_v)$ time, which uses \Cref{thm:132 interval queries}.

~\\The following lemma explains how achieve the first two steps.
\begin{lemma}\label{lemma:shrink}
    There is a data structure \shrink with the following guarantees:
    \begin{itemize}
        \item \textbf{Preprocessing:} Given a set of $m$ points $P$ in $[n]\times[n]$, \shrink produces a new set of points $P'$ in $[m]\times[m]$, which is order-preserving. Preprocessing time is $\tf(m)$.
        \item \textbf{Transforming a Query:} Given a query rectangle $Q=[a,b]\times[1,n]$, \shrink computes a new rectangle $Q'=[a',b']\times[1,m]$, such that $\indd(Q;P)=\indd(Q';P')$. The transformation takes $\tf(1)$ time.
    \end{itemize}
\end{lemma}

To achieve the third step we use a data structure for answering $132$-\emph{slab}-queries \Cref{thm:132 interval queries} which we prove in the next subsection (see \Cref{ssec:132 single interval}).
We explain how to prove \Cref{thm3:subquery} using \Cref{lemma:shrink} and \Cref{thm:132 interval queries}.
\begin{proof}[Proof of \Cref{thm3:subquery} using \Cref{lemma:shrink} and \Cref{thm:132 interval queries}]
    We prove the theorem for nodes in $\Th$; the proof for nodes in $\Tv$ is symmetric.
    We first explain the preprocessing phase. Consider a node $v\in \Th$ with rectangle $R_v=[1,n]\times[c_v,d_v]$ containing point set $P_v=P\cap R_v$, of size $m_v=|P_v|=d_v-c_v+1$.
    We apply \Cref{lemma:shrink} to the set of points $P_v$ to get a new set of points $P'_v$ in $[m_v]\times[m_v]$, and a data-structure that transforms a query rectangle $Q=[a,b]\times[1,n]$ into a new rectangle $Q'=[a',b']\times[1,m_v]$ such that $\indd(Q;P_v)=\indd(Q';P'_v)$.

    We also apply \Cref{thm:132 interval queries} to $P'_v$ to build a data structure \intQuery that answers queries of the form $\indd([a',b']\times[1,m_v];P'_v)$ in $\To(1)$ time for any $a',b'\in[m_v]$.
    This completes the preprocessing phase, which takes $\To(m_v)$ time by the complexities of \Cref{lemma:shrink} and \Cref{thm:132 interval queries}.

    To answer queries, given $a,b\in[n]$, we first use the data structure from \Cref{lemma:shrink} to compute the new rectangle $Q'=[a',b']\times[1,m_v]$ such that $\indd([a,b]\times[1,n];P_v)=\indd(Q';P'_v)$, and then we use the data structure \intQuery from \Cref{thm:132 interval queries} to compute $\indd(Q';P'_v)$ in $\To(1)$ time, and return it as the answer to the query.
    The correctness follows by the correctness of both \Cref{lemma:shrink} and \Cref{thm:132 interval queries}.
\end{proof}

We prove \Cref{lemma:shrink}. The main idea is to replace each $x$ or $y$ coordinate of a point by its rank among the $x$ or $y$ coordinates of the points in $P_v$, respectively, and to applying basically the same transformation to the query rectangle, even if the endpoints of the query rectangle are not in the set of coordinates of the points in $P_v$.
\begin{proof}[Proof of \Cref{lemma:shrink}]
    \newcommand{\fleq}{\rank_x}
    \newcommand{\fgeq}{f_{\geq}}
    Let $R_v=[1,n]\times[c_v,d_v]$ denote the strip of semi-canonical node $v\in\Th$, with $m_v=d_v-c_v+1$ points in $R_v\cap P$.

    \paragraph{Preprocessing.}
    We define two functions $\rank_x,\rank_y:[n]\to[m_v]$.
    The $y$-transformation is straightforward;
    for $a\in [n]$ define $\rank_y(a)\triangleq a-c_v+1$, which is a simple shift of the $y$-coordinates.
    For the $x$-transformation, we ``shrink'' the $x$-coordinates of the points as follows.
    Let $X_v$ denote the set of $x$-coordinates of the points in $R_v\cap P$.
    We define two functions $\fleq$ as follows:
    \begin{align*}
        \forall a\in [n]\quad  \fleq(a)\triangleq |\set{x\in X_v\mid x\leq a}|\;.
    \end{align*}
    In words, $\fleq(a)$ counts how many $x$-coordinates in $X_v$ are at most $a$.
    The new set of points $P'_v$ is defined as follows:
    \begin{align*}
        P'_v=\set{(\fleq(p.x),\rank_y(p.y))\mid p\in R_v\cap P}\;.
    \end{align*}

    \paragraph{Query Transformation.}
    We explain how to transform a query rectangle $Q=[a,b]\times[1,n]$ into a new rectangle $Q'=[a',b']\times[1,m_v]$ such that $\indd(Q;P)=\indd(Q';P')$.
    Let $\chi_v(a)$ be the indicator function of $X_v$, i.e., $\chi_v(a)=1$ if $a\in X_v$ and $\chi_v(a)=0$ otherwise.
    Then, the transformation is as follows:
    \begin{align*}
        \psia(a) & = \fleq(a)+1 - \chi_v(a)\;,  \\
        \psib(a) & = \fleq(a)\;.
    \end{align*}

    In words, if $a$ is in $X_v$, then $\psia(a)=\psib(a)=\fleq(a)$, and otherwise $\psia(a)=\fleq(a)+1$, and $\psib(a)=\fleq(a)$.
    This completes the description of the query transformation.
    The functions $\psia,\psib$ can be stored using $O(m_v\log n)$ space, and computed in $O(\log n)$ time.

    \paragraph{Correctness.}
    Let $Q=[a,b]\times[1,n]$.
    Let $(x_i,y_i)$ be the $i$-th point in $P$, and define a matching point $(z_i,w_i)$ in $P'$ by
    $(z_i,w_i)\triangleq(\fleq(x_i),\rank_y(y_i))$, for every $i\in[m_v]$.
    We prove that $\indd(Q;P)=\indd(Q';P')$.
    Assume that $(x_1,y_1),(x_2,y_2),(x_3,y_3)\in P\cap Q$ are three points that form a $132$-pattern.
    Then,
    \begin{enumerate}
        \item $a\le x_1 \iff \psia(a)\le z_1$\quad and\quad $x_3\le b \iff z_3\le \psib(b)=\fleq(b)$ %
        \item $c_v\leq y_1<y_3<y_2\leq d_v \iff 1\leq w_1<w_3<w_2\leq m_v$.
    \end{enumerate}
    Therefore, $(x_1,y_1),(x_2,y_2),(x_3,y_3)$ forms a $132$-pattern in $P\cap Q$ if and only if $(z_1,w_1),(z_2,w_2),(z_3,w_3)$ form a $132$-pattern in $P'\cap Q'$, which completes the proof.
    \end{proof}

\subsection{Slab Queries ($\indd(\slab)$)}
\label{ssec:132 single interval}
The main goal of this subsection is to prove \Cref{thm:132 interval queries}, which we restate here for convenience:
\ThmIntervalACB* %
The theorem can be extended to any pattern $\sigma$ of length $3$.
We consider a permutation $\pi:[m]\to[m]$, and the set of points $P=\{(i,\pi(i))\mid i\in[m]\}$ it defines.
We also use $P[i]=\pi(i)$, i.e. we treat $P$ as an array.
For an interval $I=[a,b]$ we use $\indd(I)$ to denote the indicator function of whether $P\cap (I\times[1,m])$ contains a $132$ pattern.
We also use $P[a:b]$ to denote the subset of points in $P$ whose $x$-coordinate is in $[a,b]$, i.e., $P[a:b]=P\cap ([a,b]\times[1,m])$.

~\\ We overview the main ideas in the proof of \Cref{thm:132 interval queries}.
We build two tables $T$ and $\hat{T}$ with $m$ entries each.
The table $T$ is defined as follows.
For every $a\in[m]$ the cell $T[a]=b$, where $b$ is the smallest index such that $P[a:b]$ contains a $132$ pattern, and if there is no such index, we set $T[a]=\infty$.
To construct $T$, we use the table $\hat{T}$.
The table $\hat{T}$ is defined by setting for every $a\in [m]$, the cell $\hat{T}[a]$ to be the minimum index $b\in[m]$ such that $\indd([a,b])=1$, with the additional restriction that the first element in the $132$ pattern is $P[a]$.
We show that we can construct $\hat{T}$ in $\tf(m)$ time, use it to compute $T$ in $O(m\log m)$ time, and that using $T$ we can answer queries in $\tf(1)$ time.

~\\The following lemma states the running time to compute $T$:
\begin{lemma}
    \label{lemma2:132 interval queries}
    There exists an algorithm to compute $T$ in time $\tf(m)$.
\end{lemma}
\begin{proof}[Proof \Cref{thm:132 interval queries} using \Cref{lemma2:132 interval queries}]
    We can answer a query $[a,b]$ using $T$ because $\indd([a,b])=1$ if and only if $T[a]\leq b$.
    The table $T$ is defined by $T[a]=\min_{b\in[m]}\{\indd([a,b])=1\}$, for every $a\in[m]$.
\end{proof}

The following definition is a formal definition of $\hat{T}$.
\begin{definition}[$\hat{T}$]
    For every $a\in [m]$, the cell $\hat{T}[a]$ is the minimum index $b\in[m]$ such that $\indd([a,b])=1$, with the additional restriction that the first element in the $132$ pattern is $P[a]$.
    Namely, there are three indices $a=x<y<z\leq b$ such that $P[x]<P[z]<P[y]$.
\end{definition}

The following lemma states the running time to compute $\hat{T}$:
\begin{lemma}
    \label{lemma2:hat T}
    There exists an algorithm to compute $\hat{T}$ in time $\tf(m)$.
\end{lemma}
We prove \Cref{lemma2:132 interval queries} using \Cref{lemma2:hat T}.
\begin{proof}[Proof of \Cref{lemma2:132 interval queries} using \Cref{lemma2:hat T}]
    \newcommand{\st}{T^\star}
    We use the following recursive formula to compute $T$ from $\hat{T}$:
    \begin{align*}
        T[a] & = \min(\hat{T}[a],T[a+1])\;,
    \end{align*}
    where $T[m+1]=\infty$.
    Clearly, this allows to compute $T$ from $\hat{T}$ in $O(m)$ time, by iterating over $a$ from $m$ down to $1$, and in each step comparing two values.
    We are left with proving the correctness of the recursive formula.
    Consider two tables:
    \begin{itemize}
        \item $T$ the correct table, i.e., $T[a]=\min_{b\in[m]}\{\indd([a,b])=1\}$, and
        \item
              $\st$ the table defined by $\st[a]=\min(\hat{T}[a],\st[a+1])$.
    \end{itemize}
    We show that $T=T^\star$.

    $(\Longrightarrow)$ Fix some $a\in[m]$ and suppose that
    $T[a]=b$. We show that $\st[a]\leq b$.
    There are three indices $a\leq x<y<z\leq b$ such that $P[x]<P[z]<P[y]$.
    Thus, $\hat{T}[x]\leq b$, and $\st[x]\leq b$.
    By the recursive formula, $\st[a]\leq \st[x]$, and therefore $\st[a]\leq b$.

    $(\Longleftarrow)$ Fix some $a\in[m]$. We show that $T[a]\leq \st[a]=b$.
    By the recursive formula, it holds that $\st[a]=\hat{T}[a]$ or $\st[a]=\st[a+1]$.
    In the former case, we have $\st[a]=\hat{T}[a]\geq T[a]$, where the inequality is by the definition of $T$ and $\hat{T}$.
    In the latter case, we have $\st[a]=\st[a+1]\geq T[a+1]\geq T[a]$. The first inequality can be proved by induction.
    This completes the proof.
\end{proof}

We are left with computing $\hat{T}$.
We need a new tool that allow us to compute $\ind_{21}(R)$ for any axis parallel rectangle $R$ in $\tf(1)$ time after $\tf(m)$ preprocessing time.
This was achieved by \cite{ben2024approximate} and also highlighted in the introduction.
Using this tool, we can prove \Cref{lemma2:hat T}.

\begin{proof}[Proof of \Cref{lemma2:hat T}]
    We define a new table $C$ with $m$ entries.
    We compute it in time $\tf(m)$, and prove that $\hat{T}$ is exactly $C$.

    The table $C$ is defined as follows.
    For every $a\in[m]$ the entry $C[a]$ is the minimum value $b\in[m]$ such that the rectangle $R_{ab}\triangleq[a+1,b]\times[P[a]+1,m]$ contains a $21$ pattern.
    Note that for every $b_1<b_2$ it holds that $R_{ab_1}\subseteq R_{ab_2}$, and therefore $\ind_{21}(R_{ab_1})\leq \ind_{21}(R_{ab_2})$.
    Thus, to compute $C[a]$ we use binary search that requires $O(\log m)$ queries of the form $\ind_{21}(R_{ab})$ for different values of $b$.
    Using the data-structure from \cite{ben2024approximate}, each query takes $\tf(1)$ time, resulting in computing $C$ in $\tf(m)$ time.
    We show that $C$ is exactly $\hat{T}$ by proving that $\hat{T}[a]=C[a]$ for every $a\in[m]$.

    ($\Longleftarrow$)
    Suppose that $\hat{T}[a]=b$. We prove that $C[a]\leq b$.
    If $\hat{T}[a]=b$, then there are three indices $a=x<y<z\leq b$ such that $P[x]<P[z]<P[y]$.
    Therefore, the rectangle $R_{ab}=[a+1,b]\times[P[a]+1,m]$ contains the two points $(y,P[y])$ and $(z,P[z])$, which form a $21$ pattern, and therefore $\ind_{21}(R_{ab})=1$, and thus $C[a]\leq b$.

    ($\Longrightarrow$)
    Suppose that $C[a]=b$. We prove that $\hat{T}[a]\leq b$.
    If $C[a]=b$, then the rectangle $R_{ab}=[a+1,b]\times[P[a]+1,m]$ contains two points $(y,P[y])$ and $(z,P[z])$ that form a $21$ pattern, meaning that $y<z$ and $P[y]>P[z]$. Then the indices $a,y,z$ form a $132$ pattern in $P$, and therefore $\hat{T}[a]\leq b$.

    \medskip
    \noindent This completes the proof that $\hat{T}[a]=C[a]$ for every $a\in[m]$, and therefore we can compute $\hat{T}$ in $\tf(m)$ time.
\end{proof}

\renewcommand{\CR}{\mathcal{C}}
\section{Hardness of \rppm}
\label{sec:hardness}
In this section we relate \rppm to online set disjointness (\Cref{prob:osd a}).
For $k=2$, we prove that multi-rectangle $132$ detection is equivalent to \prba.
For every $k\ge 3$, we prove only one direction: \prbb reduces to multi-slab detection of the pattern $(1,k+1,k,k-1,\ldots,2)$ of length $k+1$.

\probOSDa*    %

Our first result is showing that \rppm for the pattern $132$ is equivalent to \prba:
\ThmOSDEquiv* %

To prove this theorem, we show a chain of reductions, as illustrated in \Cref{fig:hardness-equivalence-chain}. We also define a succinct variant of \prba, denoted by \prbd, which we show is equivalent to \prba:
\begin{problem}[\prbd]
\label{def:preprocessing-set-disjointness k succinct}
Preprocess the family of $m$ sets $\FF=(S_1,\ldots,S_m)$, where $S_i=\set{[a_1,b_1],\ldots,[a_{h},b_h]}$ is a set of closed intervals with $|S_i|=h$, and $N=\sum_{S\in\FF}|S|$ is the total input size.
Given a query $(i,j)$, the task is to determine whether there exists $[a,b]\in S_i$ and $[c,d]\in S_j$ such that $[c,d]\subseteq[a,b]$.
\end{problem}

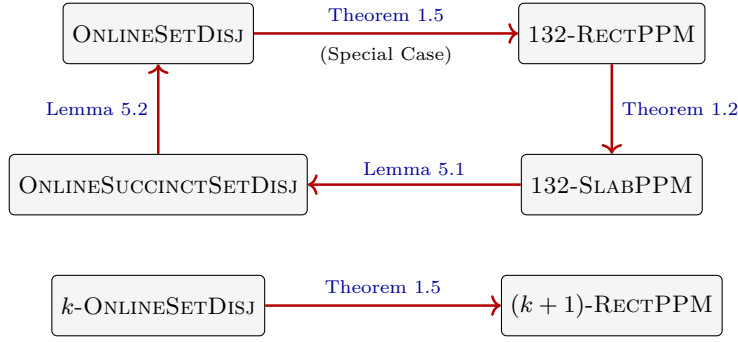
\begin{figure}[h!]
    \centering
    \begin{tikzpicture}[scale=1,
            prob/.style={draw, rounded corners=2pt, minimum width=20mm, minimum height=8mm, fill=gray!8, font=\small},
            redarr/.style={<->, draw=red!70!black, line width=0.95pt}
        ]
        \node[prob] (prba) at (0,0)  {\prba};
        \node[prob] (rppm) at (6,0)  {\rppmc};
        \node[prob] (sppm) at (6,-2)  {\sppmc};
        \node[prob] (prbd) at (0,-2) {\prbd};

        \draw[->, draw=red!70!black, line width=0.95pt] (prba) -- node[midway, above=1pt, font=\scriptsize] {\Cref{thm:hardness interval k-1}} node[midway, below=1pt, font=\scriptsize] {(Special Case)} (rppm);
        \draw[->, draw=red!70!black, line width=0.95pt] (rppm) -- node[midway, ,right=0pt, font=\scriptsize] {\Cref{thm:re ppm}} (sppm);
        \draw[->, draw=red!70!black, line width=0.95pt] (sppm) -- node[midway, above=0pt, font=\scriptsize] {\Cref{lemma:red interval to prbd}} (prbd);
        \draw[->,draw=red!70!black, line width=0.95pt] (prbd) to node[midway, left=0pt, font=\scriptsize] {\Cref{lemma:eq prba and prbd}} (prba);

        \node[prob] (prbb) at  (0,-3.6)  {\prbb};
        \node[prob] (rppmk) at (6,-3.6)  {$(k+1)$-\rppm};
        \draw[->, draw=red!70!black, line width=0.95pt] (prbb) -- node[midway, above=1pt, font=\scriptsize] {\Cref{thm:hardness interval k-1}} (rppmk);

    \end{tikzpicture}
    \caption{Equivalence chain used in this subsection.}
    \label{fig:hardness-equivalence-chain}
\end{figure}

The second result is showing that \rppm for patterns of length $k+1$ is at least as hard as \prbb for $k$ sets.

\ThmHardLongInter* %
A special case of \Cref{thm:hardness interval k-1}, allows us to reduce \prba to $\ind_{132}(Q_1,Q_1,Q_2)$.

\paragraph{Roadmap.}
We first reduce \rppmc to \prbd, and then show that \prbd and \prba are equivalent up to the same preprocessing/query tradeoff.
Together with the reduction from \prba to $132$-\rppm (the case $k=2$ of \Cref{thm:hardness interval k-1}), this proves \Cref{thm:equiv osd and rppm}.
We then prove the general reduction from \prbb to $(k+1)$-\rppm, showing that $(k+1)$-\rppm is at least as hard as \prbb.

\subsection{From \sppmc to \prbd}
\newcommand{\intQ}{\wpa{IntervalQuery}}
In this subsection we prove a reduction from \sppmc to \prbd stated in the following lemma:
\begin{lemma}
    \label{lemma:red interval to prbd}
    \sppmc is reducible to \prbd.
\end{lemma}

\begin{proof}[Proof of \Cref{lemma:red interval to prbd}]
    \newcommand{\JJ}{\mathcal{J}}
    We show how to reduce slab queries to \prbd.
    We only prove how to reduce vertical slab $[a,b]\times[1,n]$ queries. The reduction for horizontal slab queries follows along the exact same lines.

    Let $P=\{(i,\pi(i))\mid i\in[n]\}$ be the set of points, we answer $\indd(I_1,I_2,I_3)$ queries using a special \prbd data structure.
    The reduction has a preprocessing phase where we construct the sets $\FF_1$ and $\FF_2$, and an online phase, where given a query $(I_1,I_2,I_3)$ we reduce it to a \prbd query on the sets in $\FF_1$ and $\FF_2$.
    As before, we can assume that $I_1=I_2\neq I_3$, otherwise the query is solvable in $\tf(1)$ time, see \Cref{claim:17} (for horizontal slab queries, we can assume that $I_1=I_3\neq I_2$).
Note that the case $I_1=I_2=I_3$ is addressed by \Cref{thm3:subquery}.

    \paragraph{Preprocessing.}
    We build one dimensional range tree $T_x$ on the $x$-coordinates of the points, namely every node $v$ in $\Tv$ is associated with a rectangle $R_v=[a_v,b_v]\times[1,n]$. Let $\RR$ denote the set of rectangles associated with the nodes of $\Tv$.
    For every node $v$ in $\Tv$, with associated rectangle $R_v=[a_v,b_v]\times[1,n]$, we construct two sets $S_v$ for $\FF_1$ and $D_v$ for $\FF_2$.
    The set $D_v$ is simply the set of $y$ values of the points in $P_v=P\cap R_v$, that is $D_v=\set{\pi(i)\mid (i,\pi(i))\in P_v}$.
    The set $S_v$ contains the maximal intervals with respect to inclusion.
    That is, for every two values $i<j\in [a_v,b_v]$, with $\pi(i)<\pi(j)$, we add the interval $[\pi(i),\pi(j)]$ to $S_v$, unless there are two values $i'<j'\in [a_v,b_v]$ such that $[\pi(i),\pi(j)]\subseteq[\pi(i'),\pi(j')]$, in which case we ignore $[\pi(i),\pi(j)]$
    Finding these intervals can be done in $\tf(|P_v|)$ time, using a simple line sweep algorithm.
    This completes the preprocessing phase, where we have constructed the sets $\FF_1$ and $\FF_2$.
    Importantly, the set $S_v$ contains at most $O(|P_v|)$ intervals, as the intervals are maximal with respect to inclusion, so each endpoint is in at most one interval. This is crucial to ensure that the total size of the sets in $\FF_1$ does not blow up.

    \paragraph{Query.}
    Given a query $\indd(I_1,I_1,I_2)$, we use the range tree $T_x$ to reduce to slab query including only slabs from $\RR$, the set of rectangles associated with the nodes of $T_x$.
    Then, each slab query is reduced to a \prbd query on the sets in $\FF_1$ and $\FF_2$.
    The first part is straightforward, we find the canonical nodes $\CR(I_1)$ and $\CR(I_2)$ in $T_x$, and reduce $\indd(I_1,I_1,I_2)$ to $\tf(1)$ queries $\JJ=\CR(I_1)\times\CR(I_1)\times\CR(I_2)$.
    We can solve all simple queries in $\JJ$ in $\tf(1)$ time, and if the answer is YES, we are done. Therefore, we assume that all queries in $\JJ$ of the first type has negative response, and focus on the queries of the second type.
    For every triplet $(R_u,R_u,R_v)$ of rectangles (slabs) from $\RR$ can be reduced to \prbd query $(S_u,D_v)$, where $S_u\in \FF_1$ and $D_v\in\FF_2$.

    We show that $\indd(R_u,R_u,R_v)=\indd(S_u,D_v)$, which completes the reduction.

    $(\Longrightarrow)$ Suppose $\indd(R_u,R_u,R_v)=1$. Then, there are three points $p_1,p_2,p_3$ such that $p_1,p_2\in R_u$, and $p_3\in R_v$ forming an $132$ pattern.
    Let $p_i=(x_i,y_i)$ for $i\in[3]$.
    Since $p_1,p_2\in R_u$, we have that $(y_1,y_2)\in S_u$ and since $p_3\in R_v$, we have that $y_3\in D_v$
    Since $(p_1,p_2,p_3)$ form an $132$ pattern, we have $y_1<y_3<y_2$ so $y_3\in (y_1,y_2)$, proving that $S_u\cap D_v\neq\emptyset$.

    $(\Longleftarrow)$ Suppose that $\indd(S_u,D_v)=1$. Then, there is an interval $[a,b]\in S_u$ such that $[a,b]\cap D_v\neq\emptyset$, meaning that there is a value $y_3\in D_v$ such that $y_3\in [a,b]$. Since $[a,b]\in S_u$, there are two points $p_1,p_2\in R_u$, where $p_i=(x_i,y_i)$ for $i\in[2]$, such that $x_1<x_2$ and $[y_1,y_2]=[a,b]$. Since $y_3\in [a,b]$, we have that $y_3\in [y_1,y_2]$, so the three points $p_1,p_2,p_3$ form an $132$ pattern, proving that $\indd(R_u,R_u,R_v)=1$.
    It holds that $x_3>x_2$ as $I_2$ is to the right of $I_1$, otherwise the query is always negative and is simply ignored.
\end{proof}

\subsection{From \prbd\ to \prba}
\label{ssec:hardness prbd to prba}
\renewcommand{\AA}{\mathcal{A}}
\renewcommand{\AB}{\mathcal{B}}
\providecommand{\Ffam}{\mathcal{F}}
\providecommand{\Rbb}{\mathbb{R}}
\providecommand{\Uosd}{U_{\mathrm{OSD}}}
\providecommand{\Path}{\mathsf{Path}}
\providecommand{\PathX}{\mathsf{Path}_x}
\providecommand{\PathY}{\mathsf{Path}_y}
\providecommand{\CovSuffixX}{\CR_x}
\providecommand{\CovSuffixY}{\CR_y}
\providecommand{\OutSet}{\mathsf{Out}}
\providecommand{\InSet}{\mathsf{In}}
In this section we prove the following reduction:
\begin{lemma}
    \label{lemma:eq prba and prbd}
    \prbd and \prba are equivalent.
\end{lemma}
We assume that the two problems have two families of sets $\FF_1$ and $\FF_2$;
instead of a single family $\FF$, where queries $(i,j)$ include the set $S_i\in \FF_1$ and $S_j\in \FF_2$. This is an equivalent formulation just by taking $\FF_1=\FF$ and $\FF_2=\FF$ and answering query $(i,j)$ by OSD query $(i,j)$. The two-family formulation is more natural for the reduction.

The reduction from \prba to \prbd is trivial:
In \prba a set $S$ is a union of elements, that can be turned into a union of intervals by replacing each element $x$ as the interval $[x,x]$.
Here we prove the reduction in the other direction, from \prbd to \prba, which is more involved.
The input of \prbd consists of two families
\(\AA=\{A_1,\dots,A_r\}\) and \(\BB=\{B_1,\dots,B_s\}\), where each \(A_i\) and each \(B_j\) is a set of intervals. Let
\[
    N=\sum_{i=1}^{r}|A_i|+\sum_{j=1}^{s}|B_j|
\]
be the total number of intervals. For query $(i,j)$ the goal is to decide if there is an interval $[a,b]\in A_i$ and an interval $[c,d]\in B_j$ such that $[c,d]\subseteq[a,b]$.
For every pair of indices $(i,j)$ we use $\chi_{i,j}=1$ to denote the event that there is a pair of intervals $[a,b]\in A_i$ and $[c,d]\in B_j$ such that $[c,d]\subseteq[a,b]$, and $\chi_{i,j}=0$ otherwise.
We reduce this to $\tf(1)$ queries of \prba over two families of sets, \(\Ffam_1\) and \(\Ffam_2\) that we define shortly, of size $\tf(N)$ each.

\begin{proof}[Proof of \Cref{lemma:eq prba and prbd}]
    We start with describing the construction of the family sets $\FF_1$ and $\FF_2$.
    We map each interval $[u,v]$ to the point $(u,-v)\in\Rbb^2$.
    We use $\phi$ to denote this mapping, and $\phi(A)$ denotes the set of points obtained by applying $\phi$ to each interval in $A$.
    Then
    \[
        [c,d]\subseteq[a,b]
        \iff
        \phi([a,b])\preceq\phi([c,d]),
    \]
    where \((x_1,y_1)\preceq(x_2,y_2)\) means \(x_1\le x_2\) and \(y_1\le y_2\), and referred to as dominance pair.
    Now a query $(i,j)$ asks if $P_i=\phi(A_i)$ contains a point that dominates some point in $P_j=\phi(B_j)$.

    We use two one-dimensional segment trees to encode this query as an online set disjointness query.
    Let $X$ and $Y$ denote the set of all left endpoints, and right endpoints, respectively, of all intervals in $\AA$ and $\BB$, and let $U$ be the maximum value of $X$ and $Y$, meaning that $X,Y\subseteq[U]$.
    We define two range trees $T_x$ on $X$ and $T_y$ on $-Y$, and use $V(T_x)$ and $V(T_y)$ to denote the node sets of each tree.
    Each leaf node in $T_x$ corresponds to the trivial interval $[a,a]$ for some $a\in X$, and each leaf node in $T_y$ corresponds to the trivial interval $[-b,-b]$ for some $b\in Y$.
    The universe for the \prba instance is $\Uosd=V(T_x)\times V(T_y)$, that is, each set in $\FF_1$ and $\FF_2$ is a subset of $V(T_x)\times V(T_y)$.

    We explain how to turn each set $A_i\in\AA$ and $B_j\in\BB$ into sets $S_i\in\FF_1$ and $T_j\in\FF_2$ such that $S_i\cap T_j\nemp$ if and only if
    $\chi_{i,j}=1$.
    For that, we define two functions $\OutSet$ and $\InSet$ that takes an interval $[a,b]$ and maps it to a subset of elements in $\Uosd$.
    The function $\OutSet$ maps intervals that are candidates for being the outer interval $[a,b]$ in the containment query, and the function $\InSet$ maps intervals that are candidates for being the inner interval $[c,d]$.
    For every $A_i\in\AA$, and $B_j\in\BB$, we define
    \begin{align*}
        S_i=\bigcup_{[a,b]\in A_i}\OutSet([a,b])\;,\quad
        T_j=\bigcup_{[c,d]\in B_j}\InSet([c,d])\;.
    \end{align*}

    The function $\OutSet$ is as follows.
    Given an interval $[a,b]$, let $\CR_x([a,U])$ be the set of nodes in the canonical cover of $[a,U]$ in $T_x$.
    We also define $\CR_y([-b,0])$ to be the set of nodes in the canonical cover of $[-b,0]$ in $T_y$.
    We define $\OutSet([a,b])=\CR_x([a,U])\times \CR_y([-b,0])$.
    The function $\InSet$ is as follows. Given an interval $[c,d]$, let $\Path_x(c)$ be the root-to-leaf path from the root to the leaf $[c,c]$ in $T_x$, which we also treat as the set of nodes in that path.
    Similarly, let $\Path_y(-d)$ be the root-to-leaf path from the root to the leaf $[-d,-d]$ in $T_y$.
    We define $\InSet([c,d])=\Path_x(c)\times \Path_y(-d)$.
    This completes the description of the reduction.

    To see that there is no blowup in the size of the sets, we note that for every interval $I$, and any value $z\in X\cup (-Y)$,
    $|\OutSet(I)|,|\InSet(I)|,|\Path_x(z)|,|\Path_y(z)|=O(\log^2 N)$.
    Therefore $\abs{S_i}=\tf(\abs{A_i})$ and $\abs{T_j}=\tf(\abs{B_j})$, and the total size of all sets in $\FF_1$ and $\FF_2$ is $\tf(N)$.

    \paragraph{Correctness.}
    We prove two things.
    \begin{enumerate}
        \item For any intervals $[a,b]$ and \([c,d]\), it holds that \([c,d]\subseteq[a,b]\) if and only if $\OutSet([a,b])\cap\InSet([c,d])\nemp$.
        \item For every $i,j$, it holds that $\chi_{i,j}=1$ if and only if $S_i\cap T_j\nemp$.
    \end{enumerate}

    \begin{proof}[Proof of (1)]

        $(\Longrightarrow)$
        Suppose \([c,d]\subseteq[a,b]\), then $a\leq c$ and $-b\leq -d$.
        Clearly, $[c,c]\subseteq[a,U]$, and $[-d,-d]\subseteq[-b,0]$.
        Let $u$ be the node in $\CR_x([a,U])$ whose interval contains $c$, and let $v$ be the node in $\CR_y([-b,0])$ whose interval contains $-d$.
        Then $(u,v)\in \OutSet([a,b])\cap \InSet([c,d])$, so $\OutSet([a,b])\cap\InSet([c,d])\nemp$.

        $(\Longleftarrow)$
        Suppose \(\OutSet([a,b])\cap\InSet([c,d])\nemp\), and pick \((u,v)\) in the intersection.
        Then
        \[
            u\in \CR_x([a,U])\cap \Path_x(c),\qquad
            v\in \CR_y([-b,0])\cap \Path_y(-d).
        \]
        Since \(u\in\Path_x(c)\), the interval represented by \(u\) in \(T_x\) contains \(c\). Since \(u\in\CR_x([a,U])\), that represented interval is contained in \([a,U]\). Hence \(a\le c\).
        Similarly, \(v\in\Path_y(-d)\) means the interval represented by \(v\) in \(T_y\) contains \(-d\), and \(v\in\CR_y([-b,0])\) implies that represented interval is contained in \([-b,0]\). Hence \(-b\le -d\), i.e., \(d\le b\).
        Therefore \(a\le c\le d\le b\), so \([c,d]\subseteq[a,b]\).
    \end{proof}
    \begin{proof}[Proof of (2)]

        \leavevmode
        $(\Longrightarrow)$
        Suppose \(\chi_{i,j}=1\). Then there are intervals \([a,b]\in A_i\) and \([c,d]\in B_j\) with \([c,d]\subseteq[a,b]\).
        By (1), \(\OutSet([a,b])\cap\InSet([c,d])\nemp\). Since \(\OutSet([a,b])\subseteq S_i\) and \(\InSet([c,d])\subseteq T_j\), it follows that \(S_i\cap T_j\nemp\).

        $(\Longleftarrow)$
        Suppose \(S_i\cap T_j\nemp\), and pick \((u,v)\in S_i\cap T_j\).
        By definition of the unions, there exist \([a,b]\in A_i\) and \([c,d]\in B_j\) such that
            $(u,v)\in \OutSet([a,b])\cap\InSet([c,d])$.
        By (1), \([c,d]\subseteq[a,b]\). Hence \(\chi_{i,j}=1\).
    \end{proof}
    This completes the proof of \Cref{lemma:eq prba and prbd}.
\end{proof}

\subsection{From \prba to \rppmc (and \Cref{thm:hardness interval k-1})}
\label{ssec:hardness interval k-1}
In this subsection we prove \Cref{thm:hardness interval k-1}.
This reduction, which reduces \prbb to \rppm, clearly also reduces \prba to the $132$-\rppm pattern. An overview of simplified reduction for patterns of length $3$ is given in \Cref{ssec:tech-ov} and specifically in \Cref{fig:to2-hardness-two-rectangles}.
\begin{proof}[Proof of \Cref{thm:hardness interval k-1}]
    We reduce \prbb to $\sigma$-\rppm, where $\sigma= \brak{1,(k+1),k,k-1,\ldots,2}$, in two steps: we first map the family of sets to an array $A$, and then map each query to $k+1$ vertical slabs.
    To reduce the input of \prbb to the input of \rppm, we construct an array $A$ of size $N\cdot (k+1)$ from the family of sets $\FF={S_1,\ldots S_m}$, where $N=\sum_{i=1}^m |S_i|$. We use $N_i=|S_i|$ for $i\in[m]$.
    We assume that each set $S$ is a subset of some universe $U=[\ell]$.
    We map the elements of $U$ into a new universe $U'=\set{2k\cdot u\mid u\in U}$ by the mapping $x\mapsto 2k x$.
    For each set $S_i$, let $E_i=(e_1,e_2,\ldots,e_{N_i})$ be the elements of $S_i$ (under the mapping) in increasing order.
    We associate with $S_i$ new $k$ new subarrays $B_1^i,\ldots,B_k^i$ defined as follows:
    \begin{enumerate}
        \item For $j\in\set{2,3,\ldots,k}$, let $z_j=k+1-j$.
              Define $B_j^i= E_i+z_j$, i.e., $B_j^i=(e_1+z_j,e_2+z_j,\ldots)$.
        \item Define $B_1^i$ as an array of $2|S_i|$ entries.
              Let $E_i'$ be the elements of $S_i$ in decreasing order, i.e., $E_i'=(e_{N_i},e_{N_i-1},\ldots)$.
              The array $B_1^i$ is defined as the interleaving of two shifted copies of $E_i'$. That is,
              \begin{align*}
                  B_1^i=(e_{N_i},e_{N_i}+k,e_{N_i-1},e_{N_i-1}+k,\ldots)\;.
              \end{align*}
    \end{enumerate}
    Then, we define the final array $A$ as the concatenation of all these subarrays, i.e.,
    \begin{align*}
        A = B_1^1 \circ B_1^2 \circ \cdots \circ B_1^m
        \circ B_2^1 \circ B_2^2 \circ \cdots \circ B_2^m
        \circ \cdots \circ B_k^1 \circ B_k^2 \circ \cdots \circ B_k^m\;.
    \end{align*}
    \begin{remark}
        The construction above defines an array $A$, which need not be a permutation as values can repeat across different sets.
        To solve this, we can replace each value $A[i]$ with its stable rank, i.e., its rank when the pairs $(A[i],i)$ are ordered lexicographically.
        This does not affect the reduction: the intended witnesses use distinct values, while equal values can only be ordered according to their indices and therefore cannot create the decreasing part of the pattern.
        For readability, the correctness proof below is written using the original values of $A$; the stable-rank replacement translates these comparisons to a permutation.
    \end{remark}

    For each $i\in[m]$ and $j\in[k]$, let $L^i_j\subseteq[N(k+1)]$ denote
    the interval of indices such that $B^i_j$ appears in $A$ at these indices.

    \paragraph{Queries.}
    To answer the query $(h_1,h_2,\ldots,h_k)$ of whether $\bigcap_{i=1}^k S_{h_i}=\emptyset$, we prepare $k+1$ rectangles $R_0,R_1,\ldots,R_k$ as follows.
    For $j\in\set{1,2,\ldots, k}$ the rectangle $R_j$ is defined by $R_j=L^{h_j}_j\times [N(k+1)]$.
    The rectangle $R_0=R_1$.
    We then query the \rppm data structure with the rectangles $(R_0,R_1,\ldots,R_k)$ and the pattern $\sigma$, and answer accordingly.

    \paragraph{Correctness.}
    We show that $\bigcap_{i=1}^k S_{h_i}\neq \emptyset$ if and only if
    $\ind_{\sigma}(R_1,R_2,\ldots,R_k)=1$.

    ($\Longrightarrow$) Suppose $\bigcap_{i=1}^k S_{h_i}\neq \emptyset$, and let $e\in \bigcap_{i=1}^k S_{h_i}$.
    Since $e \in S_{h_1}$, the element $e$ appears in the sorted list $E_{h_1}$. By the interleaving construction of $B_1^{h_1}$, the pair $(e, e+k)$ appears at two consecutive positions in $B_1^{h_1}$. Let $t_0 < t_1$ be the corresponding indices in $A$ (both lying in $L_1^{h_1}$).
    For $j \in \{2,\ldots,k\}$, since $e \in S_{h_j}$, the element $e$ appears in $E_{h_j}$, and thus $e + z_j = e + (k+1-j)$ appears in $B_j^{h_j}$. Let $t_{j}$ be the corresponding index in $A$ (lying in $L_j^{h_j}$).
    Since $A$ concatenates all $B_1$ blocks before all $B_2$ blocks, and so on, we have $t_0 < t_1 < t_2 < \cdots < t_k$, and the values are:
    \begin{align*}
        A[t_0] = e,\quad A[t_1] = e+k,\quad A[t_{j}] = e + (k+1-j) \text{ for } j=2,\ldots,k.
    \end{align*}
    That is, the values are $e, e+k, e+(k-1), \ldots, e+1$.
    The ranks are:
    \begin{itemize}
        \item $A[t_0]=e$ has rank $1$,
        \item $A[t_1]=e+k$ has rank $k+1$,
        \item $A[t_{j}]=e+(k+1-j)$ has rank $k+2-j$ for $j=2,\ldots,k$.
    \end{itemize}
    This is exactly the pattern $(1, k+1, k, k-1, \ldots, 2) = \sigma$. Moreover, $t_0, t_1 \in L_1^{h_1}$ and $t_{j} \in L_j^{h_j}$ for $j \geq 2$, so the rectangle constraints are satisfied. Hence $\ind_{\sigma}(R_1,R_2,\ldots,R_k)=1$.

    \medskip
    ($\Longleftarrow$)
    Suppose $\ind_{\sigma}(R_1,R_2,\ldots,R_k)=1$, witnessed by indices
    $t_0 < t_1 < t_2 < \cdots < t_k$ that are order-isomorphic to $\sigma=(1,k{+}1,k,\ldots,2)$.
    That is,
    \begin{align}\label{eq:rank-order}
        A[t_0] < A[t_k] < A[t_{k-1}] < \cdots < A[t_2] < A[t_1]\;.
    \end{align}

    Since $t_0,t_1\in L_1^{h_1}$, both index into $B_1^{h_1}$.
    We prove that $t_1=t_0+1$, and that $t_0$ is odd. This proves that $A[t_0]=e$ and $A[t_1]=e+k$.

    \begin{claim}\label{claim:t1 t0 plus 1}
        $t_1=t_0+1$, and $t_0$ is odd.
    \end{claim}
    \begin{proof}[Proof of \Cref{claim:t1 t0 plus 1}]
        \newcommand{\Mm}{W}
        Write $B = B_1^{h_1}$, and let $\Mm = |S_{h_1}|$.
        Let $e_1 < e_2 < \cdots < e_{\Mm}$ be the elements of $S_{h_1}$, so that
        \[
            B=(e_{W},e_{W}+k,e_{W-1},e_{W-1}+k,\ldots,e_1,e_1+k)\;.
        \]
        We rewrite this as
        \[
            B=(f_1,f_1+k,f_2,f_2+k,\ldots,f_{\Mm},f_{\Mm}+k),
        \]
        where $f_i = e_{\Mm+1-i}$ for $i=1,\ldots,\Mm$ to simplify the notation in the proof, where $f_1 > f_2 > \cdots > f_{\Mm}$ are the elements of $S_{h_1}$ in decreasing order.
        We use $1$-based indexing for $B$, so $B[1]=f_1$, $B[2]=f_1+k$, and so on.
        Let $s_0 < s_1$ be the positions of $t_0, t_1$ within $B$.
        We show that $s_1 = s_0 + 1$ and that $s_0$ is odd.

        We claim that for every even position $2p$ in $B$, it holds that $B[2p] \geq B[2p+z]+2k$ for every $z \geq 1$.
        \begin{enumerate}
            \item If $z$ is even, i.e., $z=2d$, then $B[2p] = f_p+k$ and $B[2p+z] = f_{p+d}+k$. Since $f_p>f_{p+d}$, and both belong to $U'$, we have $f_p \geq f_{p+d}+2k$, and hence $B[2p] \geq B[2p+z]+2k$. We used the fact that the elements of $U'$ are spaced $2k$ apart, so $f_p-f_{p+d}\ge 2k$.
            \item If $z$ is odd, i.e., $z=2d-1$, then $B[2p] \geq B[2p+2d]+2k=B[2p+2d-1]+3k$, which completes the proof.
        \end{enumerate}
        To complete the proof, we combine the proven inequality, with the equality $B[2p-1]=B[2p]-k$ for every $p$.
        We get that for $z\geq 1$,
        \begin{align*}
            B[2p-1] = B[2p]-k \geq B[2p+z]+2k -k = B[2p+z]+k,
        \end{align*}
        which implies that $B[2p-1] \geq B[2p+z]+k$.
        This implies that if $s_1 > s_0 + 1$, then $B[s_1] < B[s_0]$, contradicting the order in \Cref{eq:rank-order}. Hence, $s_1$ must be equal to $s_0 + 1$.
        We also get that if $s_0$ is even, then $B[s_0] > B[s_1]$ contradicting the order in \Cref{eq:rank-order}. Hence, $s_0$ must be odd, which completes the proof of the claim.
    \end{proof}

    Since $t_1=t_0+1$ and $t_0$ is odd, the two indices $t_0$ and $t_1$ correspond to two consecutive entries in $B_1^{h_1}$, which are of the form $(e,e+k)$ for some $e\in S_{h_1}$.
    We are left with showing that $e\in S_{h_j}$ for every $j\in\{2,\ldots,k\}$, which implies that $e\in \bigcap_{i=1}^k S_{h_i}$, and completes the proof of the theorem.
    \paragraph{Showing $e \in S_{h_j}$.}
    Since $t_j\in L_j^{h_j}$, the index $t_j$ falls in $B_j^{h_j}$.
    Recall that the entries of $B_j^{h_j}$ are of the form
    $E_{h_j}+z_j$, where $z_j=k+1-j\in\{1,\ldots,k-1\}$.
    Thus $A[t_j]=f+z_j$ for some $f\in S_{h_j}\cap U'$.
    We show that $f$ must be equal to $e$, which implies that $e\in S_{h_j}$, and completes the proof of the claim.
    By~\eqref{eq:rank-order}, $e = A[t_0] < A[t_j] < A[t_1] = e+k$, i.e.,
    \[
        e \;<\; f + z_j \;<\; e+k\;.
    \]
    \begin{itemize}
        \item By subtracting $f$ from the first inequality, we get $e-f < z_j<k$.
        \item By subtracting $(e+z_j)$ from the second inequality, we get $f-e< k-z_j<k$.
    \end{itemize}
    In both cases we get $|f-e|<k$.
    But $f$ and $e$ are both in $U'$, whose elements are spaced $2k$ apart,
    so $|f-e|<k$ forces $f=e$. Hence $e\in S_{h_j}$.
\end{proof}


\section{Pattern Detection via CSP}
\label{sec:csp basic}
In this section, we reformulate the \rppm problem as a binary constraint satisfaction problem~(CSP), following the framework of Berendsohn, Kozma, and Marx~\cite{BKM21}.
We then show how to solve special classes of CSP instances efficiently, based on the primal graph of the CSP, including instances whose primal graph is acyclic or one edge away from being acyclic.
One of the main results of this section is \Cref{thm3:csp Gm} which we state later because it relies on the definitions we have not yet introduced.
\subsection{\rppm as CSP}
Recall that the \rppm problem (\cref{prob:rppm}) asks whether there exists a monotonically increasing embedding
$\phi:[k]\to[n]$ such that
\begin{enumerate}[label=$(\roman*)$]
    \item the points $(\phi(i),\pi(\phi(i)))_{i\in[k]}$ are order-isomorphic to $\sigma$,
    \item $\phi(i)\in P_i$ for every $i\in[k]$, and
    \item $(\phi(i),\pi(\phi(i)))\in R_i$ for every $i\in[k]$.
\end{enumerate}
We now reformulate \rppm as a binary CSP.

\begin{definition}[CSP formulation]\label{def:csp}
    Given a \rppm instance $(\pi,\sigma,(P_i)_{i\in[k]},(R_i)_{i\in[k]})$, we define a CSP instance $(V,\mathcal{D},C)$ as follows.
    \begin{itemize}
        \item \textbf{Variables.} $V=\{v_1,\ldots,v_k\}$, one for each element of the pattern.
        \item \textbf{Domains.} For each $i\in[k]$, the domain of $v_i$ is
              $D_i = \sset{t\in P_i \mid (t,\pi(t))\in R_i}$.
        \item \textbf{Constraints.} For every pair $1\le i<j\le k$, we impose two binary constraints:
              \begin{description}[labelindent=2em]
                  \item[($x$-order)] $v_i < v_j$, encoding the left-to-right order of the pattern.
                  \item[($y$-order)] $\pi(v_{\si(i)}) < \pi(v_{\si(j)})$, encoding the bottom-to-top order of the pattern.
              \end{description}
    \end{itemize}
    A \emph{satisfying assignment} is a function $\phi:[k]\to[n]$ with $\phi(i)\in D_i$ for every $i\in[k]$ that satisfies all constraints in~$C$.
\end{definition}
This is a binary CSP with $k$ variables, $\binom{k}{2}$ binary constraints, and $k$ unary constraints ensuring that the assignment of each variable falls inside the corresponding domain.
The number of binary constraints can be reduced from $\binom{k}{2}$ to at most $2(k-1)$ by observing that transitivity makes most constraints redundant. Specifically, it suffices to retain:
\begin{enumerate}
    \item The $x$-order constraints $v_i < v_{i+1}$ for $i \in [k-1]$.
    \item The $y$-order constraints $\pi(v_{\si(i)}) < \pi(v_{\si(i+1)})$ for $i \in [k-1]$.
\end{enumerate}
Further constraints may become redundant by the domains. For instance, if every value in $D_i$ is smaller than every value in $D_{i+1}$, then the constraint $v_i < v_{i+1}$ is automatically satisfied and can be omitted without affecting feasibility.

To formalize this, we distinguish between a maximal graph determined only by the pattern and the actual graph induced by the constraints that are retained in a given CSP instance.
\begin{definition}[Full Primal Graph $G_\sigma$]
    Given a pattern $\sigma:[k]\to[k]$, the \emph{full primal graph} $G_\sigma$ is the multigraph on the vertex set $[k]$ whose edges are the union of the two Hamiltonian paths
    \begin{align*}
        P_{x} & = (1 , 2 , \cdots , k),                                        \\
        P_{y} & = (\sigma^{-1}(1) , \sigma^{-1}(2) , \cdots , \sigma^{-1}(k)).
    \end{align*}
    \begin{itemize}
        \item
              The edge $(i,i+1)$ of $P_{x}$ represents the $x$-order constraint $v_{i}<v_{i+1}$.
        \item
              The edge $(\sigma^{-1}(i),\sigma^{-1}(i+1))$ of $P_{y}$ represents the $y$-order constraint $\pi(v_{\sigma^{-1}(i)}) < \pi(v_{\sigma^{-1}(i+1)})$.
    \end{itemize}
\end{definition}
Thus $G_\sigma$ has at most $2(k-1)$ edges, maximum degree at most~$4$, and each edge has multiplicity at most~$2$ (when both paths share an edge).
It is the densest primal graph associated with the pattern $\sigma$, and it arises exactly when all of the consecutive $x$-order and $y$-order constraints are kept.

Given a \rppm instance $\Lambda=(\pi,\sigma,(R_i)_{i\in[k]})$, or equivalently a CSP instance $M=(V,\mathcal{D},C)$ as in \Cref{def:csp}, we define the \emph{primal graph} of $\Lambda$ (or $M$) to be the subgraph of $G_\sigma$.
\begin{definition}[CSP Primal Graph]
    Given a CSP instance $M=(V,\mathcal{D},C)$, the \emph{primal graph} of $M$, denoted by $G_M$, is a subgraph of $G_\sigma$, with the following edges:
    \begin{itemize}
        \item
              The edge $(i,i+1)$ is in $G_M$ if the $x$-order constraint $v_{i}<v_{i+1}$ is in $C$.
        \item
              The edge $(\sigma^{-1}(i),\sigma^{-1}(i+1))$ is in $G_M$ if the $y$-order constraint $\pi(v_{\sigma^{-1}(i)}) < \pi(v_{\sigma^{-1}(i+1)})$ is in $C$.
    \end{itemize}
\end{definition}
For a \rppm instance $\Lambda$ with corresponding CSP instance $M$, we write $G_\Lambda$ and $G_M$ to denote their primal graph.
This primal graph allows us to identify specific structural properties of the instance that we can solve efficiently.
Such properties include being acyclic or having small cut, which we explain how to exploit in the next subsection.

We also define the sub-CSP instance.
We say that $\mathcal{D}'=(D_1',\ldots, D_k')$ is a \emph{narrowing} of $\mathcal{D}=(D_1,\ldots, D_k)$, if $D_i'\subseteq D_i$ for every $i\in[k]$.
Then, $M'=(V',\mathcal{D}',C)$ is a \emph{sub-instance} of $M=(V,\mathcal{D},C)$ if $V'\subseteq V$, $C'\subseteq C$, and $\mathcal{D}'$ is a narrowing of $\mathcal{D}$.
Note that the satisfiability of $M'$ or $M$ does not necessarily imply the satisfiability of the other.
Suppose $M$ is not satisfiable but removing all the constraints might make it satisfiable.
On the other hand if $M$ is satisfiable, then narrowing the domains might make it unsatisfiable.

Here, we care about a special type of sub-instance, which we call a \emph{rectangular narrowing}, where the domains are narrowed by intersecting them with axis-parallel rectangles. We call this type of sub-instance a \emph{rectangular sub-instance} of the original instance.

\subsection{Decomposing a CSP Across a Cut of Size One}
\label{ssec:cut-csp}
The main goal of this subsection, is to show how to decompose a CSP instance into two almost independent CSP instances, given that the primal graph has a cut of size at most one.
Informally, if by removing a single constraint we get two independent sub-instances, then we can solve each sub-instance independently and combine their solutions to solve the original instance.

Let $M=(V,\mathcal{D},C)$ be a CSP instance, and partition $V$ into two parts $V_1$ and $V_2$. This induces a partition of $C$ into $C_1, C_2$, and $C_{1,2}$, where $C_i$ contains the constraints with both endpoints in $V_i$, and $C_{1,2}$ contains the rest. If $C_{1,2}=\emptyset$, then $M$ decomposes into two independent instances
\[
    M_1=(V_1,\mathcal{D}_{1},C_1) \qquad\text{and}\qquad M_2=(V_2,\mathcal{D}_{2},C_2),
\]
where $\mathcal{D}_{i}$ is the (trivial) restriction of $\mathcal{D}$ to $V_i$, and any satisfying assignments $\phi_1$ of $M_1$ and $\phi_2$ of $M_2$ combine into a satisfying assignment $\phi_1\cup\phi_2$ of $M$.

We show that this decomposition extends to the case $|C_{1,2}|=1$. Before stating the lemma, we set up some terminology.
Recall that $\mathcal{D}'=(D_1',\ldots, D_k')$ is a \emph{rectangular-narrowing} of $\mathcal{D}=(D_1,\ldots, D_k)$ if $D_i'$ is the intersection of $D_i$ and some axis-parallel rectangle $R_i$ for every $i\in[k]$. Additionally, $M'=(V',\mathcal{D}',C')$ is a \emph{sub-instance} (or \emph{rectangular-sub-instance}) of $M=(V,\mathcal{D},C)$ if $V'\subseteq V$, $C'\subseteq C$, and $\mathcal{D}'$ is a narrowing (rectangular-narrowing) of $\mathcal{D}$.

\begin{definition}[Oracle for a CSP Instance]
    Let $M=(V,\mathcal{D},C)$ be a binary CSP instance. An \emph{oracle} for $M$ is an algorithm that takes as input a rectangular-sub-instance $M'$ of $M$ and decides whether $M'$ is satisfiable.
\end{definition}

\begin{claim}\label{claim:binary-search}
    Let $M=(V,\mathcal{D},C)$ be a binary CSP instance with an oracle $O_1$.
    Using $O(k\log n)=O(\log n)$ queries to $O_1$, one can either certify that $M$ is unsatisfiable or find a satisfying assignment $\phi$ to $M$.
    Moreover, if $M$ is satisfiable, then given a single index $i\in[k]$ and one of the four coordinate orders $\prec$ (left-to-right, right-to-left, bottom-to-top, or top-to-bottom), the algorithm can ensure that $\phi(i)$ is $\prec$-minimal among all satisfying assignments.
\end{claim}
Note that this guarantee applies to one index-order pair at a time; minimality under multiple indices or orders simultaneously is not claimed.
We give the details for completeness.
\begin{proof}[Proof of \Cref{claim:binary-search}]
    The proof is the usual search-to-decision self-reduction. We keep a satisfiable rectangular-sub-instance and shrink one domain to a singleton without losing all satisfying assignments.

    We explain how to find a satisfying assignment $\phi$ to $M$
    where $\phi(i)$ is $\prec$-minimal, for specific $i\in[k]$.
    First we check that $M$ is satisfiable by querying $O_1$, otherwise we certify that $M$ is unsatisfiable.

    For simplicity, we assume that the $\prec$ order is the total ordering with respect to the $x$-coordinate, $p\prec q$ if and only if $p_x<q_x$. The other three orders are handled symmetrically.

    Now we consider a new rectangular-sub-instance of $M$ that is obtained by restricting the domain of $v_i$.
    Specifically, we partition the domain $D_i$ into two almost equal halves using the halfplane $H=\{p\in D_i : p_x\le t\}$ for some $t$ that gives a balanced partition of $D_i$ into two halves.
    The new input $N^1$ is the rectangular-sub-instance of $M$ obtained by restricting the domain of $v_i$ to $H$. If $N^1$ is satisfiable, then we can recurse on $N^1$. Otherwise, we recurse on the rectangular-sub-instance $N^2$ obtained by restricting the domain of $v_i$ to the other half, $D_i\setminus H$. In either case, we have reduced the size of the domain of $v_i$ by a factor of at least $2$, and we can continue this process until the domain of $v_i$ is a singleton. The value in this singleton is exactly the $\prec$-minimal value that can appear in a satisfying assignment of $M$.
    Since the domain of $v_i$ has size at most $n$, the number of oracle queries is $O(\log n)$.
\end{proof}

The following lemma is the main result of this subsection:
\begin{lemma}[Eliminating Cuts of Size One]
    \label{lemma:P3}
    Let $M=(V,\mathcal{D},C)$ be a binary CSP instance, and let $V_1,V_2$ be a partition of $V$ that induces a partition of $C$ into $C_1,C_2,C_{1,2}$ as above. Assume that $|C_{1,2}|\leq 1$.
    Given two oracles $O_1$ and $O_2$ for $M_1=(V_1,\mathcal{D}_1,C_1)$ and $M_2=(V_2,\mathcal{D}_{2},C_2)$, respectively, we can find an assignment $\phi$ to $M$ or certify that $M$ is unsatisfiable, using $O(\log n)$ queries to $O_1$ and $O(\log n)$ queries to $O_2$.
\end{lemma}

\begin{proof}[Proof of \Cref{lemma:P3}]
    If $C_{1,2}=\emptyset$, then $M$ decomposes into the two independent instances $M_1$ and $M_2$, and by \Cref{claim:binary-search} we can either certify that one of them is unsatisfiable or find satisfying assignments to both using $O(\log n)$ queries per oracle.
    In the latter case, their union is a satisfying assignment to $M$.

    Assume then that $C_{1,2}=\{C_{i,j}\}$, where $v_i\in V_1$, $v_j\in V_2$, and $C_{i,j}$ is the constraint $v_i\prec v_j$ for one of the four coordinate orders.

    We first query $O_1$ on $M_1$. If $M_1$ is unsatisfiable, so is $M$. Otherwise, by \Cref{claim:binary-search} we find a satisfying assignment $\phi_1$ to $M_1$ in which $\phi_1(v_i)$ is $\prec$-minimal among all satisfying assignments to $M_1$, using $O(\log n)$ queries to $O_1$.
    Let $M_2'=(V_2,\mathcal{D}_2',C_2)$ be the sub-instance of $M_2$ obtained by restricting the domain of $v_j$ to $\{p\in D_j : \phi_1(v_i) \prec p\}$ and leaving all other domains unchanged.
    Since $\prec$ is a coordinate order, this restriction is an axis-parallel halfplane and therefore $M_2'$ is a rectangular-sub-instance.
    We show below that $M$ is satisfiable if and only if $M_2'$ is.
    Assuming this, we query $O_2$ on $M_2'$; if satisfiable, \Cref{claim:binary-search} yields a satisfying assignment $\phi_2$ using $O(\log n)$ queries, and $\phi_1\cup\phi_2$ satisfies $M$.

    \paragraph{Correctness.}
    We claim that $M$ is satisfiable if and only if $M_2'$ is satisfiable.

    $(\Longleftarrow)$ Let $\phi_2$ be a satisfying assignment to $M_2'$. Then $\phi=\phi_1\cup\phi_2$ satisfies all constraints in $C_1\cup C_2$, and since $\phi_2(v_j)\in D_j'$ we have $\phi_1(v_i)\prec \phi_2(v_j)$, so the crossing constraint $C_{i,j}$ is satisfied as well.

    $(\Longrightarrow)$ Let $\psi$ be a satisfying assignment to $M$, with restrictions $\psi_1$ and $\psi_2$ to $V_1$ and $V_2$.
    Since $\psi_1$ satisfies $M_1$ and $\phi_1(v_i)$ is $\prec$-minimal among all satisfying assignments to $M_1$, we have $\phi_1(v_i)\preceq\psi_1(v_i)\prec\psi_2(v_j)$, where the second inequality holds because $\psi$ satisfies $C_{i,j}$.
    Hence $\psi_2(v_j)\in D_j'$, so $\psi_2$ satisfies $M_2'$.
\end{proof}

\begin{claim}[Peeling Degree-One Variables]
    \label{claim:peeling-degree-one}
    After $\tf(n)$ preprocessing, given a rectangular sub-instance $M=(V,\mathcal{D},C)$ of an \rppm CSP and a variable $v_i$ incident to at most one constraint, one can in $\tf(1)$ time either certify that $M$ is unsatisfiable or construct a rectangular sub-instance $M'$ on $V\setminus\{v_i\}$ such that $M$ is satisfiable if and only if $M'$ is satisfiable.
    Moreover, any satisfying assignment of $M'$ extends to one of $M$ in $\tf(1)$ time.
\end{claim}

\begin{proof}[Proof of \Cref{claim:peeling-degree-one}]
    The preprocessing step is to build one range tree for each color class of the input point set, i.e., for each $P_i$,
    that supports emptiness queries in $\tf(1)$ time, and thus allows to get the leftmost, rightmost, bottommost, and topmost points in any axis-parallel rectangle in $\tf(1)$ time.

    If $D_i=\emptyset$, then $M$ is unsatisfiable.
    If $v_i$ has no incident constraint, choose any $p_i\in D_i$ and let $M'$ be obtained by deleting $v_i$.
    This preserves satisfiability, and any solution of $M'$ extends by setting $v_i=p_i$.

    Otherwise, write the unique incident constraint as $v_i\prec v_j$, where $\prec$ is one of the four coordinate orders.
    Let $p_i=\min_{\prec}D_i$.
    Let $D_j'=D_j\cap\{q:p_i\prec q\}$,
    and let $M'$ be obtained by deleting $v_i$ and replacing $D_j$ by $D_j'$.
    This is still a rectangular narrowing, since $\{q:p_i\prec q\}$ is an axis-parallel halfplane.
    If $M'$ is satisfiable, then any satisfying assignment to $M'$ extends to one of $M$ by setting $v_i=p_i$, since the new domain of $v_j$ guarantees $p_i\prec v_j$.
    Conversely, suppose $M$ has a satisfying assignment $\psi$.
    By the choice of $p_i$, we have $p_i\preceq \psi(v_i)$.
    Since $\psi$ satisfies the unique incident constraint, $\psi(v_i)\prec \psi(v_j)$.
    By transitivity, $p_i\prec \psi(v_j)$, and therefore $\psi(v_j)\in D_j'$.
    Thus the restriction of $\psi$ to $V\setminus\{v_i\}$ satisfies $M'$.
\end{proof}

\begin{corollary}[Iterated Degree-One Peeling]
    \label{cor:iterated-degree-one-peeling}
    After $\tf(n)$ preprocessing, any rectangular sub-instance $M$ can be reduced in $\tf(1)$ time, unless unsatisfiability is certified, to a rectangular sub-instance $M^\circ$ whose primal graph is empty or has minimum degree at least $2$, while preserving satisfiability and allowing any solution of $M^\circ$ to be extended to a solution of $M$.
\end{corollary}

\begin{proof}[Proof of \Cref{cor:iterated-degree-one-peeling}]
    Repeatedly apply \Cref{claim:peeling-degree-one} until no vertex of degree at most one remains.
    There are at most $k$ applications, and $k$ is constant.
\end{proof}

Applying \Cref{cor:iterated-degree-one-peeling} to forests yields:
\begin{corollary}
    \label{cor:csp is a forest}
    Let $M=(V,\mathcal{D},C)$ be a binary CSP instance whose primal graph $G_M$ is a forest (with no edge multiplicities).
    $M$ can be solved in $\tf(1)$ time, and after $\tf(n)$ preprocessing on the input point set.
\end{corollary}
We extend this to the case where the primal graph is a forest with edge multiplicities, i.e., a multigraph whose underlying simple graph is a forest:

\begin{lemma}[Acyclic Multigraphs]
    \label{lemma:acyclic multigraphs}
    Let $M=(V,\mathcal{D},C)$ be a binary CSP instance whose primal graph $G_M$ is a forest (with edge multiplicities).
    Then, there exists an algorithm that computes an oracle $O_1$ for $M$ in $\tf(n)$ preprocessing time, where $O_1$ answers any rectangular-sub-instance $M'$ of $M$ in $\tf(1)$ time.
\end{lemma}

The only missing ingredient for the proof of \Cref{lemma:acyclic multigraphs} is a data structure for the \rppm problem with monotone patterns. (Note that we did not try to optimize the exponents in the claim below.)
\begin{restatable}[Monotone Patterns]{claim}{LemmaMono}\label{lemma:monotone patterns}
    There is a data structure for the \rppm problem with monotone pattern $\sigma=12\cdots k$ (or $\sigma=k(k-1)\cdots 1$) that takes $2^{O(k^2)} \cdot O(n \log^{3k} n)$ preprocessing time and $(O(\log n))^{3k^2}$ query time for every $k$.
\end{restatable}
In other words, \Cref{lemma:monotone patterns} states that there is an algorithm taking as input a CSP $M$ on $k$ variables (for constant $k$) whose primal graph is a path, and in $\tf(n)$ time outputs an oracle $O_1$ for $M$ that answers any rectangular-sub-instance $M'$ of $M$ in $\tf(1)$ time.

This is a natural extension of the technique presented in \cite{ben2024approximate} that solve the \rppm problem for the pattern $12$ in $\tf(n)$ preprocessing time and $\tf(1)$ query time.
The proof follows the same idea of using a range tree, and is deferred to the appendix.

\begin{proof}[Proof of \Cref{lemma:acyclic multigraphs} Using \Cref{lemma:monotone patterns}]
    We prove the lemma by induction on the number of non-fixed variables.
    The base case is trivial.
    To prove the induction step, we first show that if $G_M$ is a forest with edge multiplicities, then at least one of the following holds:
    \begin{enumerate}
        \item $G_M$ has a cut of size at most one.
        \item $G_M$ is a path where all edge has multiplicity two.
    \end{enumerate}
    If $G_M$ is not connected, it has a cut of size zero, so $(1)$ holds.
    We assume $G_M$ is connected, and therefore is a tree, with edge multiplicities.
    If one edge has multiplicity one, then $G_M$ has a cut of size one, and $(1)$ holds.
    We are left with proving that if every edge has multiplicity two, then $G_M$ is a path.
    Recall that the edge set of the full primal graph $G_\sigma$ is the union of two hamiltonian paths $P_x$ and $P_y$, where
    \begin{align*}
        P_{x} = (1 , 2 , \cdots , k)\;, &&
        P_{y} = (\si(1) , \si(2) , \cdots , \si(k))\;.
    \end{align*}
    If every edge $(i,i+1)$ has multiplicity two, then the path $P_y$ must share the exact same edges as the path $P_x$, hence:
    \begin{enumerate*}[label=$(\roman*)$]
        \item $P_y=P_x$ and $\sigma=12\cdots k$, or
        \item $P_y$ is the reverse of $P_x$ and $\sigma=k(k-1)\cdots 1$.
    \end{enumerate*}
    In both cases, $G_M$ is a path, completing the proof that at least one of the two cases holds.

    We explain how to solve $M$ in each of the two cases.
    Case $(2)$ is handled by \Cref{lemma:monotone patterns}, so we focus on case $(1)$.
    If $M$ has a cut of size at most one, let $M_1$ and $M_2$ be the two rectangular sub-instances on the two sides of the cut.
    Since they have fewer variables (and edges) than $M$, the primal graph of each sub-instance is also a forest with edge multiplicities.
    Thus, by the induction hypothesis we can construct two oracles $O_1$ and $O_2$ for $M_1$ and $M_2$, respectively, in $\tf(n)$ preprocessing time, where each oracle answers any rectangular-sub-instance of its corresponding instance in $\tf(1)$ time.
    Using these oracles, we can apply \Cref{lemma:P3} to solve $M$ in $\tf(1)$ time.
    This proves the existence of an oracle $O_{1,2}$ for $M$ that answers any rectangular-sub-instance of $M$ in $\tf(1)$ time, after $\tf(n)$ preprocessing time.
\end{proof}

\newcommand{\JJ}{\mathcal{J}}

\renewcommand{\CR}{\mathcal{C}}

\subsection{Solving Special \rppm Instances}
In this subsection, we isolate additional structural properties of the primal graph of a CSP instance that allow us to solve it efficiently.
We already showed how to solve instances whose primal graph, which is a multigraph, is acyclic, meaning the simple graph obtained by keeping only one copy of each edge is acyclic.
We now show how to solve instances whose primal graph is a multigraph that is a cycle with edge multiplicities.
Finally, we show how to solve instances whose primal graph can be turned into a cycle, by iteratively removing vertices of degree at most one.
We call this graph a whiskered cycle.
We use whiskered $c$-cycle to denote a whiskered cycle whose underlying cycle has at most $c$ vertices.

Our main result in this section is a data structure for solving special classes of \rppm instances $\Lambda$ whose primal graph is a whiskered cycle.
\begin{restatable}[Preprocessing-Query \rppm]{theorem}{thmRRPM}
    \label{thm2:csp Gm}
    There is a data structure that takes as input a permutation $\pi$ of length $n$ and two parameters $c\leq k$.
    After $\tf(n^{2-1/c})$ time, the data structure can answer in time $\tf(n^{1-1/c})$ on \rppm queries $\Lambda = (\pi,\sigma, (R_i)_{i \in [k]})$ with $|\sigma| \leq k$ where $G_\Lambda$ is a whiskered $c$-cycle.
\end{restatable}

We prove the above theorem using a slightly modified version.

\begin{restatable}[Preprocessing-Query \rppm: modified version]{theorem}{thmRRPMB}
    \label{thm3:csp Gm}
    There is a data structure that takes as input a permutation $\pi$ of length $n$ and two parameters $c$ and $k$, where $c\leq k$.
    After $\tf(n^{2-1/c})$ time, the data structure can answer in time $\tf(n^{1-1/c})$ on \rppm queries $\Lambda = (\pi,\sigma, (R_i)_{i \in [k]})$ with $|\sigma| \leq k$, assuming $G_\Lambda$ is a subgraph of a $c$-cycle (with possible edge multiplicities).
\end{restatable}
\begin{proof}[Proof of \Cref{thm2:csp Gm} using \Cref{thm3:csp Gm}]
    We are given a permutation $\pi$ of length $n$ and two parameters $c\leq k$, and construct a data structure $\AC$, that can answer \rppm queries $\Lambda = (\pi,\sigma, (R_i)_{i \in [k]})$ with $|\sigma| \leq k$ where $G_\Lambda$ is a $c$-cycle, in $\tf(n^{1-1/c})$ time.
    We use $\AC$ to get a new data structure $\AC'$ that can answer \rppm queries $\Lambda' = (\pi,\sigma, (R_i)_{i \in [k]})$ with $|\sigma| \leq k$ where $G_\Lambda$ is a whiskered $c$-cycle, in $\tf(n^{1-1/c})$ time.
    Given the query $\Lambda'$, we peel the vertices of degree at most one in $G_{\Lambda'}$ using \Cref{cor:iterated-degree-one-peeling}.
    We thus get a new query $\Lambda''$ whose primal graph $G_{\Lambda''}$ has no vertices of degree at most one, keeping only the vertices of the underlying $c$-cycle.
    We then answer $\Lambda''$ using $\AC$, and return the answer for $\Lambda'$.
\end{proof}
For the rest of this section we prove \Cref{thm3:csp Gm}. We start with an intuition.

We use two ideas.
The first one is to use heavy-light decomposition of the queries.
We say that a rectangle $R$ is $\tau$-\emph{heavy} if $|P\cap R|>\tau$, i.e., if it contains more than $\tau$ points of $P$, and otherwise we say that it is $\tau$-\emph{light}.
A query $\Lambda=(\pi,\sigma, (R_i)_{i \in [k]})$ is $\tau$-\emph{light} if it contains at least one light rectangle, and \emph{heavy} otherwise.
We solve light queries by a simple brute-force approach, and preprocess all answers for heavy queries.

The second idea, which we already used, is to reduce an \rppm query $\Lambda=(\pi,\sigma, (R_i)_{i \in [k]})$ with general ranges $R_i$ to polylogrithmically many \rppm queries $\Lambda'=(\pi,\sigma, (R'_i)_{i \in [k]})$ where each $R'_i$ is a rectangle from a family $\RR$ of $\tf(n)$ rectangles.
This is obtained by range tree decomposition, as explained in \Cref{sec:preliminaries}.

We first explain how to solve light queries, and then we explain how to preprocess all heavy queries.
\begin{claim}\label{claim2:light}
    Every $\tau$-light query $\Lambda=(\pi,\sigma, (R_i)_{i \in [k]})$, where the primal graph $G_\Lambda$ is a subgraph of a cycle (with edge multiplicities), can be solved in $\tf(\tau)$ time.
\end{claim}

\begin{proof}[Proof of \Cref{claim2:light}]
    Let $M=(V,\mathcal{D},C)$ be the CSP instance corresponding to $\Lambda$, with primal graph $G_M$.
    If $G_M$ is acyclic, then $M$ can be solved in $\tf(1)$ time by \Cref{lemma:acyclic multigraphs}, so assume $G_M$ is a cycle.

    Let $R_i$ be a rectangle minimizing $s=|P\cap R_i|$; since $\Lambda$ is light, $s\le\tau$.
    For each of the $s$ points $p\in P\cap R_i$, fixing the $i$th variable to $p$ yields a CSP instance $M_p$ whose primal graph is acyclic, as removing a vertex from a cycle leaves a path.
    Then $M$ is satisfiable if and only if some $M_p$ is, and each $M_p$ is solved in $\tf(1)$ time by \Cref{lemma:acyclic multigraphs}, for a total of $\tf(s)\le\tf(\tau)$.
\end{proof}

We move on to preprocessing heavy queries.
Instead of preprocessing all possible heavy query, we use a range tree decomposition of the plane $[n]\times[n]$ into a family $\RR$ of $\tf(n)$ rectangles, with the property that every rectangle $Q$ has a partition into $O(\log^2 n)$ rectangles from $\RR$, where such partition can be computed in $\tf(1)$ time.
Then, we reduce each heavy query $\Lambda=(\pi,\sigma, (R_i)_{i \in [k]})$ to $\tf(1)$ queries $\Lambda'=(\pi,\sigma, (R'_i)_{i \in [k]})$ where each $R'_i$ is a rectangle from $\RR$.
The family of rectangles $\RR$ is obtained from a standard two-dimensional range tree decomposition of the plane, as explained in \Cref{sec:preliminaries}.
It has three important properties:
\begin{enumerate}
    \item $|\RR|=O(n\log^2 n)$.
    \item Every axis-parallel rectangle $Q$ can be partitioned into $O(\log^2 n)$ rectangles from $\RR$, denoted by $\CR(Q)$. Furthermore, $\CR(Q)$ can be found in $\tf(1)$ time.
    \item For every $R\in \RR$, let $|R|$ denote the number of points in $P$ that fall in $R$. For every $x\ge 1$, there are at most $O(n\log^2 n/x)$ rectangles in $\RR$ with $|R|\ge x$.
\end{enumerate}

We first prove the following restricted version, and then reduce the general problem to it.
\begin{claim}\label{claim2:heavy}
    Every $\tau$-heavy query $\Lambda=(\pi,\sigma, (R_i)_{i \in [c]})$, with  $(i)$ $R_i$ is a rectangle from $\RR$ for every $i\in[c]$, and $(ii)$
    the primal graph $G_\Lambda$ is a subgraph of a $c$-cycle (with edge multiplicities), can be solved in $\tf(1)$ time, after $\tf(\tau (n/\tau)^c)$ preprocessing time. The preprocessing step takes $\pi$, $\RR$, and $c$ as input.
\end{claim}

\begin{proof}[Proof of \Cref{claim2:heavy}]
    \newcommand{\Rhp}{\RR_H^\prime}
    Let $\RR_H$ denote the subset of $\RR$ that contains only rectangles with at least $\tau$ points from $P$.
    There are at most $O(n\log^2 n/\tau)$ rectangles in $\RR_H$, by the properties of range trees.
    We also use $S_H$ to denote the set of all rectangles in $\RR_H$ that contains at most $4\tau$ points from $P$.
    Clearly, $|S_H|\leq |\RR_H|$, as $S_H\subseteq \RR_H$.
    We construct a table $B$ that stores the answers to all queries $\Lambda=(\pi,\sigma, (R_i)_{i \in [c]})$ where each $R_i$ is in $\RR_H$.

    To construct the table $B$, we first fill the entries that correspond to queries $X=(R_1,\ldots,R_k)$ where each $R_j$ is in $S_H$.
    There are $\tf((n/\tau)^k)$ such entries to fill, and each can be filled in time $\tf(\tau)$ by \Cref{claim2:light}, since each entry corresponds to a $4\tau$-light query.

    Now, we explain how to fill the rest of the table $B$ using only the entries that we have already filled.
    Let $X=(R_1,\ldots,R_k)$ and $Y=(R_1',\ldots,R_k')$ be two queries.
    We say that $Y$ is an \emph{extension} of $X$ if for every $j\in[k]$, the node $u_j'$ in the range tree corresponding to $R_j'$ is an ancestor of the node $u_j$ corresponding to $R_j$.
    This implies that $R_j\subseteq R_j'$ for every $j\in[k]$.
    Note that if $Y$ is an extension of $X$, then, $\ind_\sigma(Y)\geq \ind_\sigma(X)$ by monotonicity, since any witness for $X$ is also a witness for $Y$.
    To fill the entire table, we perform a bottom-up dynamic programming.
    After computing $B[X]$ for $X=(R_1,\ldots,R_k)$ with rectangles in $S_H$,
    we consider all extensions $Y$ of $X$ and set $B[Y]=\max(B[Y],B[X])$.

    Finding all extensions of $X$ can be done efficiently using the structure of the range tree; namely, by traversing the paths from the nodes corresponding to the rectangles in $X$ to the root of the range tree, and looking for all tuples of nodes on these paths that correspond to queries $Y$.
    In other words, for every $X=(R_1,\ldots,R_k)$ with rectangles in $S_H$, we have to update $\tf(1)$ entries in $B$, and therefore the total time for this step is $\tf(\tau(n/\tau)^k)$, which is the same as the time to fill the entries corresponding to queries with rectangles in $S_H$.
\end{proof}

Since every query is either $\tau$-light or $\tau$-heavy, we can combine \Cref{claim2:light} and \Cref{claim2:heavy} to obtain:
\begin{conclusion}\label{conc:tau queries}
    There is a data structure that takes as input a permutation $\pi$ of length $n$ and a parameter $c$.
    After $\tf(n^{2-1/c})$ preprocessing, the data structure can answer \rppm queries $\Lambda = (\pi,\sigma, (R_i)_{i \in [c]})$, in time $\tf(n^{1-1/c})$, assuming $(i)$ that $G_\Lambda$ is a subgraph of a $c$-cycle (with possible edge multiplicities), and $(ii)$ that each $R_i$ is a rectangle from $\RR$.
\end{conclusion}
Note that \Cref{conc:tau queries} is a restricted version of \Cref{thm3:csp Gm}, since it assumes that each $R_i$ is a rectangle from $\RR$.
\begin{proof}[Proof of \Cref{conc:tau queries}]
    We set $\tau=n^{1-1/c}$.
    By \Cref{claim2:light}, every $\tau$-light query can be answered in $\tf(\tau)=\tf(n^{1-1/c})$ time.
    By \Cref{claim2:heavy}, every $\tau$-heavy query can be answered in $\tf(1)$ time after $\tf(\tau(n/\tau)^c)=\tf(n^{2-1/c})$ preprocessing time.
\end{proof}
We now show how to remove the restriction $(ii)$ in \Cref{conc:tau queries}, and thus prove \Cref{thm3:csp Gm}:

\begin{proof}[Proof of \Cref{thm3:csp Gm}]
    Let $\Lambda=(\pi,\sigma, (R_i)_{i \in [c]})$ be a \rppm query, where $G_\Lambda$ is a subgraph of a $c$-cycle (with edge multiplicities).
    For each $j\in[c]$, we find a partition $\CR(R_j)$ of $R_j$ into $O(\log^2 n)$ rectangles from $\RR$ in $\tf(1)$ time.
    Define
    \begin{align*}
        \JJ=\CR(R_1)\times\cdots\times \CR(R_c)\;.
    \end{align*}
    Each element of $\JJ$ is a tuple $(R_1',\ldots,R_c')$ with $R_j'\in \CR(R_j)$ for every $j\in[c]$.
    Since every $R_j'$ lies in $\RR$, each $(R_1',\ldots,R_c')$ is a restricted \rppm query.
    Moreover,
    \begin{align*}
        \abs{\JJ}=\prod_{j=1}^c \abs{\CR(R_j)}=O(\log^{2c} n)=\tf(1)\;.
    \end{align*}
    The query algorithm evaluates $\ind_\sigma((R_1',\ldots,R_c'))$ for every $(R_1',\ldots,R_c')\in \JJ$ using \Cref{conc:tau queries}, and returns their disjunction.

    ~\textbf{Correctness.}
    Let $X=(R_1,\ldots,R_c)$ be the original query.
    We claim that
    \[
        \ind_\sigma(X)=\bigvee_{Y\in \JJ}\ind_\sigma(Y).
    \]
    Recall that $\ind_\sigma(X)=1$ if there exist indices $t_1<\cdots<t_c$ such that the sequence $A[t_1],\ldots,A[t_c]$ is order-isomorphic to $\sigma$, and for every $j\in[c]$ the point $(t_j,A[t_j])$ is in the rectangle $R_j$.

    $(\Longrightarrow)$ Suppose that $\ind_\sigma(X)=1$.
    Then there exist indices $t_1<\cdots<t_c$ such that $(t_j,\pi(t_j))\in R_j$ for every $j\in[c]$, and the sequence $\pi(t_1),\ldots,\pi(t_c)$ is order-isomorphic to $\sigma$.
    Since $\CR(R_j)$ is a partition of $R_j$, for each $j$ there is a unique rectangle $R_j'\in \CR(R_j)$ containing $(t_j,\pi(t_j))$.
    Hence the tuple $Y=(R_1',\ldots,R_c')$ belongs to $\JJ$, and $\ind_\sigma(Y)=1$, by the same witness $t_1<\cdots<t_c$.
    Therefore the right-hand side is $1$.

    $(\Longleftarrow)$ Suppose that $\ind_\sigma(Y)=1$ for some $Y=(R_1',\ldots,R_c')\in \JJ$.
    Since $R_j'\subseteq R_j$ for every $j\in[c]$, any witness for $Y$ is also a witness for $X$. Thus $\ind_\sigma(X)=1$.

    ~\textbf{Running Time.}
    By \Cref{conc:tau queries}, each restricted query is answered in $\tf(n^{1-1/c})$ time after $\tf(n^{2-1/c})$ preprocessing.
    Since $\abs{\JJ}=\tf(1)$, the total query time remains $\tf(n^{1-1/c})$.
    The preprocessing is exactly that of the restricted data structure, namely $\tf(n^{2-1/c})$.
    This proves \Cref{thm3:csp Gm}.
\end{proof}

\section{Approximate Counting 6-Patterns}
\label{sec:approx counting6}
This section is devoted to proving \Cref{thm:approx counting6}.
\thmApproxCountingSix* %

To prove it, we first explain how to solve special cases of \rppm instances based on their primal graph.

\subsection{Global Separators}
\label{ssec:global-separators}

In this subsection we show how to reduce a CSP instance~$M$ to $\tf(n)$ simpler
sub-instances, each of which has two fewer edges in its constraint graph
(one $P_x$-edge and one $P_y$-edge removed).
The technique is adapted from~\cite[Section~6.1]{ben2024approximate}.

\medskip

Let $\Lambda=(\pi,\sigma,(R_i)_{i\in[k]})$ be an \rppm instance with
corresponding CSP instance $M=(V,\mathcal{D},C)$ and constraint graph~$G_M$.
Recall that $G_M$ is a subgraph of the full constraint graph~$G_\sigma$,
whose edge set is the union of the two Hamiltonian paths
\[
    P_x = (1,2,\ldots,k)
    \qquad\text{and}\qquad
    P_y = (\sigma^{-1}(1),\sigma^{-1}(2),\ldots,\sigma^{-1}(k)).
\]
We call the edges of~$P_x$ the \emph{$x$-edges} and those of~$P_y$ the \emph{$y$-edges}.
To state the main lemma of this section, we need to define the sub-instances we obtain from~$M$ by removing one $x$-edge and one $y$-edge.
For that we need to use the structure of the two-dimensional range tree built on the point set $P=\{(t,\pi(t)):t\in[n]\}$, which we recall next.
\paragraph{Range Tree Recap.}
The reduction uses the two-dimensional range tree~$\TT$ built on the point set $P=\{(t,\pi(t)):t\in[n]\}$, or simply one range tree for every color class $
(P_1,\ldots, P_k)$.
We briefly recall its structure; full details appear in \Cref{ssec:range-trees}.
The range tree has two levels.
The first level is a balanced binary tree on the $x$-coordinates, called the \emph{outer tree}.
Each outer node~$v$ is associated with a vertical slab
$R_v=[a_v,b_v]\times[1,n]$ and contains the subset $P_v=P\cap R_v$.
The root corresponds to $[1,n]\times[1,n]$; the two children of a node with
slab $[a,b]\times[1,n]$ correspond to
$[a,c]\times[1,n]$ and $[c+1,b]\times[1,n]$, where $c$ is the midpoint of $[a,b]$. We stop splitting a slab into two rectangles, as soon as it contains a single point of~$P$, and therefore each leaf of the outer tree contains a single point of~$P$.

The second level augments each outer node~$v$ with an \emph{inner tree}~$T_v$,
which is a balanced binary tree on the $y$-coordinates of~$P_v$.
Each inner node~$w$ of~$T_v$ is associated with a rectangle
$R_w=[a_v,b_v]\times[c_w,d_w]$ and contains $P_w=P_v\cap R_w$.
The root of~$T_v$ corresponds to $R_v$ itself, and each leaf contains a single point of~$P_v$.
In the inner tree~$T_v$, the two children of a node with rectangle $[a_v,b_v]\times[c,d]$ correspond to $[a_v,b_v]\times[c,e]$ and $[a_v,b_v]\times[e+1,d]$, where $e$ is not the midpoint of $[c,d]$, but rather the median of the $y$-coordinates of the points in $P_v\cap ([a_v,b_v]\times[c,d])$, so that the two children contain roughly the same number of points of~$P_v$.

~\\We are now ready to define the set of sub-instances.
Every sub-instance is associated with four parameters $(v,w,h_x,h_y)$ where $v$ is an outer node of~$\TT$, $w$ is an inner node in $T_v$ the inner tree of~$v$, and $h_x,h_y$ are integers in $\set{1,2,\ldots,k-1}$.
Therefore, the number of sub-instances is $\tf(nk^2)$ as we soon prove.

\paragraph{The Instance $M_{v,w,h_x,h_y}$.}
Let $R_v=[a_v,b_v]\times[1,n]$ and $R_w=[a_v,b_v]\times [c_w,d_w]$ be the rectangles associated with $v$ and $w$, respectively.
Suppose the two children of~$v$ are $v_L,v_R$ with rectangles $[a_v,e_v]\times[1,n]$ and $[e_v+1,b_v]\times[1,n]$, respectively, and the two children of~$w$ are $w_B,w_T$ with rectangles $[a_v,b_v]\times[c_w,f_w]$ and $[a_v,b_v]\times[f_w+1,d_w]$, respectively.
For $i\in[k]$:
\begin{itemize}
    \item Narrow $D_i$ to $\set{t\in D_i\mid (t,\pi(t))\in R_w}$.
    \item if $i\leq h_x$, narrow $D_i$ to $D_i\cap \set{t\in [n]\mid t\le e_v}$.
    \item if $i>h_x$, narrow $D_i$ to $D_i\cap \set{t\in [n]\mid t> e_v}$.
    \item if $\sigma(i)\leq h_y$, narrow $D_i$ to $D_i\cap \set{t\in [n]\mid \pi(t)\le f_w}$.
    \item if $\sigma(i)>h_y$, narrow $D_i$ to $D_i\cap \set{t\in [n]\mid \pi(t)> f_w}$.
\end{itemize}
Note that as a byproduct of the above narrowing, the $x$-edge $(h_x,h_x+1)$ and the $y$-edge $(\sigma^{-1}(h_y),\sigma^{-1}(h_y+1))$ become redundant as the domain of the $h_x$th and $(h_x+1)$th variables are now separated by the vertical line $x=e_v$, and the domain of the $\sigma^{-1}(h_y)$th and $(\sigma^{-1}(h_y+1))$th variables are now separated by the horizontal line $y=f_w$
Therefore, the edges $(h_x,h_x+1)$ and $(\sigma^{-1}(h_y),\sigma^{-1}(h_y+1))$ can be removed from the constraint graph.
We emphasize that the instance $M_{v,w,h_x,h_y}$ is obtained from $M$ by narrowing the domains of the variables.
Therefore, any satisfying assignment to $M_{v,w,h_x,h_y}$ is also a satisfying assignment to $M$.

~\\We are ready to state the main lemma of this section showing that $M$ is satisfiable if and only if at least one of the sub-instances $M_{v,w,h_x,h_y}$ is satisfiable.
\begin{lemma}[Global Separators]\label{lem:global-sep}
    Let $M=(V,\mathcal{D},C)$ be a CSP instance as above.
    There is a family of $\tf(nk^2)$ sub-instances
    $\{M_{v,w,h_x,h_y}\}$ as above.
    It holds that $M$ is satisfiable if and only if at least one of the sub-instances $M_{v,w,h_x,h_y}$ is satisfiable.
\end{lemma}
Clearly, any $T$ time algorithm for solving all sub-instances $M_{v,w,h_x,h_y}$ implies a $T$ time algorithm for solving $M$.

\begin{proof}[Proof of \Cref{lem:global-sep}]
    The proof proceeds in two parts.
    We first show that if a sub-instance is satisfiable then so is~$M$ (\emph{soundness}), and then show that if~$M$ is satisfiable then at least one sub-instance is satisfiable (\emph{completeness}).

    \paragraph{Soundness.}
    Every sub-instance $M_{v,w,h_x,h_y}$ is obtained from~$M$ by narrowing domains. Therefore any satisfying assignment to $M_{v,w,h_x,h_y}$ is also a satisfying assignment to~$M$.

    \paragraph{Completeness.}
    Suppose $M$ is satisfiable, and let $\phi$ be a satisfying assignment.
    We show that there exists a choice of $(v,w,h_x,h_y)$ such that $\phi$ satisfies $M_{v,w,h_x,h_y}$.
    We explain how given $\phi$ we can find such a choice of parameters, and then verify that $\phi$ indeed satisfies the corresponding sub-instance.

    First, we find $v$ and $h_x$.
    For each $i\in[k]$, let $u_i$ be the leaf of the outer tree containing the point $(\phi(i),\pi(\phi(i)))$.
    Let $v$ be the least common ancestor~(LCA) of $u_1,\ldots,u_k$ in the outer tree, and let $v_L,v_R$ be the two children of~$v$ with rectangles $[a_v,e_v]\times[1,n]$ and $[e_v+1,b_v]\times[1,n]$, respectively.
    Since $v$ is the LCA, not all leaves $u_1,\ldots,u_k$ lie in the same child's subtree.
    Let $h_x\in[k-1]$ be the unique index such that exactly $h_x$ of the leaves lie in the subtree of~$v_L$ and the remaining $k-h_x$ leaves lie in the subtree of~$v_R$.
    Since $\phi$ satisfies the $x$-order constraints, i.e., $\phi(1)<\phi(2)<\cdots<\phi(k)$, it must be the case that the leaves $u_1,\ldots,u_{h_x}$ lie in the subtree of~$v_L$ and $u_{h_x+1},\ldots,u_k$ lie in the subtree of~$v_R$, in other words, $\phi(i)\le e_v$ for $i\le h_x$ and $\phi(i)>e_v$ for $i>h_x$.

    Next, we find $w$ and $h_y$.
    All points $(\phi(i),\pi(\phi(i)))$ lie in~$R_v$, so they all appear in the inner tree~$T_v$ of~$v$.
    For each $i\in[k]$, let $\hat{u}_i$ be the leaf of~$T_v$ containing $(\phi(i),\pi(\phi(i)))$.
    Let $w$ be the LCA of $\hat{u}_1,\ldots,\hat{u}_k$ in~$T_v$, and let $w_B,w_T$ be the two children of~$w$ with rectangles $[a_v,b_v]\times[c_w,f_w]$ and $[a_v,b_v]\times[f_w+1,d_w]$, respectively.
    Since $w$ is the LCA, not all leaves $\hat{u}_1,\ldots,\hat{u}_k$ lie in the same child's subtree.
    Let $h_y\in[k-1]$ be the unique index such that exactly $h_y$ of the leaves lie in the subtree of~$w_B$ and the remaining $k-h_y$ leaves lie in the subtree of~$w_T$.
    Since $\phi$ satisfies the $y$-order constraints, i.e., $\pi(\phi(\si(1)))<\cdots<\pi(\phi(\si(k)))$, it must be the case that $\hat{u}_{r}$ lies in the subtree of~$w_B$ for $\sigma(r)\le h_y$ and in the subtree of~$w_T$ for $\sigma(r)>h_y$, in other words, $\pi(\phi(i))\le f_w$ when $\sigma(i)\le h_y$ and $\pi(\phi(i))>f_w$ when $\sigma(i)>h_y$.

    \medskip
    We now verify that $\phi$ satisfies $M_{v,w,h_x,h_y}$.
    Since all points $(\phi(i),\pi(\phi(i)))$ lie in~$R_w\subseteq R_v$, the narrowing $D_i\cap R_w$ still contains $\phi(i)$.
    The $x$-coordinate narrowings are respected since $\phi(i)\le e_v$ for $i\le h_x$ and $\phi(i)>e_v$ for $i>h_x$.
    The $y$-coordinate narrowings are respected since $\pi(\phi(i))\le f_w$ when $\sigma(i)\le h_y$ and $\pi(\phi(i))>f_w$ when $\sigma(i)>h_y$.
    The completeness follows.

    \paragraph{Number of sub-instances.}
    We count how many sub-instances $M_{v,w,h_x,h_y}$ we have, i.e., how many choices of parameters $(v,w,h_x,h_y)$ are there.
    There are $(k-1)^2$ choices for $h_x$ and $h_y$.
    We show that there are $\tf(n)$ choices for $v$ and $w$, resulting in $\tf(nk^2)$ sub-instances.

    The outer tree has $O(n)$ nodes, so there are $O(n)$ choices for $v$.
    For each outer node~$v$, the inner tree~$T_v$ has $|T_v|$ choices for $w$, where $|T_v|$ is the number of nodes in the inner tree~$T_v$.
    Note that $T_v$ has $|P_v|$ leaves, and therefore at most $|P_v|\log(n)$ internal nodes.
    On the other hand, each point $(t,\pi(t))$ is contained in $O(\log n)$ outer nodes of~$T_v$, since the depth of the outer tree is $O(\log n)$ and each point is contained in the slab of exactly one node at each level of the outer tree.
    We therefore get
    \begin{align*}
        \#(v,w) \text{ pairs}
        =\sum_v |T_v|
        \leq \sum_v |P_v|\log(n)
        \leq \log n\cdot \sum_v |P_v|
        \leq n\log^2 n\;.
    \end{align*}
    Therefore, the number of sub-instances is $O(n\log^2 n k^2)$ as claimed.
\end{proof}
We now analyze the total time needed to solve all sub-instances $M_{v,w,h_x,h_y}$.
Recall that each sub-instance $M_{v,w,h_x,h_y}$ is associated with an inner node~$w$ of the inner tree~$T_v$, and that $P_w$ denotes the set of points in the rectangle~$R_w$.
The following claim is the main application of \Cref{lem:global-sep}: it shows that if each sub-instance can be solved in time depending on~$|P_w|$, then $M$ can be solved efficiently.

\begin{claim}\label{claim:solving subinstances}
    Suppose that there is an algorithm to solve any sub-instance $M_{v,w,h_x,h_y}$ in $O(|P_w|^q)$ time for some $q\geq 1$. Then, $M$ can be solved in $\tf(n^q)$ time.
\end{claim}

We use the following standard fact about balanced binary trees.
\begin{fact}\label{fact:subtree-sum}
    Let $T$ be a balanced binary tree with root~$r$ and depth $O(\log |L_r|)$, where for every node~$w$ we write $L_w$ for the set of leaves in the subtree rooted at~$w$.
    Then, for every $q \geq 1$,
    \[
        \sum_{w \in T} |L_w|^q \leq |L_r|^q \cdot O(\log |L_r|)\;.
    \]
\end{fact}
\begin{proof}
    Since $q \geq 1$, the function $x \mapsto x^q$ is convex, so for any $a,b\geq 0$ we have $a^q + b^q \leq (a+b)^q$.
    It follows by induction that at each depth~$d$ of the tree, $\sum_{w \text{ at depth } d} |L_w|^q \leq |L_r|^q$.
    Summing over $O(\log |L_r|)$ depths gives the result.
\end{proof}

\begin{proof}[Proof of \Cref{claim:solving subinstances}]
    By \Cref{lem:global-sep}, it suffices to solve all sub-instances $M_{v,w,h_x,h_y}$ and check whether at least one is satisfiable.
    The total solving time is
    \[
        \sum_{(v,w,h_x,h_y)} O(|P_w|^q)
        \;\leq\; k^2 \cdot \sum_{v\in \TT}\;\sum_{w\in T_v} |P_w|^q\;,
    \]
    where the factor of~$k^2$ accounts for the $(k-1)^2$ choices of $(h_x,h_y)$.
    We bound the sum:
    \begin{align*}
        \sum_{v\in \TT}\sum_{w\in T_v} |P_w|^q
        \leq \sum_{v\in \TT} |P_v|^q \cdot \log n
        \leq \log n \cdot (\log n \cdot n^q)
        = O(n^q \log^2 n)
        = \tf(n^q)\;.
    \end{align*}
    The first inequality follows by applying \Cref{fact:subtree-sum} to the inner tree $T_v$ (whose root is $v$ and has $|P_v|$ points and leaves).
    The second inequality follows by applying \Cref{fact:subtree-sum} to the outer tree (whose root has $n$ points).
\end{proof}

\subsection{Putting Everything Together}
\label{ssec:computer assisted}
In this subsection we show how to approximate the number of embeddings of any \pppm instance with pattern $\sigma$ of length $6$:
\begin{theorem}
    \label{thm:partitioned6}
    There is a $\tf(n^{7/4})$ time algorithm for the \pppm problem for every pattern $\sigma$ of length $6$.
\end{theorem}
By combining \Cref{thm:partitioned6} with the black box reduction of \cite{DellL21,DellLM22}, we get \Cref{thm:approx counting6}.
Before proving \Cref{thm:partitioned6}, we need one additional statement on the primal graph of the CSP instances obtained after applying the global separators.
\begin{restatable}[Nice Graph]{definition}{DefNiceGraph}
    \label{def:nice graph}
    Let $\Lambda=(\pi,\sigma,(P_i)_{i\in[k]},(R_i)_{i\in[k]})$ be a \rppm instance with $k\leq 5$, and let $M=(V,\mathcal{D},C)$ be the corresponding CSP instance with primal graph $G_M$.
    Let $G_M'$ be the graph obtained from $G_M$ by repeatedly removing any vertex of degree at most one, where degrees count edge multiplicities.
    We say that $G_M$ is \emph{nice} if the underlying simple graph of $G_M'$ is a forest, or if the variables of $G_M'$ can be partitioned into two sets $V_1$ and $V_2$, with the following properties.
    \begin{enumerate*}[label=(\roman*)]
        \item $|V_1|\geq 3 > |V_2|$,
        \item The cut between $V_1$ and $V_2$ has size at most $1$ counting edge multiplicities,
        \item The underlying simple graph of $G_M'[V_1]$ is a subgraph of a $c$-cycle for some $c\leq 4$.
    \end{enumerate*}
\end{restatable}

\begin{restatable}[Almost Nice Graph]{definition}{DefAlmostNiceGraph}
    \label{def:almost nice graph}
    Let $\Lambda=(\pi,\sigma,(P_i)_{i\in[6]},(R_i)_{i\in[6]})$ be a \rppm instance, and let $M=(V,\mathcal{D},C)$ be the corresponding CSP instance with primal graph $G_M$.
    We say that $G_M$ is \emph{almost nice} if, after removing any $x$-edge and any $y$-edge from $G_M$, there exists a vertex $v\in V$ such that deleting $v$ and all its incident edges leaves a nice graph.
\end{restatable}
\begin{proposition}\label{claim:nice graph solve}
    There is a data structure that takes as input a permutation $\pi$ of length $n$ and the color classes $(P_i)_{i\in[k]}$, where $k\leq 5$.
    After $\tf(n^{7/4})$ preprocessing time, the data structure solves any \rppm instance $\Lambda=(\pi,\sigma,(P_i)_{i\in[k]},(R_i)_{i\in[k]})$ whose primal graph $G_M$ is nice, in $\tf(n^{3/4})$ time.
    In particular, any such instance can be solved in $\tf(n^{7/4})$ time.
\end{proposition}
\begin{proof}[Proof of \Cref{claim:nice graph solve}]
    \paragraph{Preprocessing.}
    We build the data structure of \Cref{thm3:csp Gm} with parameters $(c,k)$, once for $c\in\set{2,3,4}$ and $k=5$.
    For $c=2$, the data structure is the one of \Cref{lemma:acyclic multigraphs}, which solves any \rppm instance whose primal graph is a forest (with edge multiplicities) after $\tf(n)$ preprocessing time, and answers any query in $\tf(1)$ time.
    For $c\in\set{3,4}$, we use the data structure of \Cref{thm3:csp Gm}.
    This takes $\tf(n^{2-1/2}+n^{2-1/3}+n^{2-1/4})=\tf(n^{7/4})$ time for the preprocessing, and then we can answer any \rppm query whose underlying simple primal graph is a subgraph of a $c$-cycle in $\tf(n^{1-1/c})\leq \tf(n^{3/4})$ time.

    \paragraph{Answering a Query.}
    Let $\Lambda$ be a \rppm instance whose primal graph $G_M$ is nice, and let $M=(V,\mathcal{D},C)$ be the corresponding CSP instance.
    We can reduce $M$ into a new instance $M'$, where every variable of degree at most one is fixed, while having the same satisfiability as $M$.
    that is, $M$ is satisfiable if and only if $M'$ is satisfiable.
    If the underlying simple graph of $G_{M'}$ is a forest, then \Cref{lemma:acyclic multigraphs} solves $M'$ in $\tf(1)$ time.
    Otherwise, since $G_M$ is nice, the primal graph $G_{M'}$ can be partitioned into two sets $V_1$ and $V_2$, with the above properties.
    Let $M_1'$ be the sub-instance of $M'$ induced by $V_1$, and let $M_2'$ be the sub-instance of $M'$ induced by $V_2$.
    We have an oracle $O_1$ for $M_1'$ that answers any rectangular-sub-instance of $M_1'$ in $\tf(n^{3/4})$ time, and an oracle $O_2$ for $M_2'$ that answers any rectangular-sub-instance of $M_2'$ in $\tf(1)$ time.
    We can then apply \Cref{lemma:P3} to combine the two oracles and solve $M'$ in $\tf(n^{3/4})$ time, which is the desired running time.
\end{proof}

\begin{restatable}[Computer-Assisted]{claim}{claimCA}\label{claim:csp graph}
    Let $\sigma$ be a pattern of length $6$, and let $G_\sigma$ be its full primal graph. Then $G_\sigma$ is almost nice.
\end{restatable}

We prove \Cref{claim:csp graph} using a computer-assisted case analysis, which is described in full detail in the appendix. The case-analysis code can be found here: \href{https://github.com/TomerEven/PPM}{https://github.com/TomerEven/PPM}.

\begin{proof}[Proof of \Cref{thm:partitioned6}]
    The input is a permutation $\pi$ of length $n$, which defines the point set $P=\{(i,\pi(i))\mid i\in[n]\}$, a partition $P_1,\ldots,P_k$ of the point set $P$, and a pattern $\sigma$ of length $k=6$.

    \paragraph{Preprocessing.}
    We build the two-dimensional range tree $\TT$ used by the global separator construction.

    \paragraph{The Algorithm.}
    Let $M=(V,\mathcal{D},C)$ be the CSP instance corresponding to the \srppm problem we need to solve.
    The first step of the algorithm is to apply the global separators to get new CSP instances $M'=(V,\mathcal{D}',C')$ by removing pairs of constraints from $M$.
    Each new instance $M'$ is encoded using a quadruple $(v,w,h_x,h_y)$, where
    $v$ is an outer vertex in the range tree $\TT$, $w$ is an inner vertex in the tree $T_v$ of rooted by $v$, and $h_x,h_y\in\set{1,2,\ldots, k-1}$.

    Recall that by \Cref{lem:global-sep}, $M$ is satisfiable if and only if at least one of the instances $M'$ is satisfiable.
    By \Cref{claim:solving subinstances}, it remains to solve each instance $M_{v,w,h_x,h_y}$ in $\tf(|P_w|^q)$ time, for some $q\geq 1$.

    \paragraph{Solving Each Instance $M'=M_{v,w,h_x,h_y}$.}
    Put $m=|P_w|$.
    By \Cref{claim:csp graph}, there is a vertex $u$ such that deleting $u$ from $G_{M'}$ leaves a nice graph.
    We build the data structure of \Cref{claim:nice graph solve} on $P_w$ and the remaining five color classes, in $\tf(m^{7/4})$ time.
    Then we iterate over all $t\in D_u$, fix $u=t$, narrow the domains of its neighbors according to the incident constraints, and delete $u$.
    Each residual instance is nice, and is solved in $\tf(m^{3/4})$ time.
    Since $|D_u|\leq m$, this solves $M'$ in $\tf(m^{7/4})+\tf(m\cdot m^{3/4})=\tf(m^{7/4})$ time.
    Applying \Cref{claim:solving subinstances} with $q=7/4$ completes the proof.
\end{proof}

\begin{remark}
    Obtaining an $\tf(n^p)$ preprocessing time and $\tf(n^q)$ query time for solving \rppm instances for patterns of length $5$ would result in an $\tf(n^{p}+n^{q+1})$ time algorithm for \pppm for patterns of length $7$, as well as approximate counting for patterns of length $7$, by the same reduction of \cite{DellL21,DellLM22}.
\end{remark}

\phantomsection
\addcontentsline{toc}{section}{Acknowledgments}
\section*{Acknowledgments}
Omri Ben-Eliezer is supported by an Alon Scholarship and by a Taub Family Foundation ``Leaders in Science \& Technology'' Fellowship.

\phantomsection
\addcontentsline{toc}{section}{AI/LLM Disclosure}
\section*{AI/LLM Disclosure}
The authors made light use of several LLMs during the course of the project, including the latest versions of ChatGPT, Claude, and Gemini. LLMs were used for $(i)$ minor editorial purposes, such as polishing the text where needed; $(ii)$ literature review including looking up known reductions and techniques from the literature; and $(iii)$ spotting minor inaccuracies in the written draft and suggesting corrections. The authors proved all of the main results in the paper, and take full responsibility for all content.

\bibliographystyle{alpha}
\bibliography{refs}

\appendix

\section{Missing Proofs}
In this section we provide the missing proofs for several claims from the main body.

\subsection{Proof of \Cref{claim:17}}
We restate the claim here for convenience:
\ClaimSev* %

Recall that a triplet of rectangles $(R_1,R_2,R_3)$ is called well-aligned if for every pair of rectangles, either they share the same $x$-projection or they have disjoint $x$-projections, and the same holds for the $y$-projections.
The proof shows that three well-aligned rectangles, that do not all share the same $x$-projection or the same $y$-projection, can be solved in $\tf(1)$ time after $\tf(n)$ preprocessing time.
We prove this by showing that the primal graph of the CSP instance defined by the triplet is acyclic, and therefore we can solve it in $\tf(1)$ time by \Cref{lemma:acyclic multigraphs}.
\begin{proof}
    Consider the CSP $M=(V,\mathcal{D},C)$ defined by the triplet of rectangles $(R_1,R_2,R_3)$, where $V=\{v_1,v_2,v_3\}$, $\mathcal{D}=(R_1\cap P,R_2\cap P,R_3\cap P)$, and $C$ contains the constraints $v_1<v_2<v_3$ and $\pi(v_1)<\pi(v_3)<\pi(v_2)$.
    Let $G_M$ be the primal graph of $M$, and $G_{\sigma}$ be the full primal graph of $\sigma=132$.
    The primal graph $G_\sigma$ is a triangle with the edge $(v_2,v_3)$ having multiplicity two.
    Since the rectangles are well-aligned, some of the constraints might hold trivially, and therefore $G_M$ might be a proper subgraph of $G_\sigma$.
    Recall that if $G_M$ is acyclic, then we can solve $M$ in $\tf(1)$ time by \Cref{lemma:acyclic multigraphs}.

    Note that if the edges $(v_1,v_2)$ or $(v_1,v_3)$ are not in $G_M$, then $G_M$ is acyclic.
    \begin{itemize}
        \item
              The edge $(v_1,v_2)$, matching the constraint $v_1<v_2$, is in $G$ if and only if the $x$-projections of $R_1$ and $R_2$ are not disjoint.
        \item
              The edge $(v_1,v_3)$, matching the constraint $\pi(v_1)<\pi(v_3)$, is in $G$ if and only if the $y$-projections of $R_1$ and $R_3$ are not disjoint.
    \end{itemize}
    But if both of these constraints hold, then $R_2$ and $R_3$ have disjoint $x$-projections and disjoint $y$-projections, and therefore the edge $(v_2,v_3)$ is not in $G_M$, and $G_M$ is acyclic.
\end{proof}

\subsection{Proof of \Cref{lemma:monotone patterns} (Monotone Patterns)}
Again, we restate the claim for convenience.

\LemmaMono*

\begin{proof}
    We consider $\sigma=12\cdots k$; the decreasing case is symmetric.
    We prove, by induction on $k$, that after a shared preprocessing of the point set there is a data structure answering arbitrary multi-rectangle queries $\ind_{12\cdots k}(R_1,\ldots,R_k)$ with the above preprocessing and query times.

    \emph{Preprocessing.}
    Build the two-dimensional range tree of \Cref{ssec:range-trees} with its canonical family $\RR$; for every $R\in\RR$ precompute $\ell(R)=\min(k,\operatorname{LIS}(P\cap R))$, the length of a longest increasing subsequence inside $R$ capped at $k$, in $O(|P\cap R|\log n)$ time per rectangle, for a total of $\tf(n)$ time; and recursively build the data structures of the induction hypothesis for all lengths $k'<k$ (the same shared trees are reused, so the preprocessing time sums to $2^{O(k^2)}O(n\log^{3k} n)$).

    \emph{Query.}
    Given $(R_1,\ldots,R_k)$, first apply the well-alignment refinement of \Cref{ssec:range-trees} and then decompose each rectangle into canonical rectangles, reducing the query, at a multiplicative cost of $O(\log^{2k}n)$, to queries $(R_1',\ldots,R_k')$ in which the rectangles are canonical and pairwise well-aligned.
    Fix such a tuple.
    If $R_1'=\cdots=R_k'=R$, the answer is the precomputed bit $[\ell(R)\ge k]$.
    Otherwise, let $i$ be the first index with $R_i'\neq R_{i+1}'$.
    By well-alignment, the $x$-projections of $R_i'$ and $R_{i+1}'$ are identical or disjoint, and likewise the $y$-projections, and not both are identical.
    If some projection is disjoint in the \emph{wrong} order (the projection of $R_{i+1}'$ precedes that of $R_i'$), the corresponding order constraint is infeasible and the answer is~$0$.
    Otherwise a disjoint projection in the right order makes one of the two order constraints between $v_i$ and $v_{i+1}$ automatically satisfied, so the primal graph (a doubled path) loses one of the two parallel edges between $v_i$ and $v_{i+1}$ and acquires a cut of size one separating $\{v_1,\ldots,v_i\}$ from $\{v_{i+1},\ldots,v_k\}$.
    Apply \Cref{lemma:P3}: the two sides are monotone instances of lengths $i$ and $k-i$, and the required oracles are supplied by the induction hypothesis (every rectangular narrowing of a monotone instance is again a monotone instance).
    This costs $O(\log n)\cdot\big(Q(i)+Q(k-i)\big)$ per tuple.
    We thus get $Q(k)\le O(\log^{2k}n)\cdot O(\log n)\cdot 2\,Q(k-1) \leq (O(\log n))^{3k} Q(k-1)$. Unrolling the recursion gives $Q(k) = (O(\log n))^{3k^2}$.
\end{proof}

\subsection{Computer-assisted proof of \Cref{claim:csp graph}}
\sloppy{
In this subsection we describe the computer-assisted proof of \Cref{claim:csp graph}. The code is available at \href{https://github.com/TomerEven/PPM}{https://github.com/TomerEven/PPM}.}
We first restate the definitions of nice and almost nice graphs and the claim itself, then give a representative example showing how the algorithm uses global separators, vertex fixing, degree-one peeling, and the nice-graph solver, and finally describe the exhaustive verification.
The example illustrates the simplest outcome; the general search is more involved because a non-forest core may remain after the deletion and peeling steps.
\DefNiceGraph* %
\DefAlmostNiceGraph* %
\claimCA*

We first provide an example of the statement.
\paragraph{Example: The Pattern $\sigma=351624$.}
The full primal graph $G_\sigma$ of the pattern is drawn in \Cref{fig:351624_1}: the $x$-edges (red) form the path $P_x=(1,2,\ldots,6)$, connecting consecutive positions, and the $y$-edges (blue) form the path $P_y=(\si(1),\ldots,\si(6))=(3,5,1,6,2,4)$, connecting positions whose values are consecutive.
\begin{figure}[ht!]\centering
    \begin{tikzpicture}[scale=0.9,
            every node/.style={draw,circle,fill=white,font=\small},
            xedge/.style={blue},
            yedge/.style={red}
        ]
        \draw[step=1cm, very thin, gray!25] (0,0) grid (7,7);
        \draw[step=0.5cm, very thin, gray!15] (0,0) grid (7,7);
        \draw[thick, gray!50] (0,0) rectangle (7,7);

        \node (v1) at (1,3) {$1,3$};
        \node (v2) at (2,5) {$2,5$};
        \node (v3) at (3,1) {$3,1$};
        \node (v4) at (4,6) {$4,6$};
        \node (v5) at (5,2) {$5,2$};
        \node (v6) at (6,4) {$6,4$};

        \draw[red] (v1) -- (v2);
        \draw[red] (v2) -- (v3);
        \draw[red] (v3) -- (v4);
        \draw[red] (v4) -- (v5);
        \draw[red] (v5) -- (v6);

        \draw[blue] (v3) -- (v5);
        \draw[blue] (v1) -- (v5);
        \draw[blue] (v1) -- (v6);
        \draw[blue] (v2) -- (v6);
        \draw[blue] (v2) -- (v4);
    \end{tikzpicture}
    \caption{The full primal graph of $\sigma=351624$. Red edges are the $x$-edges, forming the path $P_x=(1,2,\ldots,6)$; blue edges are the $y$-edges, forming the path $P_y=(3,5,1,6,2,4)$.}\label{fig:351624_1}
\end{figure}
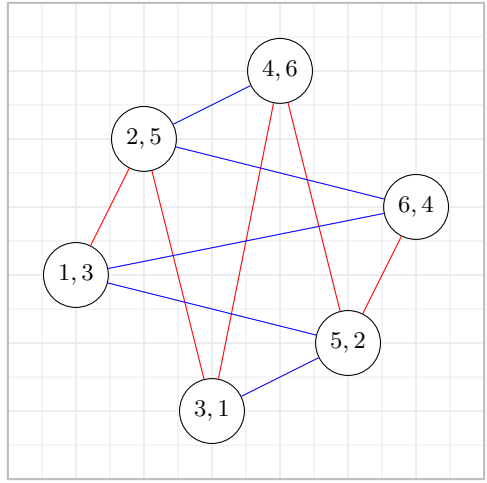
Recall that the global separators of \Cref{lem:global-sep} replace the input instance by $\tf(n)$ sub-instances.
In each sub-instance the domains are split by one vertical and one horizontal line, which removes one $x$-order constraint and one $y$-order constraint.
Suppose, for example, that the removed edges are the $x$-edge $(1,2)$ and the $y$-edge $(\si(5),\si(6))=(2,4)$, as in \Cref{fig:351624_2}.

\begin{figure}[ht!]\centering
    \begin{tikzpicture}[scale=0.9,
            every node/.style={draw,circle,fill=white,font=\small},
            xedge/.style={blue},
            yedge/.style={red}
        ]
        \draw[step=1cm, very thin, gray!25] (0,0) grid (7,7);
        \draw[step=0.5cm, very thin, gray!15] (0,0) grid (7,7);
        \draw[thick, gray!50] (0,0) rectangle (7,7);

        \node (v1) at (1,3) {$1,3$};
        \node (v2) at (2,5) {$2,5$};
        \node (v3) at (3,1) {$3,1$};
        \node (v4) at (4,6) {$4,6$};
        \node (v5) at (5,2) {$5,2$};
        \node (v6) at (6,4) {$6,4$};

        \draw[red, opacity=0.2] (v1) -- (v2);
        \draw[red] (v2) -- (v3);
        \draw[red] (v3) -- (v4);
        \draw[red] (v4) -- (v5);
        \draw[red] (v5) -- (v6);

        \draw[blue] (v3) -- (v5);
        \draw[blue] (v1) -- (v5);
        \draw[blue] (v1) -- (v6);
        \draw[blue] (v2) -- (v6);
        \draw[blue, opacity=0.2] (v2) -- (v4);

        \draw[red, dashed, opacity=0.5] (1.5,0) -- (1.5,7);
        \draw[blue, dashed, opacity=0.5] (0,5.5) -- (7,5.5);
    \end{tikzpicture}
    \caption{The primal graph after removing the $x$-edge $(1,2)$ and the $y$-edge $(\si(5),\si(6))=(2,4)$; the removed edges are faded. The dashed lines are the vertical and horizontal separating lines that make the two removed constraints redundant.}\label{fig:351624_2}
\end{figure}
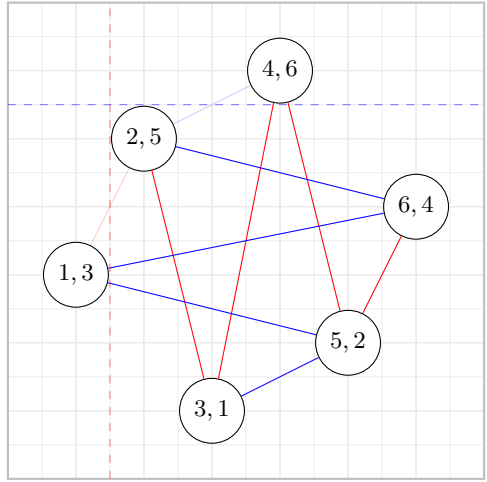

It remains to solve each sub-instance in $\tf(n^{7/4})$ time.
For this, we reduce the instance to several residual instances on five vertices, and solve each residual instance in $\tf(n^{3/4})$ time.
This is done by deleting a single vertex, or equivalently, by fixing the assignment of a single variable and solving the residual instance on the remaining variables.
Informally, we fix a vertex whose removal leaves a graph that our tools can solve.
The existence of such a vertex is exactly what we verify using the computer-assisted proof of \Cref{claim:csp graph}.
Here $v_5$ is such a vertex.
We iterate over the at most $n$ possible assignments of $v_5$; fixing $v_5$ turns every constraint incident to it into a narrowing of the domain of the other endpoint. Thus, each such choice results in a new instance with the primal graph shown in \Cref{fig:351624_3}.

For this example, the peeling step of \Cref{claim:peeling-degree-one} removes all remaining vertices.
Hence each residual instance is decided by \Cref{lemma:acyclic multigraphs}, and we continue over the assignments of $v_5$ until we either find a solution or exhaust all possibilities.

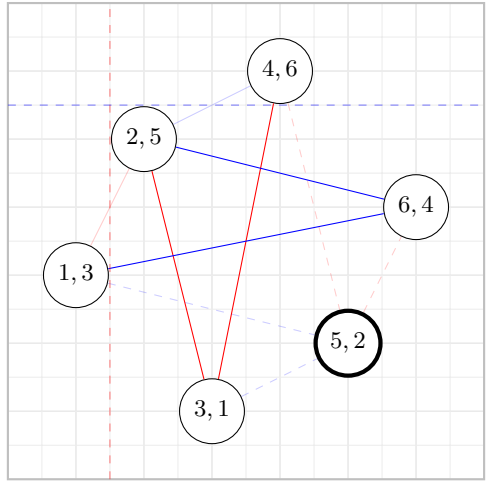
\begin{figure}\centering
    \begin{tikzpicture}[scale=0.9,
            every node/.style={draw,circle,fill=white,font=\small},
            nodeFix/.style={draw,circle,fill=white,font=\small,ultra thick},
            xedge/.style={blue},
            yedge/.style={red}
        ]
        \draw[step=1cm, very thin, gray!25] (0,0) grid (7,7);
        \draw[step=0.5cm, very thin, gray!15] (0,0) grid (7,7);
        \draw[thick, gray!50] (0,0) rectangle (7,7);

        \node (v1) at (1,3) {$1,3$};
        \node (v2) at (2,5) {$2,5$};
        \node (v3) at (3,1) {$3,1$};
        \node (v4) at (4,6) {$4,6$};
        \node[nodeFix] (v5) at (5,2) {$5,2$};
        \node (v6) at (6,4) {$6,4$};

        \draw[red, opacity=0.2] (v1) -- (v2);
        \draw[red] (v2) -- (v3);
        \draw[red] (v3) -- (v4);
        \draw[red, dashed, opacity=0.2] (v4) -- (v5);
        \draw[red, dashed, opacity=0.2] (v5) -- (v6);

        \draw[blue, dashed, opacity=0.2] (v3) -- (v5);
        \draw[blue, dashed, opacity=0.2] (v1) -- (v5);
        \draw[blue] (v1) -- (v6);
        \draw[blue] (v2) -- (v6);
        \draw[blue, opacity=0.2] (v2) -- (v4);

        \draw[red, dashed, opacity=0.5] (1.5,0) -- (1.5,7);
        \draw[blue, dashed, opacity=0.5] (0,5.5) -- (7,5.5);
    \end{tikzpicture}
    \caption{The primal graph after fixing the variable $v_5$ to a specific value in its domain; its incident edges (faded) become domain narrowings of its neighbors and are removed. The remaining graph is the path $(4,3,2,6,1)$.}\label{fig:351624_3}
\end{figure}

\paragraph{The General Verification.}
The example above is simple: after fixing one vertex and applying the peeling step, no cyclic core remains.
In the general case, more tools are used:
\begin{enumerate}
    \item Peeling degree-one vertices.
    \item Partitioning the variables into two sets with cut size at most one between then, where each part has primal graph that is a subgraph of a cycle of length at most four.
    \item Checking if the graph is acyclic (with edge multiplicities).
\end{enumerate}
Each one of these cases can be solved in the desired running time.

\end{document}